\documentclass[11pt]{article}
\usepackage[ruled]{algorithm2e} 
\usepackage{natbib} 

\usepackage{fullpage}

\usepackage{amssymb}
\usepackage{graphicx} 
\usepackage{subcaption}
\usepackage{latexsym, amsfonts, amsmath, mathrsfs, float, amsthm}
\usepackage{thmtools, thm-restate}
\usepackage{xcolor}

\usepackage{hyperref, cleveref}

\newcommand{\Var}{\mathrm{Var}}

\DeclareMathOperator*{\prob}{\mathbb{P}}
\DeclareMathOperator*{\expect}{\mathbb{E}}
\DeclareMathOperator*{\E}{\expect}
\renewcommand{\Pr}{\prob}

\usepackage{accents}
\newcommand{\vect}[1]{\accentset{\rightharpoonup}{#1}}

\newcommand{\transpose}[1]{{#1}^{\mathsf{T}}}

\newcommand{\iid}{\overset{\textrm{i.i.d.}}{\sim}}

\newcommand{\numactions}[0]{n}
\newcommand{\numactionscolumn}[0]{m}
\newcommand{\numsamples}[0]{k}
\newcommand{\randomrow}[1]{I_{#1}}
\newcommand{\randomcolumn}[1]{J_{#1}}

\newcommand{\randomsample}[1]{X_{#1}}

\newcommand{\estimator}[0]{g}
\newcommand{\estimatorcolumn}[0]{h}

\newcommand{\mutualinformation}[1]{G_{#1}}

\newcommand{\mutualinformationabstract}[0]{G}

\newcommand{\tallymatrix}[0]{q}
\newcommand{\tallymatrixrandom}[0]{Q}
\newcommand{\tallymatrixentry}[2]{q_{#1, #2}}
\newcommand{\tallymatrixrandomentry}[2]{Q_{#1, #2}}

\newcommand{\jointdistribution}[0]{F}
\newcommand{\jointdistributionentry}[2]{F_{#1, #2}}

\newcommand{\tallyset}[3]{C(#1, #2 \times #3)}
\newcommand{\multinomialdistribution}[2]{\mathrm{Mult}(#1, #2)}

\newcommand{\indexpermutation}[0]{\pi}

\newcommand{\alphabetpermutation}[0]{P}

\newcommand{\timeindex}[0]{t}

\newcommand{\probvectorone}[0]{\vect{v}}
\newcommand{\probvectoroneentry}[1]{v_{#1}}

\newcommand{\probvectortwo}[0]{\vect{w}}
\newcommand{\probvectortwoentry}[1]{w_{#1}}

\newcommand{\diagonalmatrix}[0]{D}

\newcommand{\convexfunction}[0]{\varphi}

\newcommand{\rowsum}[1]{r_{#1}}
\newcommand{\columnsum}[1]{c_{#1}}

\newcommand{\xparallelograminit}[0]{\varepsilon_{0}}
\newcommand{\yparallelograminit}[0]{\delta_{0}}
\newcommand{\xparallelogram}[0]{\varepsilon}
\newcommand{\yparallelogram}[0]{\delta}
\newcommand{\diagtopleft}[0]{p}

\newcommand{\mgfparam}[2]{t_{#1, #2}}
\newcommand{\mgfexp}[0]{\exp(\vect{t} \cdot \vect{Q})}
\newcommand{\mgfA}[0]{A}
\newcommand{\mgfB}[0]{B}
\newcommand{\mgfC}[0]{C}
\newcommand{\mgfD}[0]{D}
\newcommand{\mgfsum}[1]{S_{#1}}

\DeclareMathOperator*{\rad}{\mathrm{rad}}

\newcommand{\ideal}{J}
\newcommand{\vanishinglocus}[1]{V(#1)}
\newcommand{\vanishingideal}[1]{I(#1)}

\newcommand{\scoringrule}{R}
\newcommand{\VoI}[2]{\mathrm{VoI}_{#2}(#1)}
\newcommand{\MI}{\mathcal{I}}
\newcommand{\logscore}{\ell}
\newcommand{\quadscore}{\mathrm{quad}}

\newcommand{\norm}[1]{\left|\left|#1\right|\right|_2}
\newcommand{\stoppingtime}{T}
\newcommand{\DMI}{\mathrm{DMI}}

\newcommand{\R}{\mathbb{R}}

\newcommand{\indicate}[2][]{\text{\bf 1}\ifthenelse{\not\equal{}{#1}}{_{#1}}{}\!\left[{\def\givenn{\middle|}#2}\right]}

\newtheorem{theorem}{Theorem}
\newtheorem{lemma}[theorem]{Theorem}
\newtheorem{proposition}[theorem]{Proposition}
\newtheorem{corollary}[theorem]{Corollary}
\newtheorem{definition}[theorem]{Definition}

\newtheorem{example}[theorem]{Example}

\newtheorem{innercustomthm}{Theorem}
\newenvironment{manualtheorem}[2][]
  {%
    \renewcommand\theinnercustomthm{#2}%
    \if\relax\detokenize{#1}\relax
      \innercustomthm
    \else
      \innercustomthm[#1]%
    \fi
  }
  {\endinnercustomthm}

\newtheorem{innercustomprop}{Proposition}
\newenvironment{manualprop}[1]
  {\renewcommand\theinnercustomprop{#1}\innercustomprop}
  {\endinnercustomprop}

\setcitestyle{authoryear}

\title{Sample Complexity of Peer Prediction\thanks{
This work was initiated while all authors were at the 2025 Winter School on Data Economics at the Moroccan Center for Game Theory at UM6P.  Special thanks to Nicole Immorlica for many conversations and guidance on the early stages of the work, and to Jacob Urisman for guidance on the nullstellensatz framing.}}

 \author{Abdellah Aznag\thanks{Columbia University. Emails: \texttt{\{aa4693, rac2239\}@columbia.edu}. Supported in part by NSF grant IIS-2147361.} \and Robin Bowers\thanks{University of Colorado Boulder. Emails: \texttt{\{robin.bowers, bwag\}@colorado.edu}.} \and Rachel Cummings\footnotemark[2]
  \and Jason Hartline\thanks{Northwestern University. Emails: \texttt{\{hartline, matthewvonallmen2026\}@northwestern.edu}.} \and Matthew vonAllmen\footnotemark[4] \and Bo Waggoner\footnotemark[3]
}

\begin{document}

\maketitle

\begin{abstract}
    Peer prediction seeks to incentivize agents to truthfully report an observed signal by rewarding joint sets of reports without observing a ground truth. Following the generalization of information-theoretic mutual information introduced in \citet{KS-19}, we call a function of a joint distribution over signals a \emph{mutual information} when it is non-negative and disincentivizes garbling reports for all information structures. An \emph{unbiased estimator} for a mutual information takes some number of samples from the distribution and returns rewards for both agents, such that the expected reward is equal to the mutual information. We seek to characterize the set of mutual informations with unbiased estimators for a given number of samples.
    
    We show that for three or fewer sampled report pairs, the only mutual information with an unbiased estimator is trivially zero, and for four or five samples with a binary report space, the Determinant Mutual Information (DMI) of \citet{K-24} is the unique mutual information (up to a scalar multiple). We further show that DMI ceases to be unique at six samples. We provide an improved estimator of DMI for any given number of samples and characterize its convergence rate.

    We also examine mutual information estimators that accept a randomized number of samples. First, we show that mutual information estimators on an ex-ante bounded number of samples (termed ``stop-short estimators'') can achieve a lower variance than an equivalent fixed-sample estimator (for DMI).  Second, we introduce the class of \emph{scoring-rule-based} mutual informations and identify in this family a mutual information that can be estimated with under three samples in expectation.
\end{abstract}

\section{Introduction}\label{sec:introduction}

{\em Peer prediction} \citep{MRZ-05} attempts to elicit truthful reports of observed signals or beliefs from several agents by comparing their reports to one another, rather than by comparing to a subsequently revealed ground truth.  The peer prediction setting is relevant when there is an absence of ground truth (e.g., if agents are reporting personal preferences), or when ground truth is difficult or impossible to observe (e.g., predicting far-future events, or a ground truth that is costly to observe).  This paper considers the \emph{multi-task} setting~\citep{DG-13}, in which each agent is assigned multiple tasks (for example, providing a rating of multiple restaurants or books).  Despite significant literature (see \Cref{sec:related-works}), there are few mechanisms that satisfy good incentive properties for this problem without assumptions on the information structure. A promising direction is the {\em mutual information} framework \citep{KS-19,K-24}. Here, the observations of the two agents are viewed as dependent random variables with some unknown joint distribution. A strategy of an agent maps their observation to a report, so a strategy is a \emph{garbling} of their observed random variable.

The mutual information framework, which assigns payments to agents to incentivize truthful reporting of observed signals, is based on two definitions: a) a \emph{mutual information} measure for joint distributions and b) an \emph{unbiased estimator} for the mutual information, as a function of samples. For a given set of sampled joint signals for agents, an unbiased estimator determines payments to agents such that the expected payment over possible samples is exactly the mutual information measure of the joint distribution of agent reports. 

A mutual information measure must satisfy two key properties: 
\begin{enumerate}
    \item the \textit{data processing inequality}: if either agent chooses to garble their reports, the expected reward becomes smaller, and 
    \item \textit{zero on independent play}: if the agents make independently distributed reports, their expected reward must be $0$.
\end{enumerate}

The data processing inequality removes the strategic incentives of agents to attempt to increase their expected reward by garbling their reports; requiring a value of zero on independent play restricts our attentions to a canonical representative for every function satisfying the data processing inequality up to addition. The classic Shannon mutual information satisfies both these properties (see \Cref{app:entropy}). 

The Shannon mutual information, however, does not have a finite-sample unbiased estimator \citep{P-03}. The existence of an unbiased estimator is crucial in our ability to port the strategic guarantees of the data processing inequality to the peer prediction setting where only samples of the distribution can be observed. Thus, the focus of this paper is on mutual informations with finite-sample unbiased estimators. 

\citet{K-24} proposed \emph{Determinant Mutual Information} (DMI) as a mutual information with a finite-sample unbiased estimator. Motivated by this result, we seek to understand the space of mutual informations with finite-sample unbiased estimators.

We also pose three questions about the \emph{sample complexity} of mutual informations:
\begin{enumerate}
    \item When are there other mutual informations besides DMI that have unbiased finite-sample estimators?
    \item What is the minimum number of samples needed to construct an unbiased estimator of a mutual information?
    \item Given a mutual information with a finite-sample unbiased estimator, how does the variance decrease as the number of samples is increased?
\end{enumerate}

For an unbiased estimator, minimizing variance is equivalent to minimizing mean-squared error. Thus, it is desirable to minimize variance as much as possible for a given number of samples. Furthermore, for practical reasons, agents cannot be asked to answer infinite number of tasks; understanding mutual informations with finite-sample unbiased estimators allows us to choose appropriate algorithms for concrete settings.

\subsection{Results}\label{sec:results}

This paper focuses on understanding the sample complexity of estimators for mutual informations. Particularly, for a small numbers of samples, we seek to understand the set of mutual informations which have an unbiased estimator using that number of samples. 

We then relax our restriction on number of samples and consider mutual informations with unbiased estimators that require an \emph{expected} number of samples, but for a single instance may require many more or much fewer samples. 

\subsubsection{Fixed-Sample Estimators for the $2 \times 2$ Setting}

We start in \Cref{sec:fixed-sample-mutual-information-estimators} by focusing on mutual informations for $2 \times 2$ action spaces which have an unbiased estimator for a fixed number of samples. Our main result of this section states that when agents submit reports in a binary space, the Determinant Mutual Information (DMI) of \citet{K-24} is the \textit{unique} mutual information that possesses an unbiased estimator when constrained to sufficiently few samples. This result holds up to arbitrary rescaling of DMI.

\begin{manualtheorem}[Uniqueness of DMI]{\ref{thm:dmi-is-unique} (Informal)}
    For a $2 \times 2$ report space and 4 or 5 samples, $\DMI$ is the unique mutual information with an unbiased estimator.
\end{manualtheorem}

At $3$ or fewer samples, it is not possible to unbiasedly estimate any non-trivial mutual information, which we show separately in \Cref{thm:three-or-fewer-samples-implies-trivial-mutual-information}. Thus, DMI has an unbiased estimator with strictly fewer samples than any other MI.

Furthermore, we also show in \Cref{thm:dmi-is-not-unique-at-six-samples} that with a sufficiently large number of samples, DMI is no longer the unique mutual information with an unbiased estimator. In a $2 \times 2$ action space, once the number of samples $\numsamples$ is $6$ or above, there are other mutual informations with unbiased estimators.

We next observe that the estimator of DMI provided by \citet{K-24} is not \emph{sample-order-invariant}: if the order of observed samples is altered, the mechanism may return a different payment. In \Cref{thm:convex-minimal-dmi-estimator-with-binary-alphabet}, we derive a sample-order-invariant estimator of DMI in a $2 \times 2$ action space and arbitrary $k \geq 4$ samples, and provide a closed form for the estimator.

\begin{manualtheorem}[$\numsamples$-sample optimal DMI estimator on binary alphabet]{\ref{thm:convex-minimal-dmi-estimator-with-binary-alphabet} (Informal)}
    For agents with a binary report space, the $k$-sample sample-order-invariant unbiased estimator for DMI has a closed form.
\end{manualtheorem}

We prove that this estimator is the unique sample-order-invariant estimator of DMI in \Cref{cor:output-of-permutation-variance-grinder-is-unique-and-convex-minimal}, and that it achieves the minimal variance among all possible fixed-sample DMI unbiased estimators.\footnote{The DMI estimator we provide is not just minimal variance, but also is \textit{convex-dominated} by all other unbiased DMI estimators. For more details, see the results presented in \Cref{cor:output-of-permutation-variance-grinder-is-unique-and-convex-minimal}.} We also characterize its convergence rate in \Cref{thm:convex-minimal-dmi-estimator-with-binary-alphabet-variance}.

\subsubsection{Ex-Ante Bounded-Sample Estimators}

Next, in \Cref{sec:ex-ante-bounded-sample-mutual-information-estimators} we consider mutual informations with estimators that may require a variable number of samples. In \Cref{sec:stop-short-estimators}, we introduce the notion of \textit{stop-short estimators}, which are non-fixed-sample estimators with even lower variance than any fixed-sample estimator with the same mean, with a focus on the $2 \times 2$ action space setting. We showcase an example of such a stop-short estimator in \Cref{thm:stop-short-estimator}.

In the remainder of \Cref{sec:ex-ante-bounded-sample-mutual-information-estimators}, we explore estimators for arbitrary $\numactions\times \numactionscolumn$ action spaces with finite samples \emph{in expectation}. We refer to these as \emph{ex-ante bounded-sample} estimators. \Cref{sec:scoring-rule-based} introduces \emph{value of information-based mutual information}, which draws from the rich literature on proper scoring rules to construct general mutual informations. We formalize the relationship between proper scoring rules and mutual informations, stated informally below.

\begin{manualtheorem}{\ref{lemma:scoring-rule-based-MI} (Informal)}
    For any proper scoring rule, there is an associated value of information-based mutual information. Moreover, if the scoring rule is strict, the associated mutual information is also strict. 
\end{manualtheorem}

The data processing inequality allows the mutual information to not change when the signal is garbled.  In fact, this is necessary as some garblings do not change the signal structure. An extreme example is when the initial signal structure is independent and the mutual information is zero, then garblings do not strictly reduce the mutual information.  On the other hand, DMI is zero whenever the joint distribution is not full rank. From such a joint distribution, further garblings cannot reduce DMI because it is already zero.  Given a joint distribution of signals we say a garbling of one player is {\em strict} if the posterior beliefs on the signal of the other players before and after garbling are distinct.  A mutual information is {\em strict} (\Cref{def:strict-MI}) if it is strictly reduced by any strict garbling. Note that for $2 \times 2$ signal structures, any nontrivial mutual information is inherently strict. 

We provide an ex-ante bounded-sample value of information-based mutual information estimator based on the quadratic scoring rule.

\begin{manualprop}{\ref{prop:better-collision} (Informal)}
    The quadratic score-based mutual information on an $\numactions\times \numactionscolumn$ action space has an ex-ante bounded $3$-sample unbiased estimator and is strict.
\end{manualprop}

Lastly, we show in \Cref{thm:no-fixed-sample-scoring-rule-mi} that no scoring-rule-based mutual information has a fixed-sample unbiased estimator, implying that scoring-rule-based mutual informations and DMI belong to altogether different classes of reward systems in peer prediction.

\subsection{Related Work}\label{sec:related-works}

\paragraph{Peer prediction.}
Peer prediction has been well-studied under a number of models. We consider the multiple task setting with two agents, where we are allowed multiple samples from the joint signal distribution with which to evaluate the agents' performance and (only) ask agents to directly report their signals. This setting was pioneered by \citet{DG-13} under some assumptions on the information structure. Further results were obtained by \citet{SAFP-16} and \citet{KS-19}. The latter work inspired the information-theoretic perspective on multi-task peer prediction underlying this work. It was the first to use the term ``mutual information'' to describe functions that (1) decrease when signals are garbled and (2) are $0$ on independent signals, which is one way to generalize the Shannon mutual information. \citet{KS-19} also includes a third axiom in the definition of a mutual information, requiring the function to be symmetric in the joint signal distribution; we do not impose this requirement on mutual informations.

\citet{K-24} introduced the first fixed-sample information-theoretic reward scheme as \emph{determinant mutual information (DMI)}, which we explore in depth. \citet{K-22} generalized the DMI approach to Volume Mutual Information, where a higher correlation or ``mutual information'' corresponds to larger volume of joint distributions by some volume measure. In contrast to these works, we seek a characterization of all possible mutual informations in certain settings. We also introduce the concept of \emph{ex-ante bounded} sample mutual information estimators, a middle ground between the deterministic number of samples used by DMI and payment schemes which are treated only in the limit.

In recent follow up work, \citet{K-26} resolves several open questions from this paper.  These are discussed in more detail in \Cref{sec:conclusions-and-future-work}.

Another line of work on peer prediction asks agents to report both their signals and their beliefs about others reports \citep{P-04}.  For a more thorough treatment of existing peer prediction mechanisms, we refer the reader to the recent survey of \citet{F-23}.

\paragraph{Proper scoring rules.}
We also build on proper scoring rules for eliciting truthful beliefs, referring to \citet{S-71,GR-07}.
Proper scoring rules are commonly used in peer prediction, including the original peer prediction~\citep{MRZ-05} and Bayesian Truth Serum~\citep{P-04} works, but have been less utilized for multi-task peer prediction.  Our proper-scoring-rule-based mutual informations define a quantity called the value of information.  Several recent works have considered optimizing scoring rules for this quantity \citep{LHSW-22,PW-22}.

\paragraph{Information Theory}
Our work draws on existing concepts from information theory, particularly regarding mutual information \citep{S-48}. In \Cref{app:entropy}, we explore the connections between our scoring-rule-based mutual informations (\Cref{sec:ex-ante-bounded-sample-mutual-information-estimators}) and the traditional notions of Shannon mutual information and entropy. Shannon mutual information does not have an unbiased estimator as we require for this work \citep{P-03}, but there is a line of work on approximate estimators (e.g. \citet{VYK-07,APP-14}). 

\paragraph{Bernoulli Factories.}
Unbiased estimators are conceptually very similar to Bernoulli factories. These protocols take in some number of (arbitrarily biased) Bernoulli random variables, and by repeatedly (and possibly unboundedly) sampling from them, return a new random variable which is $1$ with probability exactly some function of the input biases and $0$ otherwise (e.g., \citet{NLS-21}). Our unbiased estimators are less restricted in their output values, but have the same goal of having an expectation exactly equal to some function of their inputs.

\paragraph{Stochastic matrix decompositions.} Our \Cref{lma:stochastic-matrix-decomposition} is a constructive variant of a proposition from \citet[Proposition 2]{VS-25}, which explores generating sets for column stochastic matrices. The topic of stochastic matrix decompositions beyond the case of $2 \times 2$ column stochastic matrices, is not currently well-explored; \citet{VS-25} provides a decomposition result for $3 \times 3$ column-stochastic matrices, but the dimensionality of the generating set provided is too high to use in our own analyses of mutual informations.

\section{Model and Preliminaries}\label{sec:model}

We work in the standard multi-task peer prediction setting \citep{DG-13,SAFP-16,KS-19,K-24}. The primitives are as follows:
\begin{itemize}
    \item There are two agents, called the row player and the column player.
    \item There are $\numsamples$ tasks, indexed by $\timeindex \in [\numsamples] \triangleq \{ 1, \ldots, \numsamples \}$.
    \item The row player's signal alphabet is $[\numactions] \triangleq \{1, \ldots, \numactions \}$, and the column player's signal alphabet is $[\numactionscolumn] \triangleq \{1, \ldots, \numactionscolumn \}$.
    \item For each task $\timeindex$, a signal pair $(\randomrow{\timeindex},\randomcolumn{\timeindex})$ is drawn i.i.d. from an unknown joint distribution $\jointdistribution \in \Delta([\numactions] \times [\numactionscolumn])$. We use the notation $\Delta([\numactions] \times [\numactionscolumn])$ to denote the set of all $\numactions \times \numactionscolumn$ matrices whose entries are probabilities which sum to $1$.
    \item The agents use possibly randomized strategies, fixed across tasks, mapping each signal to a (possibly random) action, also called a report. A row strategy is a possibly randomized mapping $[\numactions] \to \Delta([\numactions])$, equivalently represented as a column-stochastic matrix $S\in \R^{\numactions\times \numactions}$. Similarly, a column strategy is a possibly randomized mapping $[\numactionscolumn]\to\Delta([\numactionscolumn])$, equivalently represented as a row-stochastic matrix $T\in \R^{\numactionscolumn\times \numactionscolumn}$. Together with $\jointdistribution$, $S$ and $T$ describe a new joint distribution $S\jointdistribution T$ of actions submitted to the mechanism.
    \item A mechanism observes the $\numsamples$ reported action pairs $(\randomrow{1}', \randomcolumn{1}'), \dots, (\randomrow{\numsamples}', \randomcolumn{\numsamples}') \overset{\textrm{i.i.d.}}{\sim} S \jointdistribution T$ and computes row and column player payments via $\estimator : ([n] \times [m])^{\numsamples} \to \mathbb{R}$ and $\estimatorcolumn : ([n] \times [m])^{\numsamples} \to \mathbb{R}$, respectively.
\end{itemize}

We study mechanisms which ask each agent to report only their own observed signal without auxiliary information such as beliefs about the other agent's report, cf. \citet{P-04}. Thus, each player's action alphabet is the same as their signal alphabet.

A mechanism is \emph{truthful} if the identity strategy maximizes each agent's expected payment.

A focal class of mechanisms are those with \emph{common payoff} $\estimator = \estimatorcolumn$. A focal class of environments are those with \emph{symmetric alphabet sizes} $\numactions = \numactionscolumn$.

As a notational convention, random variables are denoted with capital letters, and possible realizations of those random variables are denoted with lowercase letters. Matrices are also denoted with capital letters.

\subsection{Joint Distributions \& Agent Strategies}

Joint distributions of play where the agents' actions are uncorrelated should be given a low expected reward, and ones that are highly correlated should be given high expected reward. This section gives a precise definition of correlated and uncorrelated joint distributions. It will be useful to express joint distributions of play as the product of an initial, ungarbled signal, followed by a garbling. \Cref{prop:stochastic-diagonal-decomposition} gives this decomposition.

\begin{definition}[Independent Play]\label{def:independent-play}
    A joint distribution $\jointdistribution$ \textit{results from independent play} if there exist probability vectors $\probvectorone \in \Delta_{\numactions}, \probvectortwo \in \Delta_{\numactionscolumn}$ such that $\jointdistribution = \probvectorone \; \transpose{\probvectortwo}$.
\end{definition}

\Cref{def:independent-play} represents joint distribution matrices where the two players's signals are completely uncorrelated. Any matrix that satisfies \Cref{def:independent-play} is rank $1$.

\begin{definition}[Diagonal]
    A joint distribution $\jointdistribution$ is  \textit{diagonal} when the probability of playing any non-diagonal action pair is zero, i.e. when, for all indices $i \in [\numactions], j \in [\numactionscolumn]$ where $i \neq j$, we have $\jointdistributionentry{i}{j} = 0$.
\end{definition}

Diagonal joint distributions represent perfect correlation between agents receiving signals in $[\min(\numactions, \numactionscolumn)]$, where the probability of each signal value is given by each diagonal entry of the joint distribution matrix. This interpretation is justified by \Cref{prop:stochastic-diagonal-decomposition}, below.

\begin{restatable}[Stochastic-diagonal decomposition]{proposition}{stochasticdiagonal}
\label{prop:stochastic-diagonal-decomposition}
    Let $\jointdistribution \in \Delta([\numactions] \times [\numactionscolumn])$ be a joint distribution. If $\numactions \geq \numactionscolumn$, then there exists a unique diagonal matrix $\diagonalmatrix \in \Delta([\numactionscolumn] \times [\numactionscolumn])$ and a (potentially non-unique) column-stochastic matrix $S$ such that $\jointdistribution = S \diagonalmatrix$. If $\numactions \leq \numactionscolumn$, then there exists a unique diagonal matrix $\diagonalmatrix \in \Delta([\numactions] \times [\numactions])$ and a (potentially non-unique) row-stochastic matrix $S$ such that $\jointdistribution = \diagonalmatrix S$.
\end{restatable}

In other words, any joint distribution can be obtained by both agents perfectly observing a common signal and then a single one of the agents independently applying a garbling to the observed signal. We prove this proposition in \Cref{app:prelims}.

\subsection{Mutual Informations \& Estimators}\label{s.truthfulmechdesign}

We are interested in designing mechanisms which incentivize \emph{truthful} reporting of signals as actions, i.e., we want selecting the strategy $S = I$ to maximize the agents' expected reward,
\begin{align*}
    \E_{(\randomrow{\timeindex}, \randomcolumn{\timeindex}) \overset{\mathrm{i.i.d.}}{\sim} S \jointdistribution} [\estimator((\randomrow{1}, \randomcolumn{1}), \dots, (\randomrow{k}, \randomcolumn{k}))],
\end{align*}
for all joint distributions $\jointdistribution$.

\citet[Theorem 3.3]{KS-19} showed that a mechanism is truthful (in fact, satisfies stronger notions such as \emph{dominant truthfulness}) as long as each agent is paid according to a \emph{mutual information}. In this paper, we therefore focus on when and how mutual information payment rules can be constructed.

\begin{definition}[Mutual Information]\label{def:mutual-information}
    A mutual information $\mutualinformationabstract : \Delta([\numactions] \times [\numactionscolumn]) \rightarrow \mathbb{R}$ is a mapping from the set of joint distributions over $[\numactions] \times [\numactionscolumn]$ into the real numbers, which satisfies the following three properties:\footnote{\citet[Definition 3.1]{KS-19} defines an ``information monotone'' expected payoff function, a precursor to our definition of a mutual information, as one that satisfies zero on independent play, non-negativity, the row data processing inequality, and symmetry in the joint distribution, i.e., that $\mutualinformationabstract(\jointdistribution) = \mutualinformationabstract(\transpose{\jointdistribution})$. This last axiom is not well-defined if $\numactions \neq \numactionscolumn$, so we make the more general assumption of satisfying both row and column data processing inequalities. Furthermore, either data processing inequality and the zero on independent play axiom implies non-negativity, so we drop this axiom as well.}
    \begin{enumerate}
        \item \textbf{The (row) data processing inequality.} For all joint distribution matrices $\jointdistribution$, $\numactions \times \numactions$ column-stochastic matrices $S$, 
        \begin{align*}
            \mutualinformationabstract(S \jointdistribution) \leq \mutualinformationabstract(\jointdistribution) \text{.}
        \end{align*}
        \item \textbf{The (column) data processing inequality.} For all joint distribution matrices $\jointdistribution$, $\numactions \times \numactions$ row-stochastic matrices $T$, 
        \begin{align*}
            \mutualinformationabstract(\jointdistribution T) \leq \mutualinformationabstract(\jointdistribution) \text{.}
        \end{align*}
        \item \textbf{Zero on independent play.} For all rank $1$ joint distribution matrices $\jointdistribution$, $\mutualinformationabstract(\jointdistribution) = 0$.
    \end{enumerate}
    A \textit{row mutual information} satisfies axioms (1) and (3), while a \textit{column mutual information} satisfies axioms (2) and (3).
\end{definition}

A row mutual information, paired with an unbiased estimator, provides a strategy-proof payment rule for a row player. In general, when $\numactions = \numactionscolumn$, we can obtain a column mutual information payment rule for the column player by flipping the agents' roles and evaluating $\mutualinformationabstract(\transpose{\jointdistribution})$, which provides a strategy-proof payment rule for the column player. A mutual information satisfying both axioms enforces the stronger requirement that \emph{neither} player can increase the payment of the row player via strategic behavior. Satisfying all three properties makes a single mutual information $\mutualinformationabstract$ an appropriate, strategy-proof payment rule to reward both row and column players with the same reward on an instance-by-instance basis.

A row (respectively, column) mutual information ignores the strategy of the second (respectively, first) player, and focuses on designing a dominantly truthful payment rule for only one player.

\paragraph{Estimators}
We note that mutual informations are functions of the underlying joint distribution, which we cannot observe. Therefore, we are interested in mutual informations which have \emph{unbiased estimators}, such that they can be estimated with a finite number of samples.

\begin{definition}[Unbiased Estimator]\label{def:unbiased-estimator}
    The function $\estimator:([\numactions]\times[\numactionscolumn])^k\rightarrow \R$ for is a $k$-sample unbiased estimator for a function $\mutualinformationabstract:\Delta([\numactions] \times [\numactionscolumn]) \rightarrow \mathbb{R}$  if,
    \begin{align*}
        \mutualinformationabstract(\jointdistribution) &= \expect_{(\randomrow{\timeindex}, \randomcolumn{\timeindex}) \overset{\textrm{i.i.d.}}{\sim} \jointdistribution}[\estimator((\randomrow{1}, \randomcolumn{1}), ..., (\randomrow{\numsamples}, \randomcolumn{\numsamples}))].
    \end{align*}
\end{definition}

We also refer to these as fixed-sample unbiased estimators, in contrast to a later generalization in which a variable number of samples may be drawn (Section \ref{sec:ex-ante-bounded-sample-mutual-information-estimators}). Throughout this paper we may drop the term "unbiased" and simply refer to estimators, which are implicitly restricted to be unbiased.

In general, a mutual-information-based mechanism pays each agent according to an unbiased estimator of some mutual information. The agents' expected utility is the mutual information of their report random variables. The data processing inequality therefore implies that truthful reporting is an agent's optimal strategy \citep[see][]{KS-19}.

In particular, \citet{K-24} introduces the \emph{determinant mutual information} (DMI), along with an unbiased estimator of DMI. She defines DMI of a joint distribution $\jointdistribution\in \Delta([\numactions] \times [\numactions])$ as $\det(\jointdistribution)^2$, and introduce the following estimator.

\begin{definition}[\citet{K-24}'s DMI Estimator]\label{def:kong-estimator}
    The DMI of a joint distribution $\jointdistribution\in \Delta([\numactions] \times [\numactions])$ can be estimated as follows:
    \begin{enumerate}
        \item Take a sufficiently large number of samples $\numsamples \geq 2\numactions$ of pairs of the agents' reports $(\randomrow{1}, \randomcolumn{1}), ..., (\randomrow{\numsamples}, \randomcolumn{\numsamples})$, where $\randomrow{\timeindex}$ is the $\timeindex^{\text{th}}$ report of the first agent and $\randomcolumn{\timeindex}$ is the $\timeindex^{\text{th}}$ report of the second agent.
        \item Split the samples into two equal-sized groups. Tally the number of report pairs of each type $(i, j)$ in each group in two separate $\numactions \times \numactions$ matrices $\tallymatrix^{1}$ and $\tallymatrix^{2}$, where $\tallymatrixentry{i}{j}^{1}$ and $\tallymatrixentry{i}{j}^{2}$ are the number of times the pair of reports $(i, j)$ was observed in the first and second group, respectively.
        \item Return $\det(\tallymatrix^{1}) \cdot \det(\tallymatrix^{2})$, scaled by a function of the number of samples.
    \end{enumerate}
\end{definition}

\subsection{Simplifying Observations \& Normalization}

We make several observations which simplify our approach, and determine canonical forms of mutual informations up to scalar multiplication. At several points throughout, we will work with sample-order-invariant fixed-sample estimators.

\begin{definition}[Sample-order-invariance]
    A $\numsamples$-sample unbiased estimator $g$ is \emph{sample-order-invariant} if, for any realized sequence of $\numsamples$ samples $(\randomrow{1}, \randomcolumn{1}), \dots, (\randomrow{\numsamples}, \randomcolumn{\numsamples})$ and any permutation $\indexpermutation$,
    \begin{align*}
        \estimator((\randomrow{1}, \randomcolumn{1}), \dots, (\randomrow{\numsamples}, \randomcolumn{\numsamples})) = \estimator((\randomrow{\indexpermutation(1)}, \randomcolumn{\indexpermutation(1)}), \dots, (\randomrow{\indexpermutation(\numsamples)}, \randomcolumn{\indexpermutation(\numsamples)})) \text{.}
    \end{align*}
\end{definition}

Note that the estimator of DMI from \citet{K-24} described in \Cref{s.truthfulmechdesign} is not sample-order-invariant, as the division into $\tallymatrix^{1}$ and $\tallymatrix^{2}$ may affect the value returned. However, we show that any fixed $k$-sample unbiased estimator can be turned into a sample-order-invariant $k$-sample unbiased estimator (\Cref{obs:make-new-estimators}), and in \Cref{sec:optimal-estimator-tools} we provide justification for focusing on these estimators.

This simplification gives rise to an easier way to express a sequence of reports $(\randomrow{1}, \randomcolumn{1}), \dots, (\randomrow{k}, \randomcolumn{k})$, by tracking only the number of occurrences of each pair $(i, j) \in [\numactions] \times [\numactionscolumn]$ in a \textit{tally matrix}.

\begin{definition}[Tally Matrix]\label{def:tally-matrix}
    A \textit{tally matrix} of $\numsamples \in \mathbb{N}$ samples from support $[\numactions] \times [\numactionscolumn]$ is a $\numactions \times \numactionscolumn$ matrix with non-negative integer entries that sum to $\numsamples$. We use $\tallyset{\numsamples}{\numactions}{\numactionscolumn}$ to denote the set of tally matrices of $\numsamples \in \mathbb{N}$ samples from support $[\numactions] \times [\numactionscolumn]$.
\end{definition}

We will draw random tally matrices from a \textit{multinomial distribution} over the action pairs of the row and column player.

\begin{definition}[Multinomial Distribution over Action Pairs]\label{def:multinomial-distribution}
    A random $\numactions \times \numactionscolumn$ tally matrix $\tallymatrixrandom$ follows a \textit{multinomial distribution} over the joint distribution on action pairs $\jointdistribution$ with $\numsamples$ samples when the probability that $\tallymatrixrandom$ is any particular tally matrix $\tallymatrix \in \tallyset{\numsamples}{\numactions}{\numactionscolumn}$ is given by:
    \begin{align*}
       \textstyle \prob(\tallymatrixrandom = \tallymatrix) = \frac{\numsamples!}{\prod_{i = 1}^{\numactions} \prod_{j = 1}^{\numactions} (\tallymatrixentry{i}{j}!)} \prod_{i = 1}^{\numactions} \prod_{j = 1}^{\numactions} \jointdistributionentry{i}{j}^{\tallymatrixentry{i}{j}} \text{.}
    \end{align*}
    We use $\tallymatrixrandom \sim \multinomialdistribution{\numsamples}{\jointdistribution}$ to denote that $\tallymatrixrandom$ follows a multinomial distribution.
\end{definition}

Through light abuse of notation, a sample-order-invariant $\numsamples$-sample unbiased estimator $\estimator$ for a mutual information $\mutualinformation{\estimator}$ is a function $\estimator : \tallyset{\numsamples}{\numactions}{\numactionscolumn} \rightarrow \mathbb{R}$ on the set of tally matrices such that
\begin{align*}
     \expect_{\tallymatrixrandom \sim \multinomialdistribution{\numsamples}{\jointdistribution}}[\estimator(\tallymatrixrandom)] = \mutualinformation{\estimator}(\jointdistribution) \text{.}
\end{align*}
We contrast this to the possibly sample-order-dependent definition of unbiased estimators provided in \Cref{def:unbiased-estimator}.

\paragraph{Normalization}  Any positive rescaling of a mutual information is itself a mutual information. Therefore, it is sometimes useful to \textit{normalize} mutual informations so that they output values between $0$ and $1$, as described in \Cref{def:unit-normalization}.

\begin{definition}[Unit Normalization]\label{def:unit-normalization}
    A mutual information $\mutualinformationabstract$ is \textit{unit-normalized} if $\mutualinformationabstract(\jointdistribution) \leq 1$ for all joint distributions $\jointdistribution$, and where this inequality is tight for some (diagonal) joint distribution $\jointdistribution$.
\end{definition}

\section{Fixed-Sample Mutual Information Estimators}\label{sec:fixed-sample-mutual-information-estimators}

This section considers mutual informations with fixed-sample unbiased estimators. Fixed-sample estimators are desirable as they place a concrete cap on the number of samples needed from each player. Fixed-sample estimators are also required for practical implementability:  players need their effort in terms of number of tasks to be related to the payoff, and the mechanism operator only has a finite number of tasks. \Cref{sec:ex-ante-bounded-sample-mutual-information-estimators} will relax this requirement to only cap the number of samples taken \textit{in expectation}.

We begin with a general characterization of the structure that mutual informations with fixed-sample unbiased estimators must take: \Cref{sec:determinant-mutual-information-tools} gives a polynomial form of such mutual informations. These results are leveraged in \Cref{sec:determinant-mutual-information-results} and \Cref{sec:non-uniqueness-results} in eliminating possible mutual informations and characterizing new ones. \Cref{sec:determinant-mutual-information-results} finds that DMI is the unique mutual information for a $2 \times 2$ action space and $\numsamples\in\{4,5\}$ samples, and that no mutual information has a fixed-sample estimator for $\numsamples < 4$. For $\numsamples = 6$ samples, \Cref{sec:non-uniqueness-results} finds a new mutual information with a $6$-sample estimator, showing that DMI is no longer unique.

In \Cref{sec:optimal-estimator-tools}, we fix any mutual information and consider its estimators. A given mutual information either has no $\numsamples$-sample estimators or an infinite family of them, and we propose that the best estimator is the unique \emph{convex-minimal} estimator. We provide tools for finding the unique convex-minimal fixed-sample estimator for a mutual information, and in \Cref{sec:optimal-estimator-results} we use these tools to explicitly state the convex-minimal estimator for DMI on the $2 \times 2$ action space and any number of samples.

All omitted proofs in this section are presented in \Cref{app.section3proofs}.

\subsection{Polynomial Characterization of Mutual Informations with Fixed-Sample Estimators}\label{sec:determinant-mutual-information-tools}

We begin by presenting several general results about the structure of mutual informations with fixed-sample estimators.  We will use these results in \Cref{sec:determinant-mutual-information-results} to show that DMI is unique for small action spaces and numbers of samples. First, we observe that all mutual informations with fixed-sample unbiased estimators can be written as polynomials.

\begin{restatable}[Fixed-sample unbiased estimators are polynomials]{lemma}{MIpolynomials}
\label{lma:mutual-informations-are-polynomials}
    Any function $\mutualinformationabstract$ on joint distribution $\jointdistribution$ has a $\numsamples$-sample unbiased estimator if and only if $\mutualinformationabstract(\jointdistribution)$ is a multivariate polynomial in the entries of the joint distribution matrix $\jointdistribution \in \Delta([\numactions] \times [\numactionscolumn])$ with degree at most $\numsamples$.
\end{restatable}

Using the polynomial characterization of mutual informations with fixed-sample estimators, we apply Hilbert's nullstellensatz theorem for polynomials to characterize functions which respect the zero on independent play axiom. 

\begin{restatable}[Characterization of zero on independent play]{theorem}{independentnullstellensatz}\label{thm:independent-characterization}
    In an $n \times m$ action space, any polynomial $\mutualinformationabstract$ on joint distribution $\jointdistribution$ which satisfies the zero on independent play axiom must be of the form,
    \begin{align*}
        \sum_{a\neq c, b\neq d}p_{abcd}(\jointdistribution) \det\nolimits_{abcd}(\jointdistribution) \text{,}
    \end{align*}
    where $p_{abcd}(\jointdistribution)$ is some polynomial in the entries of $\jointdistribution$ and $\det_{abcd}(\jointdistribution) = \jointdistribution_{ab}\jointdistribution_{cd} - \jointdistribution_{ad}\jointdistribution_{bc}$, i.e., the determinant of the $2 \times 2$ submatrix of $\jointdistribution$ restricted to rows $a,c$ and columns $b,d$.
\end{restatable}

For the $2 \times 2$ action space, we further refine this result in \Cref{lma:mutual-informations-are-divisible-by-determinant-squared-of-joint-distribution-matrix}, placing stronger restrictions on the polynomials which can form mutual informations in the binary action space.

\begin{restatable}[Factorization by squared determinant]{lemma}{determinantfactorization}\label{lma:mutual-informations-are-divisible-by-determinant-squared-of-joint-distribution-matrix}
    In a $2 \times 2$ action space, any polynomial $\mutualinformationabstract$ on joint distribution $\jointdistribution$ which is non-negative and satisfies the zero on independent play axiom must be divisible by the determinant squared of the joint distribution matrix $\det(\jointdistribution)^{2}$.
\end{restatable}

We note that \Cref{lma:mutual-informations-are-divisible-by-determinant-squared-of-joint-distribution-matrix} applies to functions $\mutualinformationabstract$ with a fixed-sample unbiased estimator that satisfy the mutual information axioms, because non-negativity is implied by a combination of the zero on independent play axiom and the data processing inequality. Non-negativity is thus a weaker assumption than the full data processing inequality.

We will use \Cref{lma:mutual-informations-are-divisible-by-determinant-squared-of-joint-distribution-matrix} to later prove the uniqueness of DMI for small numbers of samples in \Cref{sec:determinant-mutual-information-results}.

\subsection{Uniqueness of DMI with Small Samples}\label{sec:determinant-mutual-information-results}

Using the characterizations of \Cref{sec:determinant-mutual-information-tools}, we find that no nontrivial mutual information has a fixed-sample estimator that takes strictly fewer than 4 samples. For $4$ and $5$ samples and a $2 \times 2$ action space, we find DMI is the only mutual information with an unbiased fixed-sample estimator.

Our first main result in this vein is that it takes at least $4$ samples to estimate any non-trivial mutual information. We note that \Cref{thm:three-or-fewer-samples-implies-trivial-mutual-information} holds regardless of the alphabet sizes $\numactions, \numactionscolumn$ of the row and column players.

\begin{theorem}[3 or fewer samples implies trivial mutual information]
\label{thm:three-or-fewer-samples-implies-trivial-mutual-information}
    No nontrivial mutual information has a fixed-sample unbiased estimator with $\numsamples \leq 3$ samples.
\end{theorem}

\begin{proof}
    Any mutual information with an unbiased estimator of $3$ or fewer samples must be expressible as a polynomial of degree at most $3$ (by \Cref{lma:mutual-informations-are-polynomials}). However, $\det(\jointdistribution)^{2}$ is a polynomial of degree $4$ on a $2 \times 2$ action space. Any mutual information with a fixed-sample unbiased estimator must be divisible by $\det(\jointdistribution)^{2}$ (by \Cref{lma:mutual-informations-are-divisible-by-determinant-squared-of-joint-distribution-matrix}). The only such polynomial is identically zero.
\end{proof}

At $4$ samples on a binary action space, it becomes possible to unbiasedly estimate a non-trivial mutual information. This first non-trivial mutual information is the Determinant Mutual Information (DMI) from \citet{K-24}, and it is unique when the number of samples is at most $5$. We state this result in \Cref{thm:dmi-is-unique}.

\begin{restatable}[DMI is unique for $\numactions = \numactionscolumn = 2$ and $\numsamples = 4$ or $\numsamples = 5$]{theorem}{dmiunique}
\label{thm:dmi-is-unique}
    When there are $\numsamples = 4$ or $\numsamples = 5$ samples on a $2 \times 2$ action space, the unique mutual information (up to a scalar multiple) is:
    \begin{align*}
        \mutualinformationabstract(\jointdistribution) = 16 \det(\jointdistribution)^{2} \text{.}
    \end{align*}
\end{restatable}

Since $16 \det(\jointdistribution)^{2}$ is unique up to a scalar multiple, the factor of $16$ in $16 \det(\jointdistribution)^{2}$ was chosen to make $16 \det(\jointdistribution)^{2}$ \textit{unit-normalized} (\Cref{def:unit-normalization}). Observe that
\begin{align*}
    16 \det\begin{pmatrix}
        1/2 & 0 \\
        0 & 1/2
    \end{pmatrix}^{2} = 1,
\end{align*}
and that any $\jointdistribution \neq \begin{pmatrix}
    1/2 & 0 \\
    0 & 1/2
\end{pmatrix}$ will produce a weakly smaller value of $16 \det(\jointdistribution)^{2}$.

\subsection{Non-uniqueness of DMI}\label{sec:non-uniqueness-results}

While $16 \det(\jointdistribution)^{2}$ of DMI (\Cref{thm:dmi-is-unique}) is the unique unit-normalized mutual information with an unbiased fixed-sample estimator on a binary action space for $4$ or $5$ samples, this uniqueness property disappears once $6$ samples have been collected. We state this fact in \Cref{thm:dmi-is-not-unique-at-six-samples}, and devote the rest of this section to showcasing two lemmas (\Cref{lma:stochastic-matrix-decomposition} and \Cref{lma:binary-alphabet-data-processing-inequality-for-mutual-informations}) used to prove it.

\begin{restatable}[DMI is not the unique mutual information at 6 samples]{theorem}{dminotunique}
\label{thm:dmi-is-not-unique-at-six-samples}
    In a $2 \times 2$ action space, the function
    \begin{align*}
        \mutualinformationabstract(\jointdistribution) = 16 \det(\jointdistribution)^{2} (1 + (\jointdistributionentry{1}{1} + \jointdistributionentry{1}{2} - \jointdistributionentry{2}{1} - \jointdistributionentry{2}{2})^{2} + (\jointdistributionentry{1}{1} + \jointdistributionentry{2}{1} - \jointdistributionentry{1}{2} - \jointdistributionentry{2}{2})^{2}),
    \end{align*}
    is a mutual information, has a $6$-sample unbiased estimator, and is not equal to DMI up to a scalar multiple.
\end{restatable}
Intuitively, the example mutual information in \Cref{thm:dmi-is-not-unique-at-six-samples} is DMI multiplied by an additional term. The additional term is based on the marginals of the row and column players taking a given action: the first term, $(\jointdistributionentry{1}{1} + \jointdistributionentry{1}{2} - \jointdistributionentry{2}{1} - \jointdistributionentry{2}{2})^2$, is equivalently $(\Pr[\randomrow{} = 1] - \Pr[\randomrow{} = 2])^2$, a measure of the marginal difference in the frequency of the row player's actions. Similarly, the second term is a measure of the difference in frequency for the column player. Thus, the example $6$-sample mutual information provides a slightly higher payoff than DMI if the \textit{marginal} action probabilities of the players are non-uniform. We construct an estimator for this additional term and use it, along with existing DMI estimators, to construct an estimator for the function $\mutualinformationabstract$ in \Cref{def:six-sample-estimator}.

To prove \Cref{thm:dmi-is-not-unique-at-six-samples}, we make use much of a much more tractable variant of the data processing inequality which holds on a $2 \times 2$ action space. We derive it from \Cref{lma:stochastic-matrix-decomposition}, which states that any $2 \times 2$ column stochastic matrix can be written as the product of four ``simple'' stochastic matrices.

\begin{restatable}[Stochastic matrix decomposition]{lemma}{stochmatrixdecomp}
\label{lma:stochastic-matrix-decomposition}
    For any $2 \times 2$ column-stochastic matrix $S$, there exist $\lambda_{1}, \lambda_{2} \in [0, 1]$ such that $S$ can be written as:
    \begin{align*}
    S = \left( \lambda_{2} \begin{pmatrix}
        1 & 1 \\
        0 & 0
    \end{pmatrix} + (1 - \lambda_{2}) \begin{pmatrix}
        1 & 0 \\
        0 & 1
    \end{pmatrix} \right) \begin{pmatrix}
        0 & 1 \\
        1 & 0
    \end{pmatrix} \left( \lambda_{1} \begin{pmatrix}
        1 & 1 \\
        0 & 0
    \end{pmatrix} + (1 - \lambda_{1}) \begin{pmatrix}
        1 & 0 \\
        0 & 1
    \end{pmatrix} \right),
    \end{align*}
    up to row permutation.
\end{restatable}

This result has the elegant implication that any row strategy decomposes into an alternating sequence of applying strategies of a) moving some $\lambda$ probability from action $2$ to action $1$ and b) swapping the labels of rows, possibly skipping the final row swap. Since the data processing inequality must hold for any $\jointdistribution$, it must hold for each application of a single strategy in the sequence, and we can simply show that it holds for arbitrary $\jointdistribution$ and both matrix types. We will use this approach to simplify the binary-alphabet data processing inequality (\Cref{lma:binary-alphabet-data-processing-inequality-for-mutual-informations}).

\paragraph{Geometric Interpretation of \Cref{lma:stochastic-matrix-decomposition}}
We can interpret this strategy decomposition geometrically. For a given initial joint distribution matrix $\jointdistribution$, \Cref{prop:stochastic-diagonal-decomposition} tells us that we may decompose $\jointdistribution$ into a diagonal matrix $\diagonalmatrix$ and an initial column-stochastic matrix $T$. The diagonal matrix $\diagonalmatrix$ can be interpreted as an initial weighting on the frequency with which the row player sees signal $1$ versus $2$, with $T$ being some initial garbling the player cannot remove. The combined strategy and decomposition terms $ST$ determine the final garbling applied to the signal observed:
\begin{align*}
    T = \begin{pmatrix}
        1 - \xparallelograminit & \yparallelograminit \\
        \xparallelograminit & 1 - \yparallelograminit
    \end{pmatrix} \quad \text{ and } \quad S T = \begin{pmatrix}
        1 - \xparallelogram & \yparallelogram \\
        \xparallelogram & 1 - \yparallelogram
    \end{pmatrix} \text{.}
\end{align*}
We note that not every column-stochastic matrix can be written as $ST$ for every $T$: in fact, we can characterize the reachable area as a parallelogram in the $(\varepsilon, \delta)$ plane with corners at $(0, 1)$, $(1 - \delta, 1 - \varepsilon)$, $(1, 0)$, and $(\varepsilon, \delta)$ (\Cref{fig:valid-stochastic-region}).

\Cref{lma:stochastic-matrix-decomposition} says that any point within the feasible $(\varepsilon, \delta)$ parallelogram, characterized by the starting point $(\varepsilon_0, \delta_0)$ from $T$, is reachable through a series of at most four moves:
\begin{enumerate}
    \item Move toward the corner $(0, 1)$
    \item Switch current coordinates $(\varepsilon_1, \delta_1)$ to $(1 - \varepsilon_1, 1 - \delta_1)$ (i.e., permute rows)
    \item Move toward the corner $(0, 1)$
    \item (optionally) Switch current coordinates $(\varepsilon_2, \delta_2)$ to $(1 - \varepsilon_2, 1 - \delta_2)$ (i.e., permute rows)
\end{enumerate}
The step of moving toward $(0, 1)$ is equivalent to shifting weight toward action 1. Two visualizations of this process using three and four moves are given in \Cref{fig:procedure-one} and \Cref{fig:procedure-two}, respectively.

\begin{figure}[ht]
    \centering
    
    \begin{subfigure}[t]{0.31\textwidth}
        \centering
        \includegraphics[width=\textwidth]{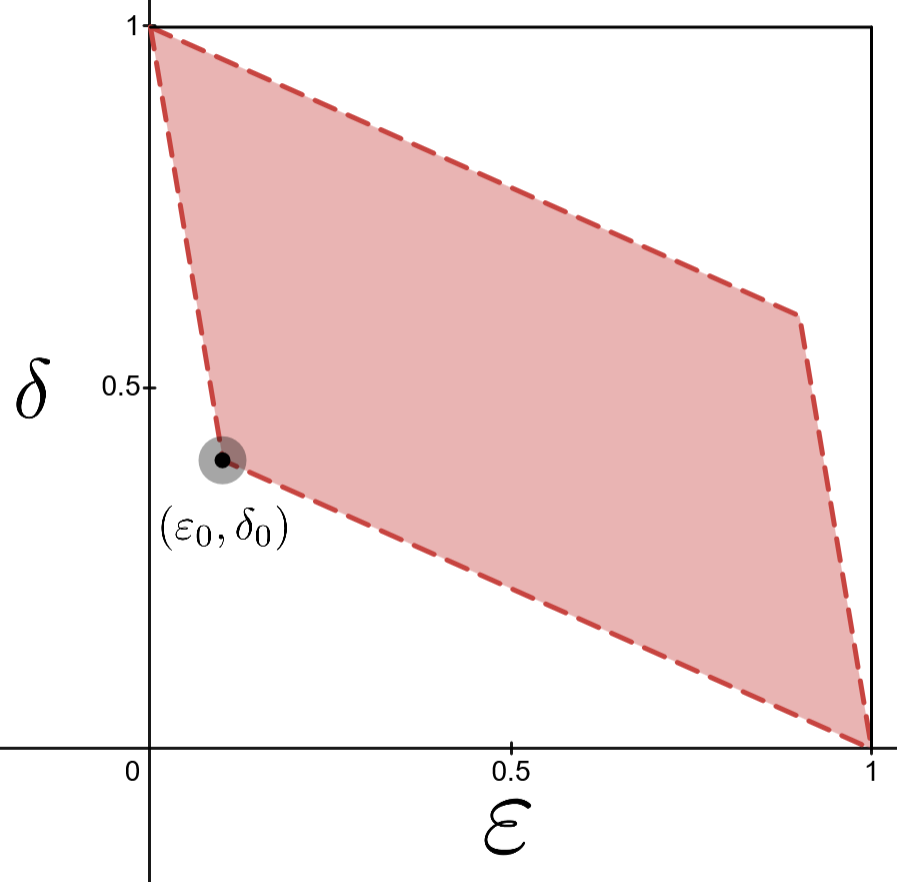}
        \caption{The set of valid $(\xparallelogram, \yparallelogram)$ that correspond to $S T$ when $(\xparallelograminit, \yparallelograminit)$ correspond to $T$}
        \label{fig:valid-stochastic-region}
    \end{subfigure}
    \hfill
    \begin{subfigure}[t]{0.31\textwidth}
        \centering
        \includegraphics[width=\textwidth]{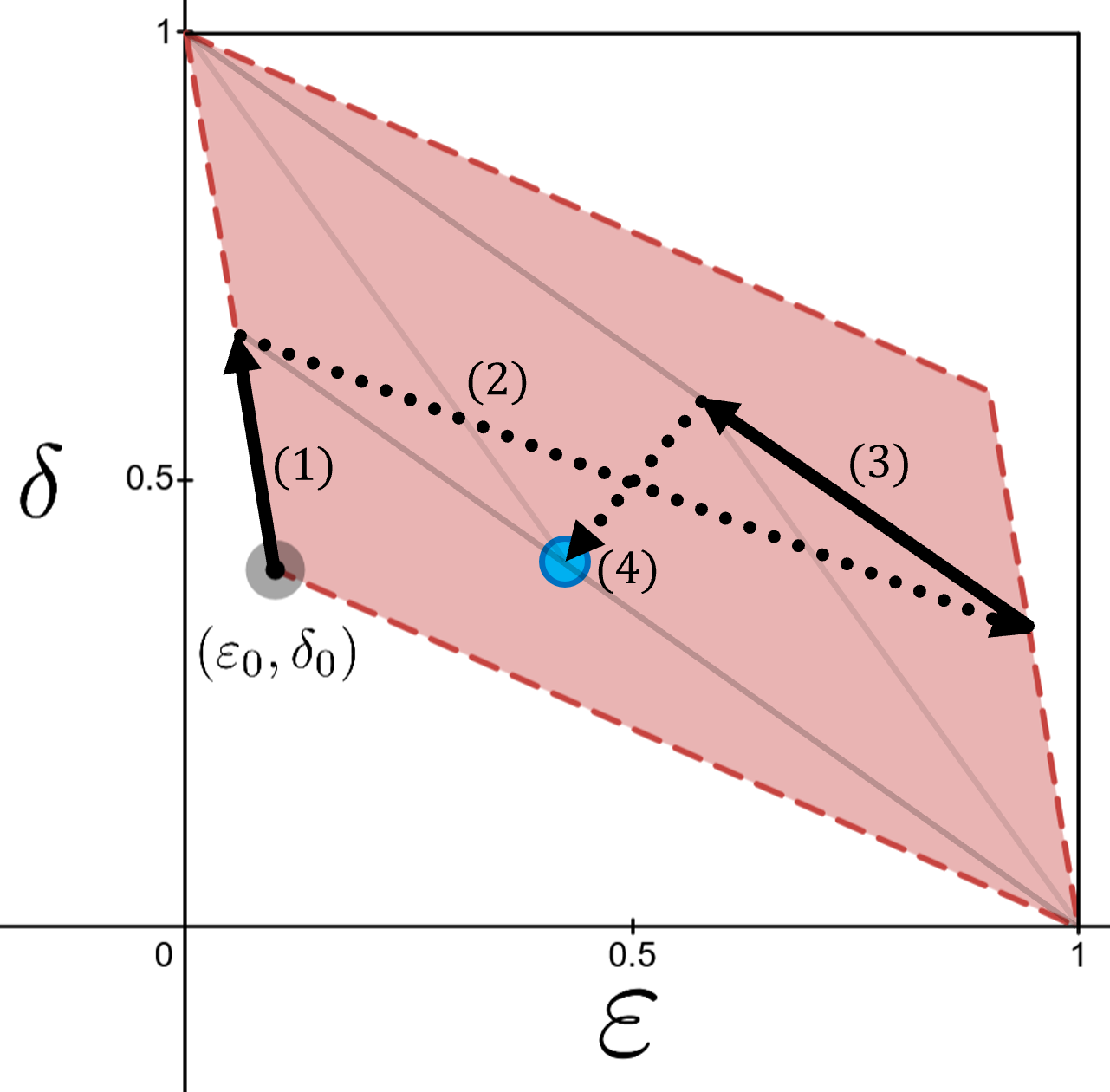}
        \caption{Moving from $(\xparallelograminit, \yparallelograminit)$ to an example $(\xparallelogram, \yparallelogram)$ (the blue dot) by a series of 4 steps: moving toward $(0, 1)$, permuting, moving toward $(0, 1)$, and permuting one final time. Each new parallelogram induced is displayed as transparent grey lines.}
        \label{fig:procedure-one}
    \end{subfigure}
    \hfill
    \begin{subfigure}[t]{0.31\textwidth}
        \centering
        \includegraphics[width=\textwidth]{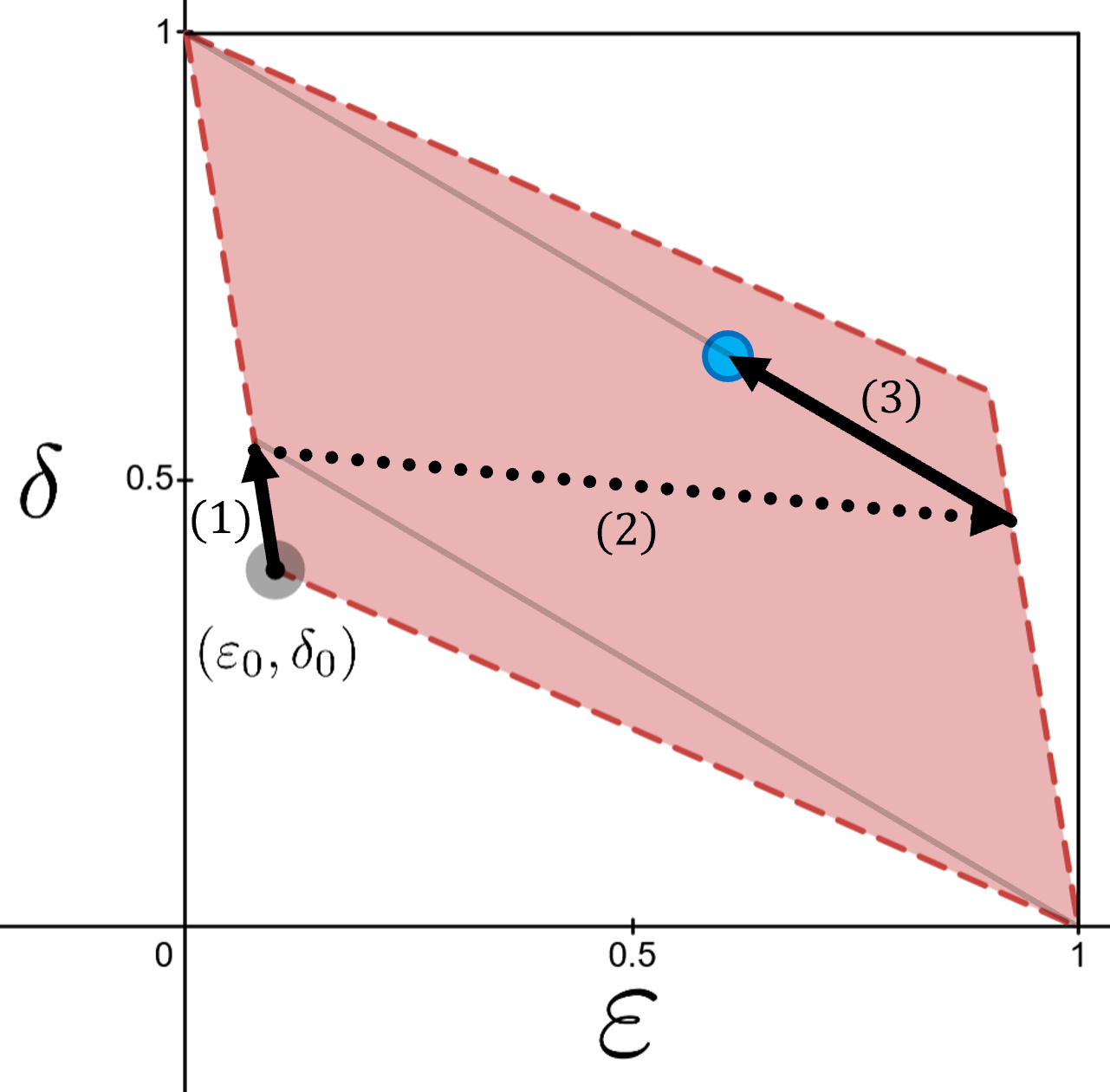}
        \caption{Moving from $(\xparallelograminit, \yparallelograminit)$ to an example $(\xparallelogram, \yparallelogram)$ (the blue dot) by a series of 3 steps: moving toward $(0, 1)$, permuting, and then moving toward $(0, 1)$ a second time. Each new parallelogram induced is displayed as transparent grey lines.}
        \label{fig:procedure-two}
    \end{subfigure}
    
    \caption{Illustrations of geometric interpretation of \Cref{lma:stochastic-matrix-decomposition}.}
    \label{fig:overall}
\end{figure}

A variant of \Cref{lma:stochastic-matrix-decomposition} can be found in \citet{VS-25}, although they provide neither an explicit procedure for creating the decomposition nor our geometric interpretation.

Using \Cref{lma:stochastic-matrix-decomposition}, it suffices to check the data processing inequality on the two ``simple'' matrices used in the decomposition, i.e., the matrix which moves some $\lambda$ probability from action $2$ to action $1$, and the matrix which swaps the labels of rows. By evaluating the data processing inequality exclusively on these two kinds of matrices, we obtain \Cref{lma:binary-alphabet-data-processing-inequality-for-mutual-informations}.

\begin{restatable}[Binary alphabet data processing inequality for MIs]{lemma}{binaryalphabet}
\label{lma:binary-alphabet-data-processing-inequality-for-mutual-informations}
    For a $2 \times 2$ action space, a function of the joint distribution $\mutualinformationabstract : \Delta([2] \times [2]) \rightarrow \mathbb{R}$ satisfies the row data processing inequality, i.e.,
    \begin{align*}
        \mutualinformationabstract(S \jointdistribution) \leq \mutualinformationabstract(\jointdistribution) \text{ for all joint distributions $\jointdistribution$ and column-stochastic $S$,}
    \end{align*}
    if and only if it satisfies both row permutation invariance, i.e.,
    \begin{align*}
        \mutualinformationabstract(\alphabetpermutation \jointdistribution) = \mutualinformationabstract(\jointdistribution) \text{ for all joint distributions $\jointdistribution$ and permutation matrices $\alphabetpermutation$,}
    \end{align*}
    and the inequality
    \begin{align*}
        0 \geq \left[ \frac{\partial \mutualinformationabstract}{\partial \jointdistributionentry{1}{1}} (\jointdistribution) - \frac{\partial \mutualinformationabstract}{\partial \jointdistributionentry{2}{1}} (\jointdistribution) \right] \jointdistributionentry{2}{1} + \left[ \frac{\partial \mutualinformationabstract}{\partial \jointdistributionentry{1}{2}} (\jointdistribution) - \frac{\partial \mutualinformationabstract}{\partial \jointdistributionentry{2}{2}} (\jointdistribution) \right] \jointdistributionentry{2}{2} \text{,}
    \end{align*}
    for all joint distributions $\jointdistribution$.
\end{restatable}

\Cref{lma:binary-alphabet-data-processing-inequality-for-mutual-informations} is important because it decreases the dimensionality of the space of inequalities implied by the data processing inequality.\footnote{Ordinarily, the data processing inequality when the row and column alphabet sizes $\numactions = \numactionscolumn = 2$ states that $\mutualinformationabstract(S \jointdistribution), \mutualinformationabstract(\jointdistribution \transpose{S}) \leq \mutualinformationabstract(\jointdistribution)$ for any $2 \times 2$ column-stochastic $S$ and $2 \times 2$ joint distribution $\jointdistribution$. Since the columns of $S$ sum to $1$, specifying $S$ requires $2$ free parameters, and since the entries of $\jointdistribution$ must sum to $1$, specifying $\jointdistribution$ requires $3$ free parameters. Therefore, to check whether $\mutualinformationabstract$ is a mutual information, one must search a $5$-dimensional space of parameters to identify potential violations of the data processing inequality. With \Cref{lma:binary-alphabet-data-processing-inequality-for-mutual-informations}, only the joint distribution $\jointdistribution$ must be searched over, reducing a $5$-dimensional parameter space to a $3$-dimensional parameter space, which is much more tractable.} \Cref{lma:binary-alphabet-data-processing-inequality-for-mutual-informations} is used in the proof of \Cref{thm:dmi-is-not-unique-at-six-samples} to verify that the function $G(F)$ given in \Cref{thm:dmi-is-not-unique-at-six-samples} with a $6$-sample estimator satisfies the data processing inequality.

\subsection{Convex-Minimal Estimators}\label{sec:optimal-estimator-tools}

So far, we have focused our attention on directly finding mutual informations which have fixed-sample estimators without considering what those estimators might be. A single mutual information, if it has any unbiased estimator, will in fact have many such estimators (\Cref{obs:make-new-estimators}). This section deals with characterizing the \textit{best} fixed-sample estimator for a mutual information for general alphabet sizes $\numactions, \numactionscolumn$ and a general number of samples $\numsamples$.

We argue that a better mutual information estimator will reduce the spread of payments to agents across sample instances, making it more predictable to agents. Formally, we consider the best estimator for a fixed number of samples to be \textit{convex-minimal}. For a pair of unbiased estimators of the same function of the joint distribution, we define convex dominance to compare the spread of the two estimators.

\begin{definition}[Convex Dominance]\label{def:convex-domiance}
    An unbiased estimator $\estimator$ \textit{convex-dominates} an unbiased estimator $\estimator'$ for the same function if, for any joint distribution $\jointdistribution$ and convex function $\convexfunction:\R\to\R$,
    \begin{align*}
        \expect_{(\randomrow{\timeindex}, \randomcolumn{\timeindex}) \iid \jointdistribution}[\convexfunction(\estimator'((\randomrow{1}, \randomcolumn{1}), \dots, (\randomrow{\numsamples}, \randomcolumn{\numsamples})))] \leq \expect_{(\randomrow{\timeindex}, \randomcolumn{\timeindex}) \iid \jointdistribution}[\convexfunction(\estimator((\randomrow{1}, \randomcolumn{1}), \dots, (\randomrow{\numsamples}, \randomcolumn{\numsamples})))] \text{.}
    \end{align*}
    We say $\estimator$ \textit{strictly} convex-dominates $\estimator'$ if $\estimator$ convex-dominates $\estimator'$ and there is some $\jointdistribution$ and $\convexfunction$ such that the inequality is strict.
\end{definition}

An estimator which is convex-dominated by \textit{all} other estimators of the same function of the joint distribution is the best estimator for minimizing spread.

\begin{definition}[Convex-Minimal]\label{def:convex-minimal}
    A $\numsamples$-sample unbiased estimator $\estimator^*$ is \textit{convex-minimal} if $\estimator^*$ is convex-dominated by every other $\numsamples$-sample estimator $\estimator$ for the same function. 
\end{definition}
If a convex-minimal unbiased estimator is strictly convex-dominated by every other estimator, then it is the \textit{unique} convex-minimal unbiased estimator for the function of the joint distribution which it estimates.

For the specific choice of $\convexfunction(x) = x^{2}$, \Cref{def:convex-minimal} enforces a minimum-variance unbiased estimator. It is thus a strong requirement for ``minimal spread,'' as it requires the expectation of \textit{any} convex function of a convex-minimal estimator be smaller than that of any other estimator.\footnote{The convex order is similar to second-order stochastic dominance when constrained to random variables with the same mean, but reverses the order of elements and does not require that the convex function $\convexfunction : \mathbb{R} \rightarrow \mathbb{R}$ in the definition be weakly increasing.}

We now seek to understand the form of convex-minimal fixed-sample estimators. We will show that any mutual information with a $\numsamples$-sample estimator has a unique $\numsamples$-sample sample-order-invariant estimator, which is also the unique $\numsamples$-sample convex-minimal estimator. 

We begin by introducing a procedure by which any $\numsamples$-sample estimator $\estimator$ can be converted into a new, sample-order-invariant $\numsamples$-sample estimator $\estimator^{*}$.

\begin{definition}[Permutation variance grinder]\label{def:permutation-variance-grinder}
    Given a $\numsamples$-sample unbiased estimator $\estimator:[\numactions]\times[\numactionscolumn]\to \R$, the \textit{permutation variance grinder} returns a new sample-order-invariant $\numsamples$-sample unbiased estimator $\estimator^{*}$ with the same expectation as $\estimator$ on any joint distribution $\jointdistribution$. Letting $\pi$ be a random variable denoting a permutation of the set of sample indices $[\numsamples]$ drawn uniformly at random, we define the new unbiased estimator as
    \begin{align*}
        \estimator^{*}((\randomrow{1}, \randomcolumn{1}), \dots, (\randomrow{\numsamples}, \randomcolumn{\numsamples})) := \expect_{\indexpermutation}[\estimator((\randomrow{\indexpermutation(1)}, \randomcolumn{\indexpermutation(1)}), \dots, (\randomrow{\indexpermutation(\numsamples)}, \randomcolumn{\indexpermutation(\numsamples)}))] \text{.}
    \end{align*}
\end{definition}

As $\estimator^{*}$ maintains the mean of $\estimator$ for any $\jointdistribution$, it is clear that any mutual information with a fixed-sample estimator in fact has a sample-order-invariant estimator.\footnote{This observation goes the other way, as well: given a sample-order-invariant estimator, we can symmetrically increase and decrease the payments for certain permutations to obtain a new estimator. As a result, every mutual information has either no nontrivial fixed-sample estimators or an infinite class of them.} For completeness, a proof of the following observation is provided in \Cref{app.section3proofs}.

\begin{restatable}{observation}{makenewestimators}\label{obs:make-new-estimators}
    If $\estimator$ is an unbiased estimator for some mutual information $\mutualinformationabstract$, then the function $\estimator^{*}$ produced by the permutation variance grinder is a sample-order-invariant unbiased estimator for $\mutualinformationabstract$.
\end{restatable}

It is a standard result in convex optimization that there is always an optimal solution that respects the symmetries of the problem \citep[e.g.,][Exercise 4.4]{BV-04}.  For estimates of a mutual information, it turns out that this standard construction, given by the permutation variance grinder and from \textit{any} $\numsamples$-sample estimator of a given mutual information, results in the \textit{unique} sample-order-invariant estimator for that mutual information.

\begin{restatable}[Mutual informations have at most one sample-order-invariant estimator]{proposition}{MIsampleorderinvariant}
\label{thm:any-mutual-information-has-at-most-one-sample-order-invariant-estimator}
    For any mutual information $\mutualinformationabstract$ on action space $[\numactions]\times[\numactionscolumn]$ and any fixed number of samples $\numsamples$, there exists at most one sample-order-invariant unbiased estimator of $\mutualinformationabstract$.
\end{restatable}

We will show that the estimator produced by the permutation variance grinder is the unique convex-minimal $\numsamples$-sample estimator for that mutual information. Since the sample-order-invariant $\numsamples$-sample estimator for a given mutual information is unique, the permutation variance grinder must produce the same estimator $\estimator^{*}$ for any two estimators $\estimator$, $\estimator'$ for the same mutual information.

We show that the resulting $\estimator^{*}$ is convex-dominated by any $\estimator$ which produced it, and is strictly convex-dominated if $\estimator$ is sample-order-dependent.

\begin{restatable}[Permutation variance grinder reduces spread]{lemma}{sampleorderinvariance}
\label{lma:sample-order-invariance-implies-convex-minimal}
    The output $\estimator^{*}$ of the permutation variance grinder is convex-dominated by $\estimator$, and is strictly convex-dominated if $\estimator$ is not sample-order-invariant.
\end{restatable}

Combining these three facts shows that the unique sample-order-invariant mutual information produced by the permutation variance grinder is the unique convex-minimal $\numsamples$-sample estimator for a given mutual information. 

\begin{restatable}[Permutation variance grinder returns unique convex-minimal estimator]{theorem}{uniqueconvexminimal}
\label{cor:output-of-permutation-variance-grinder-is-unique-and-convex-minimal}
    Any mutual information $\mutualinformationabstract$ with a $\numsamples$-sample unbiased estimator has a unique convex-minimal $\numsamples$-sample unbiased estimator, produced by the permutation variance grinder.
\end{restatable}

\begin{proof}
    By \Cref{obs:make-new-estimators} and \Cref{thm:any-mutual-information-has-at-most-one-sample-order-invariant-estimator}, the permutation variance grinder returns the only sample-order-invariant $\numsamples$-sample estimator for a given mutual information. By \Cref{lma:sample-order-invariance-implies-convex-minimal}, this unique estimator is strictly convex-dominated by any other (sample-order-dependent) estimator.
\end{proof}

\subsection{Convex-Minimal Estimation of DMI}\label{sec:optimal-estimator-results}

We now implement the results of \Cref{sec:optimal-estimator-tools} to obtain a closed form of the convex-minimal estimator for DMI in the $[2]\times[2]$ action space setting, with $\numsamples \geq 4$. The original estimator given for DMI by \citet{K-24} (see \Cref{def:kong-estimator}) gives variable payments depending on how the samples are divided into tally matrices $\tallymatrix^1$ and $\tallymatrix^2$.  This dependence on sample order results in less consistent payments to agents, particularly as the number of samples increases. This motivates our closed form of the convex-minimal estimator for DMI.

Since the convex-minimal estimator is the unique sample order invariant estimator (\Cref{thm:any-mutual-information-has-at-most-one-sample-order-invariant-estimator}), we can focus on estimators defined on tally matrices of the form $\estimator:\tallyset{\numsamples}{\numactions}{\numactionscolumn}\to \R$, rather than those defined on sequences of samples.

\begin{restatable}[$\numsamples$-sample convex-minimal DMI estimator on binary alphabet]{theorem}{convexminimalbinary}
\label{thm:convex-minimal-dmi-estimator-with-binary-alphabet}
    Fix alphabet sizes $\numactions = \numactionscolumn = 2$ and number of samples $\numsamples \geq 4$. Then, the convex-minimal $\numsamples$-sample unbiased estimator $\estimator$ of DMI, i.e. of $16 \det(\jointdistribution)^{2}$, is:
    \begin{align*}
        \estimator(\tallymatrix) = \frac{16 (\numsamples - 4)!}{\numsamples!} ( \tallymatrixentry{1}{1} \tallymatrixentry{2}{2} ((\tallymatrixentry{1}{1} - 1) (\tallymatrixentry{2}{2} - 1) - \tallymatrixentry{2}{1} \tallymatrixentry{1}{2})
     + \tallymatrixentry{2}{1} \tallymatrixentry{1}{2} ((\tallymatrixentry{2}{1} - 1) (\tallymatrixentry{1}{2} - 1) - \tallymatrixentry{1}{1} \tallymatrixentry{2}{2}) ) \text{,}
    \end{align*}
    for any tally matrix $\tallymatrix \in \tallyset{\numsamples}{2}{2}$.
\end{restatable}

The expression for the estimator $\estimator$ in \Cref{thm:convex-minimal-dmi-estimator-with-binary-alphabet} can also be written as
\begin{align*}
    \estimator(\tallymatrix) = \frac{(\numsamples - 4)! \numsamples^{4}}{\numsamples!} \left( 16 \det(\widehat{\jointdistribution})^{2} + \frac{16}{\numsamples^{4}}(\tallymatrixentry{1}{1} \tallymatrixentry{2}{2} (1 - \tallymatrixentry{1}{1} - \tallymatrixentry{2}{2}) + \tallymatrixentry{1}{2} \tallymatrixentry{2}{1} (1 - \tallymatrixentry{1}{2} - \tallymatrixentry{2}{1})) \right) \text{,}
\end{align*}
i.e., as the plug-in estimator for DMI (DMI calculated with respect to the empirical joint distribution $\widehat{\jointdistribution}$), but rescaled and shifted to make it unbiased.\footnote{The plug-in estimator of DMI is a consistent estimator, but is not unbiased. It is a common practice in statistics to produce consistent \textit{and} unbiased estimators of statistical quantities by rescaling and shifting the plug-in estimator. For other examples of estimators produced in this manner, consider the sample covariance, or the unbiased estimator for the maximum of a uniform distribution over any finite range of integers (i.e., the German tank problem).}

To provide an intuitive picture of kinds of patterns of play that the estimator from \Cref{thm:convex-minimal-dmi-estimator-with-binary-alphabet} rewards and punishes, we examine the payoffs it provides for the minimum $\numsamples = 4$ samples required. These payoffs can be implemented using a lookup table, described in \Cref{cor:four-sample-convex-minimal-dmi-estimator-with-binary-alphabet}.

\begin{restatable}[$4$-sample convex-minimal DMI estimator on binary alphabet]{corollary}{foursample}
\label{cor:four-sample-convex-minimal-dmi-estimator-with-binary-alphabet}
    In a $2 \times 2$ action space with $\numsamples = 4$ samples, the convex-minimal unbiased estimator of the determinant mutual information $16 \det(\jointdistribution)^{2}$ is:
    \begin{align*}
        \estimator \left( \begin{bmatrix}
            2 & 0 \\
            0 & 2
        \end{bmatrix} \right)
        =
        \estimator \left( \begin{bmatrix}
            0 & 2 \\
            2 & 0
        \end{bmatrix} \right)
        =
        \frac{8}{3}, && 
        \estimator \left( \begin{bmatrix}
            1 & 1 \\
            1 & 1
        \end{bmatrix} \right)
        =
        -\frac{4}{3}, \text{ and }&&
        \estimator(\text{anything else}) &= 0 \text{.}
    \end{align*}
\end{restatable}

As the lookup table demonstrates, high positive rewards are provided to patterns of play that appear highly correlated or anticorrelated, while negative rewards are provided when play looks maximally uncorrelated.

To characterize the convergence rate of the estimator from \Cref{thm:convex-minimal-dmi-estimator-with-binary-alphabet}, we provide a closed form for its variance in \Cref{thm:convex-minimal-dmi-estimator-with-binary-alphabet-variance}, presented in \Cref{app:convergence-rate}.

\section{Ex-ante bounded-sample Mutual Information Estimators}\label{sec:ex-ante-bounded-sample-mutual-information-estimators}

Thus far, we have required that an estimator take \textit{exactly} $\numsamples$ samples. This section relaxes that assumption, and proposes several mutual informations which rely on unbiased estimators which take a \textit{variable} number of samples.

We begin by offering a definition of \textit{ex-ante bounded-sample estimators}, which take at most $\numsamples$ samples in expectation. We then use this definition to reduce the number of samples required by DMI in \Cref{sec:stop-short-estimators}, using \textit{stop-short estimators}. Furthermore, we construct estimators for DMI in the $2 \times 2$ setting with strictly lower variance than the standard $4$-sample estimator, while maintaining the sample complexity \textit{in expectation}. Next, \Cref{sec:scoring-rule-based} introduces a new class of mutual informations based on the broad existing literature of \textit{proper scoring rules} and their connections to information theory (see \Cref{sec:related-works}). We show that DMI is not scoring-rule-based, and prove that no scoring-rule-based mutual information has a finite-sample estimator. \Cref{sec:collision-estimators} proposes two new ex-ante bounded-sample estimators for the quadratic score (\Cref{def:quadratic-score}) based mutual information. Finally, \Cref{sec:arbitrarily-few-expected-sample-estimators} considers the possibility and consequences of arbitrarily scaling down the expected number of samples. We conclude that reducing the expected number of samples directly increases the variance of expected payments. All omitted proofs and additional facts for this section are provided in \Cref{app:ex-ante-lemmas}.

We now formally define an ex-ante-bounded sample estimator, which takes a variable number of samples according to some stopping time rule. We contrast this definition to that of fixed-sample estimators in \Cref{sec:fixed-sample-mutual-information-estimators}. This definition generalizes the finite-sample unbiased estimator of \Cref{def:unbiased-estimator} to require $\numsamples$ samples only \emph{in expectation}.
\begin{definition}[Ex-Ante Bounded-Sample Estimator]\label{def:ex-ante-bounded-sample-estimator}
    An ex-ante $k$-sample unbiased estimator for a function $\mutualinformationabstract:\Delta([\numactions]\times [\numactionscolumn])\to \mathbb{R}$ comprises a \emph{stopping rule} $\stoppingtime(\randomsample{})$ and an \emph{estimator function} $\estimator:([\numactions]\times [\numactionscolumn])^*\to \R$ defined on any number of samples such that for any $\jointdistribution\in \Delta([\numactions]\times [\numactionscolumn])$,
    \begin{align*}
        \E_{\randomsample{1}, \dots \overset{\textrm{i.i.d.}}{\sim} F}[\estimator(\randomsample{1}, \dots, \randomsample{\stoppingtime(\randomsample{})})] = \mutualinformationabstract(\jointdistribution) \text{,} \text{ and }
        \E_{\randomsample{1}, \dots \overset{\textrm{i.i.d.}}{\sim} F}[\stoppingtime(\randomsample{})] = k \text{,}
    \end{align*}
    where $\stoppingtime(\randomsample{})$ denotes the number of samples taken before the stopping rule is satisfied and a single sample $\randomsample{\timeindex} = (\randomrow{\timeindex}, \randomcolumn{\timeindex})$.
\end{definition}
This definition replaces the fixed sample number parameter $\numsamples$ of \Cref{def:unbiased-estimator} with a random variable stopping time which depends on the observed samples. It then uses an estimator function $\estimator$ defined on arbitrary numbers of samples to choose a reward. The expectation of $\estimator$ over all sequences of samples and implied stopping times must be equal to the value of the function $\mutualinformationabstract$ for which $\estimator$ is an unbiased estimator.

Note that any $k$-sample mutual information estimator is also an ex-ante $k$-sample estimator for the same mutual information, with stopping rule $\stoppingtime = k$.

There is an important use case for ex-ante bounded sample estimators: it is sometimes possible to take a fixed-sample unbiased estimator of a mutual information, then convert it into an ex-ante bounded sample estimator with the same expected number of samples, but with a lower variance. We call these kinds of estimators \textit{stop-short estimators}.

\subsection{Stop-short Estimators}\label{sec:stop-short-estimators}

Sometimes, for a $\numsamples$-sample estimator, we can tell what value the of the estimator must be before observing all $\numsamples$ samples. In this case we can \textit{stop-short}, reducing the number of samples used in expectation by choosing specific sequences of samples where we do not observe all $\numsamples$ samples. This subsection explores sequences of samples for which the convex-minimal fixed-sample DMI estimator can stop early, as well as ways to leverage these techniques to design more desirable mutual information estimators.

\begin{example}
    Consider the $4$-sample DMI unbiased estimator given in \Cref{cor:four-sample-convex-minimal-dmi-estimator-with-binary-alphabet} for the $2 \times 2$ action space which sees samples $(1, 1), (1, 1), (1, 1)$. For any fourth sample the resulting tally matrix $\tallymatrix$ has either a 3 or 4 in coordinate $(1, 1)$ and all such tally matrices have estimator value 0. Therefore, when the observed sequence begins with three identical report pairs, we can always stop-short and determine the agents' payoffs without asking for a fourth sample.

    Some entry in $\jointdistribution$ must always be at least $1 / 4$, so by stopping short whenever we encounter the same pair three times in a row we can reduce the expected sample complexity to at most $4 - (1 / 4)^{3}$.
\end{example}

By relaxing our requirement that an estimator take exactly a fixed $\numsamples$ samples, we beat the triviality result of \Cref{thm:three-or-fewer-samples-implies-trivial-mutual-information}, which states that any fixed $\numsamples$-sample estimator for $\numsamples < 4$ must be the trivial estimator.

It is also possible to keep the expected number of samples the same and instead decrease estimator variance with stop-short estimators. This can be accomplished by mixing over stop-short estimators. \Cref{thm:stop-short-estimator} constructs such a stop-short estimator, which maintains an expected $4$ samples while reducing the variance of payments below that of DMI.

\begin{restatable}[Stop-short estimator reduces variance]{theorem}{stopshortestimator}
\label{thm:stop-short-estimator}
    Let the row and column alphabet sizes be $\numactions = \numactionscolumn = 2$. Consider the unbiased estimator that, with probability $1/4$ runs the convex-minimal $4$-sample DMI estimator, and with probability $3/4$ runs the convex-minimal $5$-sample DMI estimator, and in either case stops short once the estimator has collected enough samples to exactly determine its output. This mixed estimator of DMI has:
    \begin{enumerate}
        \item at most $4$ samples in expectation (and this bound is tight over different possible mixing probabilities),
        \item a pointwise variance below that of the convex-minimal $4$-sample DMI estimator, and
        \item a worst-case variance strictly below that of the convex-minimal $4$-sample DMI estimator.
    \end{enumerate}
\end{restatable}

\begin{proof}
    It can be shown that the $4$-sample stop-short estimator requires $3.75$ samples in expectation in the worst case, which occurs when the joint distribution $\jointdistribution = \begin{pmatrix}
        1/2 & 0 \\
        0 & 1/2
    \end{pmatrix}$.

    Similarly, it can be shown that the $5$-sample stop-short estimator requires $4.75$ samples in expectation in the worst case, which also occurs when the joint distribution $\jointdistribution = \begin{pmatrix}
        1/2 & 0 \\
        0 & 1/2
    \end{pmatrix}$.

    Mixing over these with probability $3/4$ and $1/4$ respectively gives a worst-case expected number of samples of $3.75 * \frac{3}{4} + 4.75 * \frac{1}{4} = 4$.

    This estimator has, for a fixed joint distribution, a lower variance than the fixed $4$-sample DMI estimator, because its variance is a convex combination of the variance of the fixed $4$-sample DMI estimator and the fixed $5$-sample DMI estimator, and the fixed $5$-sample DMI estimator has a lower variance than the fixed $4$-sample DMI estimator.
\end{proof}

The stop short estimators constructed above randomize over a fixed set of finite sample estimators.  Thus, they have an ex-post finite bound on the number of samples required. Ex-ante bounded-sample estimators in general need not have an upper cap on the ex-post number of samples used. This flexibility allows us to consider constructing estimators for a much broader class of mutual informations, constructed from \emph{proper scoring rules}.

\subsection{Scoring-Rule-Based Mutual Informations}\label{sec:scoring-rule-based}

We now construct ex ante bounded mutual information estimators from proper scoring rules: reward functions that incentivize truthful predictions~\citep{GR-07}.

\begin{definition}[Proper Scoring Rule]\label{def:scoring-rule}
    A proper scoring rule over some feasible outcome space $\mathcal{O}$ is a function $\scoringrule:\Delta(\mathcal{O})\times \mathcal{O}\to \mathbb{R}$, such that for any belief $p\in \Delta(\mathcal{O})$,
    \begin{align*}
        \E_{o\sim p}[\scoringrule(p, o)]&\geq \E_{o\sim p}[\scoringrule(r, o)], ~~~\forall r\in \Delta(\mathcal{O}).
    \end{align*}
    A scoring rule is strict if the inequality is strict for all $r\neq p$.
\end{definition}

A scoring rule assigns a payout to a reporting agent based on their report and an observed outcome. A proper scoring rule is one in which an agent maximizes their expected reward by reporting their true belief. Scoring rules implicitly define a \emph{value of information}: receiving a signal about the likely outcome can improve an agent's expected score.

In our setting, we consider the row player's ability to predict the column player's action. Given a joint distribution $\jointdistribution$, a sample $(\randomrow{},\randomcolumn{})$ is drawn. The row player observes $\randomrow{}$, and we imagine that they are asked to predict $\randomcolumn{}$. We define the value of observing the random variable $\randomrow{}$ as the row player's expected improvement in $\scoringrule$ over predicting the marginal distribution of the column player's action without observing $\randomrow{}$. It will turn out that this improvement can sometimes be estimated even with only sample reports and not prediction reports.

Let $p$ denote the marginal distribution of the column player's sample in $\jointdistribution$, with the $j$th entry $\sum_{i=1}^\numactions \jointdistributionentry{i}{j}$. Similarly, denote the posterior after observing the row $\randomrow{}$ by $p_{\randomrow{}}$, where the $j$th entry is $\jointdistributionentry{\randomrow{}}{j} / \sum_{k = 1}^\numactionscolumn \jointdistributionentry{\randomrow{}}{k}$. 

\begin{definition}[Value of Information]\label{def:value-of-info}
    The value of information relative to a scoring rule $\scoringrule:\Delta([\numactionscolumn])\times [\numactionscolumn]\to \mathbb{R}$ of joint distribution $\jointdistribution$ is,
    \begin{align*}
        \VoI{\jointdistribution}{\scoringrule}&= \expect_{(\randomrow{}, \randomcolumn{})\sim F}[\scoringrule(p_{\randomrow{}}, \randomcolumn{})] - \expect_{(\randomrow{}, \randomcolumn{})\sim F}[\scoringrule(p, \randomcolumn{})].
    \end{align*}
\end{definition}

We show that for a proper scoring rule $R$, $\VoI{\jointdistribution}{R}$ can be used to construct a row mutual information (\Cref{def:mutual-information}).

\begin{theorem}\label{lemma:scoring-rule-based-MI}
    For a proper scoring rule $\scoringrule:\Delta([\numactionscolumn])\times [\numactionscolumn] \to \mathbb{R}$, the function $\mutualinformationabstract(\jointdistribution) = \VoI{\jointdistribution}{\scoringrule}$ is a row mutual information.
\end{theorem}

\begin{proof}
    $G$ satisfies zero on independent play because, in this case, $p=p_i$ for every row $i$ with positive marginal probability and $\VoI{\jointdistribution}{\scoringrule} = 0$ trivially. $G$ satisfies the row data processing inequality by Blackwell monotonicity, i.e., a garbling always provides less expected utility than the original signal in a decision problem such as a proper scoring rule.
\end{proof}

$\VoI{\jointdistribution}{\scoringrule}$ generally will not also satisfy the column data processing inequality for arbitrary scoring rules $\scoringrule$.\footnote{In particular, the quadratic score (\Cref{def:quadratic-score}) does not satisfy the column data processing inequality.}  Thus, to use such mutual informations to construct peer prediction mechanisms, the agents need to be rewarded with distinct functions $G$ and $H$.

Throughout \Cref{sec:fixed-sample-mutual-information-estimators}, we used DMI as an example of a mutual information with an unbiased estimator. One natural question is whether DMI belongs to the class of scoring-rule-based mutual informations. We find that this is not the case. More generally, no scoring-rule-based mutual information has a fixed-sample estimator, meaning that scoring-rule-based mutual informations are in an entirely different class of mutual informations from DMI.

\begin{restatable}{theorem}{fixedsamplescoringruletrivial}\label{thm:no-fixed-sample-scoring-rule-mi}
    No non-trivial scoring-rule-based row mutual information has a fixed-sample unbiased estimator.
\end{restatable}

The proof of \Cref{thm:no-fixed-sample-scoring-rule-mi} proceeds in two parts. First, we use the polynomial characterization of unbiased estimators (\Cref{lma:mutual-informations-are-polynomials}) to show that any scoring rule that produces a scoring-rule-based mutual information with a fixed-sample unbiased estimator must have an affine expectation in the prior. We then observe that the value of information for a scoring rule with affine expectation must always be zero, implying that no nontrivial scoring-rule-based mutual information has a fixed-sample unbiased estimator. 

Scoring-rule-based estimators have two significant advantages. 
Firstly, they can be \emph{strict}: informally, losses from garbling are strict whenever Shannon mutual information has strict losses from garbling.  
Secondly, there are simple scoring-rule-based estimators with a small number of samples in expectation.

\begin{definition}[Strict mutual information]\label{def:strict-MI}
    Given an information structure $\jointdistribution$, let $(\randomrow{},\randomcolumn{})\sim \jointdistribution$, and let $\randomrow{}'$ be obtained by applying the row garbling $S$ to $\randomrow{}$. The garbling $S$ is {\em strict} for $\jointdistribution$ if the posterior distribution of $\randomcolumn{}$ conditional on $\randomrow{}$ differs with positive probability from the posterior distribution of $\randomcolumn{}$ conditional on $\randomrow{}'$. A row mutual information is {\em strict} if the data processing inequality holds strictly for any strict row garbling.
\end{definition}

Beyond the $2 \times 2$ setting $\DMI$ is not generally strict: the matrix must be of full rank for $\DMI$ to be positive. In essence, if the correlation within $\jointdistribution$ is not across all actions, $\DMI$ does not detect it and will behave as the trivial mutual information. For incentivizing agents to report truthfully, then, a mutual information which is positive for \emph{any} correlation is desirable, and not difficult to achieve with scoring rules.

\begin{restatable}{proposition}{scoringMIstrict}
\label{prop:scoring-MI-strict} 
    For any strictly proper scoring rule $\scoringrule$, the value of information $\VoI{\cdot}{\scoringrule}$ is a strict mutual information.
\end{restatable}

A number of common scoring rules are strictly proper, such as the quadratic score considered in \Cref{sec:collision-estimators} and the log score (discussed in \Cref{app:entropy}). Recall, the value of information corresponding to the log scoring rule is identically the Shannon mutual information.

\subsection{Collision Estimators}\label{sec:collision-estimators}

We construct an ex-ante bounded-sample estimator for the quadratic score-based mutual information. This gives a concrete example of a strict mutual information, with an estimator that requires a reasonable number of samples in expectation.

\begin{definition}[Quadratic Score \citep{B-50}]\label{def:quadratic-score}
    The quadratic, or Brier, score is given by 
    $\quadscore(p, o) = 1 - \sum_{j} (\indicate{j=o} - p(j))^2$
    where $p(o)$ is the probability $p$ places on outcome $o$.
\end{definition}
The quadratic score produces the quadratic score-based mutual information, $\VoI{\jointdistribution}{\quadscore} = \E_{\randomrow{}\sim \jointdistribution}[||p_{\randomrow{}}||^2_2] - ||p||^2_2$, where $p$ is the marginal distribution of $\randomcolumn{}$ and $p_{\randomrow{}}$ is the distribution of $\randomcolumn{}$ conditioned on having seen $\randomrow{}$. 

We construct an ex-ante ($\numactions+1$)-sample asymmetric mutual information estimator for the quadratic score-based mutual information, which we call a \emph{collision estimator}.

\begin{definition}[Quadratic Collision Estimator]\label{def:collision-estimator}
  The quadratic collision estimator is:
  \begin{enumerate}
        \item draw sample $(I_1,J_1)$,
        \item draw samples until $I_{K+1} = I_1$ for some $K+1$ (this is $K$ new samples), and
        \item output $\indicate{J_1 = J_{K+1}} - \indicate{J_1 = J_{2}}$.
  \end{enumerate}
\end{definition}

Intuitively, the Quadratic Collision Estimator estimates the two terms $\E_{\randomrow{}\sim \jointdistribution}[||p_{\randomrow{}}||^2_2]$ and $-||p||^2_2$ and then adds them together. We estimate the latter term, $-||p||^2_2$, by returning  $-1$ if two column samples are the same, unconditioned on the values of the associated rows. We estimate $\E_{\randomrow{}\sim \jointdistribution}[||p_{\randomrow{}}||^2_2]$ by drawing some row ($\randomrow{1})$ and then waiting for two column samples from the same row report ($\randomrow{1}$ and $\randomrow{K+1}$). Returning 1 if the associated row reports are the same will estimate $\E_{\randomrow{}\sim \jointdistribution}[||p_{\randomrow{}}||^2_2]$.

\begin{restatable}{proposition}{collisionestimator}
\label{prop:collision-estimator-is-quad-VoI}
    For a joint distribution $\jointdistribution\in \Delta([\numactions] \times [\numactionscolumn])$, the collision estimator is an ex-ante  $(\numactions+1)$-sample unbiased estimator for $\VoI{\jointdistribution}{\quadscore}$.
\end{restatable}

While the collision estimator requires a reasonable number of samples to be drawn in expectation, there is an estimator for the same MI which requires even fewer samples. The main idea is that we can look for a collision in the first two draws and then correct for the underestimate due to the unlikeliness of such a collision. This can be done by estimating $1/\Pr[i = I]$ as the expected value of a geometric random variable with probability $\Pr[i=I]$, i.e., by drawing new samples until we see $i=I$. This adjustment only needs to be calculated when there is a collision in the first two samples, which happens with probability $\Pr[i=I]$.

\begin{definition}[Fast Quadratic Collision Estimator]\label{def:fast-collision-estimator}
    The fast quadratic collision estimator is:
    \begin{enumerate}
        \item draw two samples $(I_1,J_1)$ and $(I_2,J_2)$,
        \item if they do not match ($I_1 \neq I_2$ or $J_1 \neq J_2$), stop and output $-\indicate{J_1 = J_2}$, and
        \item if they match, draw samples until $I_{K+2} = I_1$ for some $K+2$ (this is $K$ new samples), and
        output $K - \indicate{J_1 = J_2} = K-1$.
    \end{enumerate}
\end{definition}

\begin{restatable}{proposition}{bettercollision}
\label{prop:better-collision}
  For a joint distribution $\jointdistribution \in \Delta([\numactions] \times [\numactionscolumn])$, the fast collision estimator is an ex-ante $3$-sample unbiased estimator for $\VoI{\jointdistribution}{\quadscore}$.
\end{restatable}

We note that if $K$ is the random number of samples drawn by the collision estimator and $K'$ the fast collision estimator, then $K$ stochastically dominates $K'-1$: both estimators draw an initial sample, then the latter draws another sample, and then both estimators draw until some sample where $I_k$ matches $I_1$.  However, the latter estimator sometimes omits this draw-until-match stage, while the former estimator always undertakes it.

\paragraph{Validity of consistent strategy assumption} 
Throughout this paper, we have assumed that agents follow \textit{consistent} strategies: a single strategy is applied independently to each sample.

It is possible that an agent may benefit from correlated strategic behavior across several samples. This is particularly true in the case of ex-ante bounded-sample estimators where the payment directly scales with the number of samples taken, as in the \Cref{def:fast-collision-estimator}.

\begin{example}
    Consider a row player being paid according to the fast quadratic collision estimator, who has reported $\randomrow{1} = \randomrow{2}$, and is asked to report $\randomrow{3}$. She knows she will continue being asked for reports until she makes some report $\randomrow{K+2}$ such that $\randomrow{K+2} = \randomrow{1}$, at which point she will be paid $K - 1$. She can increase her payment on this instance arbitrarily with her willingness to sit through more questions, reporting $\randomrow{i} \neq \randomrow{1}$ for all $i \geq 3$.
\end{example}

This problem is not present in the quadratic collision estimator of \Cref{def:collision-estimator}, as the number of samples taken does not affect the payment, but it may still admit other correlated strategies across samples. A recent note \citep{F-26} finds that for DMI, truthful strategies are still optimal even when correlated strategies across all $\numsamples$ samples may be considered. Understanding estimators, both fixed-sample and ex-ante bounded-sample, that are robust to these more complex strategies may be of practical interest for future investigation.

\subsection{Arbitrarily Few-Expected Sample Estimators}\label{sec:arbitrarily-few-expected-sample-estimators}

We can lower the expected number of samples required for the collision estimator -- or any estimator -- to be arbitrarily close to $1$, at the cost of increasing the ex-post payments made. We define a general class of scaled estimators which lower the expected number of samples by trading off the variance of the payments made.

\begin{definition}[$(\alpha, k)$-Scaled Mutual Information Estimator]\label{def:scaled-estimator}
    Given an ex-ante $k$-sample mutual information unbiased estimator with stopping rule $\stoppingtime$ and payment rule $\estimator$ and $\alpha\in (0,1]$, an $(\alpha, k)$-scaled unbiased estimator $\estimator_{\alpha}$ either takes one sample and pays $\estimator_{\alpha} = 0$ with probability $\alpha$ or takes $\stoppingtime$ samples and pays $\estimator_{\alpha}(\randomsample{1}, \dots, \randomsample{\stoppingtime{\randomsample{}}}) =\frac{1}{\alpha} \estimator(\randomsample{1}, \dots, \randomsample{\stoppingtime{\randomsample{}}})$ with probability $1-\alpha$.
\end{definition}

The expected number of samples taken by a $(\alpha, k)$-Scaled Mutual Information Estimator will be $1+\alpha(k-1)$. We note that the expected value of the scaled estimator is the same as the expected value of the starting estimator on the same $\jointdistribution$:
\begin{align*}
    \expect_{\randomsample{1}, \dots,\randomsample{\stoppingtime} \sim \jointdistribution}[\estimator_{\alpha}(\randomsample{1}, \dots, \randomsample{\stoppingtime{\randomsample{}}})] = \alpha \expect_{\randomsample{1}, \dots,\randomsample{\stoppingtime} \sim \jointdistribution}[\frac{1}{\alpha}\estimator(\randomsample{1}, \dots, \randomsample{\stoppingtime{\randomsample{}}})] = \expect_{\randomsample{1}, \dots,\randomsample{\stoppingtime} \sim \jointdistribution}[\estimator(\randomsample{1}, \dots, \randomsample{\stoppingtime{\randomsample{}}})] \text{.}
\end{align*} 

Using this approach, \emph{any} mutual information with \emph{any} estimator can be made to take, in expectation, as few samples as desired. This leads to some uncertainty about what makes a ``good'' mutual information estimator, especially with ex-ante bounded samples. If all mutual informations can be made to take the same number of samples, how do we evaluate the relative quality of mutual informations? One possible answer follows \Cref{sec:fixed-sample-mutual-information-estimators}, in seeking lower-variance mutual informations.

\begin{restatable}{proposition}{scalingincreasesvariance}
\label{prop:scaling-increases-variance}
    Let $R$ be the random variable denoting the payment of an unbiased estimator $\estimator$. Then, the variance of the payment of the scaled unbiased estimator $\estimator_{\alpha}$ is:
    \begin{align*}
        \tfrac{1-\alpha}{\alpha}\E[R^2] + \Var(R) \text{.}
    \end{align*}
\end{restatable}

We propose one possible criterion of minimizing variance for a fixed expected number of samples, but leave the problem of comparing ex-ante bounded-sample mutual information quality open both technically and philosophically.

\section{Conclusions and Future Work}\label{sec:conclusions-and-future-work}

Mutual information unbiased estimators are a powerful peer prediction mechanism. Our work aims to understand the space of mutual informations that admit finite-sample unbiased estimators. In particular, \Cref{sec:fixed-sample-mutual-information-estimators} proves several uniqueness theorems about the Determinant Mutual Information of \citet{K-24}, and characterize its sample complexity for small action spaces. 

In \Cref{sec:ex-ante-bounded-sample-mutual-information-estimators}, we relax $\numsamples$-sample unbiased estimators to take $\numsamples$ samples only in \textit{expectation}. We design stop-short estimators for DMI which leverage this relaxed assumption to reduce the variance in payments of estimators, and formalize connections between scoring rules as a tool for peer prediction and mutual informations. We then present two new unbiased estimators for the quadratic score-based mutual information.

There remain several lines of inquiry to explore mutual information unbiased estimators, some of which have been addressed by recent follow-up work in \citet{K-26}. We left unresolved whether DMI was the unique mutual information with an unbiased estimator on an action space of size $\numactions > 2$. \citet{K-26} answers this, finding that for $\numsamples\in \{2\numactions, 2\numactions+1\}$ samples, DMI is the unique mutual information with an unbiased estimator using $\numsamples$ samples, and for $\numsamples<2\numactions$ there is no nontrivial mutual information with a fixed-sample estimator. \citet{K-26} also finds that there are no mutual informations with fixed-sample estimators for settings where $\numactions \neq \numactionscolumn$, implying that for such settings we \textit{must} use a ex-ante bounded estimator such as those considered in \Cref{sec:ex-ante-bounded-sample-mutual-information-estimators}.

Finally, \Cref{cor:output-of-permutation-variance-grinder-is-unique-and-convex-minimal} allows us to easily find the convex-minimal fixed-sample estimator of a mutual information, but no analogue is yet known for ex-ante bounded-sample estimators. While \Cref{thm:stop-short-estimator} shows that ex-ante bounded-sample estimators can have strictly lower variance than any fixed-sample estimator, there may be even lower variance ex-ante bounded-sample estimators.

\bibliographystyle{ACM-Reference-Format}
\bibliography{bibliography}

\clearpage

\appendix
\crefalias{section}{appendix}
\crefalias{subsection}{subappendix}

\section{Omitted Proof from \Cref{sec:model}}\label{app:prelims}

\stochasticdiagonal*

\begin{proof}
    We examine the case where $\numactions \geq \numactionscolumn$. Cases where $\numactions \leq \numactionscolumn$ can be shown as a corollary of the cases where $\numactions \geq \numactionscolumn$ by considering the transpose of $\jointdistribution$ and following the same proof steps.

    Define the diagonal entries of $\diagonalmatrix$ as
    \begin{align*}
        \diagonalmatrix_{j, j} = \sum_{i = 1}^{\numactions} \jointdistributionentry{i}{j}
    \end{align*}
    and the off-diagonal entries of $\diagonalmatrix$ as $0$. Then, depending on whether $\diagonalmatrix_{j, j}$ is zero, we construct each the $j$th column of $S$ as
    \begin{align*}
        S_{i, j} = \jointdistributionentry{i}{j} / \diagonalmatrix_{j, j}
    \end{align*}
    if $\diagonalmatrix_{j, j} \neq 0$, and instead make the $j$th column of $S$ an arbitrary probability vector if $\diagonalmatrix_{j, j} = 0$.

    The resulting matrices $S$ and $\diagonalmatrix$ are column-stochastic and diagonal, respectively, and the entries of $\diagonalmatrix$ are non-negative and sum to $1$, thus, $\diagonalmatrix \in \Delta([\numactionscolumn] \times [\numactionscolumn])$. Note that if any of the columns of $\jointdistribution$ are zero, then the choice of $S$ is non-unique by construction.
\end{proof}

\section{Omitted Proofs from \Cref{sec:fixed-sample-mutual-information-estimators}}\label{app.section3proofs}

\subsection{Omitted Proofs from \Cref{sec:determinant-mutual-information-tools} (Polynomial Characterization)}\label{app:polynomial-characterization}

\MIpolynomials*

\begin{proof}
    Suppose first that $\mutualinformationabstract$ has a $\numsamples$-sample unbiased estimator $\estimator$. Let
    \begin{align*}
        (\randomrow{1},\randomcolumn{1}),\ldots,(\randomrow{\numsamples},\randomcolumn{\numsamples})
    \end{align*}
    be $\numsamples$ independent samples from $\jointdistribution$. By unbiasedness,
    \begin{align*}
        \mutualinformationabstract(\jointdistribution)
        = \expect\left[\estimator\big((\randomrow{1},\randomcolumn{1}),\ldots,(\randomrow{\numsamples},\randomcolumn{\numsamples})\big)\right].
    \end{align*}
    Expanding this expectation over all possible sample sequences gives
    \begin{align*}
        \mutualinformationabstract(\jointdistribution)
        = \sum_{i_1=1}^{\numactions}\sum_{j_1=1}^{\numactionscolumn}\cdots
        \sum_{i_{\numsamples}=1}^{\numactions}\sum_{j_{\numsamples}=1}^{\numactionscolumn}
        \estimator\big((i_1,j_1),\ldots,(i_{\numsamples},j_{\numsamples})\big)
        \prod_{r=1}^{\numsamples}\jointdistributionentry{i_r}{j_r}.
    \end{align*}
    For each sample sequence, the corresponding value of $\estimator$ is a real number independent of $\jointdistribution$. Thus, this expression is a polynomial of degree at most $\numsamples$ in the entries of $\jointdistribution$.

    Conversely, suppose that $\mutualinformationabstract(\jointdistribution)$ is a polynomial of degree at most $\numsamples$. Let $G_\ell$ be its $\ell$th monomial, with coefficient $a_\ell$. If
    \begin{align*}
        G_\ell(\jointdistribution)
        = a_\ell \prod_{r=1}^{d}\jointdistributionentry{i_r}{j_r},
        \qquad d\leq \numsamples,
    \end{align*}
    define
    \begin{align*}
        \estimator_\ell
        = a_\ell \prod_{r=1}^{d}
        \indicate{(\randomrow{r},\randomcolumn{r})=(i_r,j_r)}.
    \end{align*}
    By independence,
    \begin{align*}
        \expect[\estimator_\ell]
        = a_\ell \prod_{r=1}^{d}\jointdistributionentry{i_r}{j_r}
        = G_\ell(\jointdistribution).
    \end{align*}
    Therefore, summing $\estimator_\ell$ over all monomials of $\mutualinformationabstract$ and applying linearity of expectation gives an unbiased $\numsamples$-sample estimator of $\mutualinformationabstract$.
\end{proof}

\independentnullstellensatz*

Our proof of \Cref{thm:independent-characterization} relies on heavy algebraic machinery, including Hilbert's nullstellensatz theorem \citep{H-1893}. 
We first introduce the notation and concepts necessary for the proof of \Cref{thm:independent-characterization}.

\begin{definition}[ideal; \citet{S-25}]
     Let $R$ be a (commutative) ring. A non-empty subset $\ideal \subset R$ is an ideal if
     \begin{enumerate}
         \item $a, b \in \ideal \implies a + b \in \ideal$, and
         \item $r \in R, a \in \ideal \implies ra \in \ideal$
     \end{enumerate}
    hold.
\end{definition}

\begin{definition}[vanishing locus, vanishing ideal; \citet{S-25}]
     For any ideal $\ideal \subset K[x_1,\dots , x_n]$ its vanishing locus $\vanishinglocus{\ideal}$ is 
    $$ \vanishinglocus{\ideal} = \{a \in K^n ~ | ~  f(a) = 0 ~\forall f \in \ideal\}.$$
    Conversely, for $A \subset K^n$ an arbitrary subset, the vanishing ideal is
    $$\vanishingideal{A} = \{ f \in K[x_1, \dots, x_n] ~ | ~ f (a) = 0 ~\forall a \in A\}.$$    
\end{definition}

\begin{definition}[radical; \citet{S-25}]
    Let $R$ be a ring and $\ideal \subset R$ an ideal.
    The radical of $\ideal$ is the ideal
    $$\rad(J) = \{ f \in R ~ | ~\exists n \in \mathbb{N} \text{ such that }f^n \in J \}.$$
\end{definition}

\begin{theorem}[strong Hilbert’s Nullstellensatz; \citet{S-25}]\label{thm:nullstellensatz}
    Let $K$ be an algebraically closed field and let $\ideal \subset K[x_1,\dots , x_n]$ be an ideal. 
    Then
    $$\vanishingideal{\vanishinglocus{\ideal}} = \rad(\ideal).$$
\end{theorem}

\begin{proof}[Proof of \Cref{thm:independent-characterization}]
Note that any estimator must be expressible as a polynomial. 

To obtain an algebraically closed field, we consider polynomials with complex coefficients and generalize the entries of $\jointdistribution$ to be possibly complex. 
Call these variables $X = \{x_{ij} ~|~ i\in [\numactions], j\in [\numactionscolumn]\}\in \mathbb{C}^{\numactions\times \numactionscolumn}$. 
Thus, the set of possible mutual information estimators is contained in the space $\mathbb{C}[X]$, the polynomial ring over $X$ with complex coefficients. 

Observe that a joint distribution $\jointdistribution$ is generated from independent play if and only if the determinant of every $2\times 2$ minor of $F$ is zero. 

Consider the ideal generated by the determinants of $2\times 2$ minors of $\jointdistribution$, as mapped into $X$: $\ideal = \langle (x_{ab}x_{cd} - x_{ad}x_{bc}) ~|~ a\neq c, b\neq d\rangle$. 
Every element of this ideal is zero when the elements of $X$ equate to a rank-1 real-valued joint distribution $\jointdistribution$, and some element is nonzero when $X$ equates to a higher-rank matrix $\jointdistribution$. 
Thus, $\{X\in \mathbb{R}^{\numactions\times \numactionscolumn} ~|~ (x_{ab}x_{cd} - x_{ad}x_{bc})=0 \forall a\neq c, b\neq d\}\subset \vanishinglocus{\ideal}$ and $\{X\in \mathbb{R}^{\numactions\times \numactionscolumn} ~|~ \exists  a\neq c, b\neq d\text{ for which } (x_{ab}x_{cd} - x_{ad}x_{bc}) > 0\}\cap \vanishinglocus{\ideal}=\emptyset$. 
The set of polynomials which satisfy zero on independent play is then exactly the set of polynomials in $\mathbb{C}[X]$ which are zero on $\vanishinglocus{\ideal}$. 

By Theorem \ref{thm:nullstellensatz}, the set of such polynomials is exactly $\vanishingideal{\vanishinglocus{\ideal}} = \rad(\ideal) = \{f\in \mathbb{C}[X] ~|~ \exists n\in\mathbb{N} \text{ such that }f^n\in \ideal\}$.
We restrict this set to only polynomials with real coefficients, and observe that this is exactly the set of functions of the form 
\begin{align*}
     \sum_{a\neq c, b\neq d}p_{abcd}(X) (x_{ab}x_{cd} - x_{ad}x_{bc}).
\end{align*}
Replacing the complex variables $X$ with entries from $\jointdistribution$ completes the proof.
\end{proof}

\determinantfactorization*

The proof of \Cref{lma:mutual-informations-are-divisible-by-determinant-squared-of-joint-distribution-matrix}, which claims that any mutual information on a $2 \times 2$ action space with a fixed-sample estimator can be divided by the squared determinant of the joint distribution matrix, follows directly from \Cref{thm:independent-characterization}. We demonstrate this as follows:

\begin{proof}[Proof of \Cref{lma:mutual-informations-are-divisible-by-determinant-squared-of-joint-distribution-matrix}]
    For a $2 \times 2$ matrix, there is exactly one $2 \times 2$ subdeterminant: itself. Thus, from \Cref{thm:independent-characterization}, we know that any polynomial $h(\jointdistribution)$ on the matrix $\jointdistribution = \begin{pmatrix}
        a & b \\
        c & d
    \end{pmatrix}$ where $h$ satisfies the mutual information axioms can be written as a product $q(\jointdistribution) \cdot \det(\jointdistribution)$, for another polynomial $q$. It follows that if $\det(F) = 0$, then this implies that $q(\jointdistribution) = 0$.
    
    Pick some $(a, b, c, d)$ which is not the zero matrix and let us assume WLOG that $d \neq 0$, and for which $\det(\jointdistribution) = 0$.
    
    Taking derivatives of both sides wrt $a$,
    \begin{align*}
        \frac{d}{da} h(\jointdistribution) &= [\frac{d}{da} \det(\jointdistribution) ]q(\jointdistribution) + \det(\jointdistribution) [\frac{d}{da}q(\jointdistribution)]
        \\ &= [d] q(\jointdistribution) + [0] [\frac{d}{da}q(\jointdistribution)]
        \\ &= [d] q(\jointdistribution)
    \end{align*}
    We claim that $q(\jointdistribution) = 0$. Assume for a contradiction that $q(\jointdistribution) \neq 0$. Then, since $d$ was not $0$, we must have had $\frac{d}{da} h(\jointdistribution) \neq 0$. That means there is an infinitesimal movement in the direction of increasing $a$ or decreasing $a$ that produces a decrease in the value of $h$. This would violate the data-processing inequality, because $h(\jointdistribution) = 0$, and any such change would make $h$ negative.

    Thus, since $q(\jointdistribution) = 0$, \Cref{thm:independent-characterization} applies to it and we know that $q$ must be divisible by $\det(\jointdistribution)$. Therefore, $h$ is divisible by $\det(\jointdistribution)^{2}$.
\end{proof}

\subsection{Omitted Proofs from \Cref{sec:determinant-mutual-information-results} (Uniqueness of DMI with Small Samples)}

\dmiunique*

\begin{proof}
    By \Cref{lma:mutual-informations-are-polynomials}, we know that any mutual information $\mutualinformationabstract$ on a $2 \times 2$ action space with an unbiased estimator from $\numsamples$ samples must be expressible as a polynomial in the entries of $\jointdistribution = \begin{pmatrix}
        \jointdistributionentry{1}{1} & \jointdistributionentry{1}{2} \\
        \jointdistributionentry{2}{1} & \jointdistributionentry{2}{2}
    \end{pmatrix}$ with degree of exactly $\numsamples$. If we apply \Cref{lma:mutual-informations-are-divisible-by-determinant-squared-of-joint-distribution-matrix} to any such mutual information, we know that we can divide $\mutualinformationabstract$ by $\det(\jointdistribution)^{2}$, which is itself a polynomial in the entries of $\jointdistribution$ of degree exactly $4$. Thus, if $\numsamples = 4$, then we know that $\mutualinformationabstract(\jointdistribution)$ is directly proportional to $\det(\jointdistribution)^{2}$. This demonstrates that if $\numsamples = 4$, the mutual information $\mutualinformationabstract(\jointdistribution) = 16 \det(\jointdistribution)^{2}$ is unique up to a scalar multiple.

    If instead $\numsamples = 5$, then we must be able to write $\mutualinformationabstract(\jointdistribution)$ as a homogeneous polynomial of degree $1$ times $\det(\jointdistribution)^{2}$, and a homogeneous polynomial of degree $1$ is a linear function, i.e.
    \begin{align*}
        \mutualinformationabstract(\jointdistribution) = \det(\jointdistribution)^{2} (a \jointdistributionentry{1}{1} + b \jointdistributionentry{2}{1} + c \jointdistributionentry{1}{2} + d \jointdistributionentry{2}{2}) \text{.}
    \end{align*}
    Since mutual informations are alphabet invariant and $\det(\jointdistribution)^{2}$ is alphabet invariant (i.e. invariant to permutations of the rows or columns of $\jointdistribution$), we know that this linear function must also be alphabet invariant. Thus, the coefficients on the linear function must all be identical, i.e.
    \begin{align*}
        \mutualinformationabstract(\jointdistribution) = \det(\jointdistribution)^{2} (a \jointdistributionentry{1}{1} + a \jointdistributionentry{2}{1} + a \jointdistributionentry{1}{2} + a \jointdistributionentry{2}{2}) \text{.}
    \end{align*}
    Lastly, since any positive rescaling of a mutual information is still a mutual information, we can choose to set the constant $a = 16$. Since $\jointdistributionentry{1}{1} + \jointdistributionentry{2}{1} + \jointdistributionentry{1}{2} + \jointdistributionentry{2}{2} = 1$ for any joint distribution matrix, we conclude that
    \begin{align*}
        \mutualinformationabstract(\jointdistribution) = \det(\jointdistribution)^{2} \text{.}
    \end{align*}
\end{proof}

\subsection{Omitted Proofs from \Cref{sec:non-uniqueness-results} (Non-uniqueness of DMI for larger $\numsamples$)}

\stochmatrixdecomp*

\begin{proof}
We can write an arbitrary column-stochastic matrix $S$, up to row permutation, as
\begin{align*}
	S &= \begin{pmatrix}1-\varepsilon & \delta\\ \varepsilon & 1-\delta\end{pmatrix},  
\end{align*}
where $1-\varepsilon \leq \varepsilon$, as this is without loss of generality up to row permutation. 
We let $\lambda_1 = \tfrac{1-\delta}{\varepsilon}$ and $\lambda_2 = 1-\varepsilon$. 
Note that $\lambda_1$ is always defined by our assumption that the bottom entry of the first column, $\varepsilon$, is always at least the top entry, meaning $\varepsilon > 0$. 

We then evaluate our decomposition with the given values of $\lambda_1$ and $\lambda_2$:
\begin{align*}
\left( \lambda_{2} \begin{pmatrix}
        1 & 1 \\
        0 & 0
    \end{pmatrix} + (1 - \lambda_{2}) \begin{pmatrix}
        1 & 0 \\
        0 & 1
    \end{pmatrix} \right) & \begin{pmatrix}
        0 & 1 \\
        1 & 0
    \end{pmatrix} \left( \lambda_{1} \begin{pmatrix}
        1 & 1 \\
        0 & 0
    \end{pmatrix} + (1 - \lambda_{1}) \begin{pmatrix}
        1 & 0 \\
        0 & 1
    \end{pmatrix} \right)\\
&=\left(\begin{pmatrix}
        \lambda_2 & \lambda_2 \\
        0 & 0
    \end{pmatrix} + \begin{pmatrix}
        1 - \lambda_{2} & 0 \\
        0 & 1 - \lambda_{2}
    \end{pmatrix} \right)
    \left(\begin{pmatrix}
        0 & 0 \\
         \lambda_{1}  &  \lambda_{1} 
    \end{pmatrix} +  \begin{pmatrix}
        0 & 1 - \lambda_{1} \\
        1 - \lambda_{1} & 0
    \end{pmatrix} \right)\\
&=\begin{pmatrix}
        1 & \lambda_2 \\
        0 & 1 - \lambda_{2}
    \end{pmatrix}\begin{pmatrix}
        0 & 1 - \lambda_{1}\\
        1  &  \lambda_{1} 
    \end{pmatrix}\\
&=\begin{pmatrix}
        \lambda_2 & 1 - \lambda_{1}(1 - \lambda_2)\\
        1 - \lambda_{2} & \lambda_1(1-\lambda_{2})\\
    \end{pmatrix}\\
&= \begin{pmatrix}
        1-\varepsilon & 1 - \frac{1-\delta}{\varepsilon}\varepsilon\\
        \varepsilon & \frac{1-\delta}{\varepsilon}\varepsilon\\
    \end{pmatrix}\\
&= \begin{pmatrix}
        1-\varepsilon & \delta\\
        \varepsilon & 1-\delta\\
    \end{pmatrix}, 
\end{align*}
as claimed. 
\end{proof}

\binaryalphabet*

\begin{proof}
    We begin with the backwards direction.

    To show the backwards direction, consider that we know from \Cref{lma:stochastic-matrix-decomposition} that any column stochastic matrix can be written as the product of at most four ``simple'' matrices. Suppose those matrices are $S_{1}, S_{2}, S_{3}, S_{4}$, and let $S = S_{1} S_{2} S_{3} S_{4}$. Then, for any joint distribution $\jointdistribution \in \Delta([2] \times [2])$, if the row data processing inequality could be shown to hold on any of the simple matrices, it would hold on $S$ as follows:
    \begin{align*}
        \mutualinformationabstract(S \jointdistribution) &= \mutualinformationabstract(S_{1} S_{2} S_{3} S_{4} \jointdistribution)
        \\ &\leq \mutualinformationabstract(S_{2} S_{3} S_{4} \jointdistribution)
        \\ &\leq \mutualinformationabstract(S_{3} S_{4} \jointdistribution)
        \\ &\leq \mutualinformationabstract(S_{4} \jointdistribution)
        \\ &\leq \mutualinformationabstract(\jointdistribution)
    \end{align*}
    Thus, to show the row data processing inequality holds, we need merely show that it holds on any of the "simple" matrices in the decomposition; namely, permutation matrices and matrices of the form
    \begin{align*}
        \lambda \begin{pmatrix}
            1 & 1 \\
            0 & 0
        \end{pmatrix} + (1 - \lambda) \begin{pmatrix}
            1 & 0 \\
            0 & 1
        \end{pmatrix}
    \end{align*}
    for some $\lambda \in [0, 1]$.

    It immediately follows from our premise that
    \begin{align*}
        \mutualinformationabstract(\alphabetpermutation \jointdistribution) = \mutualinformationabstract(\jointdistribution) \text{ for all joint distributions $\jointdistribution$ and permutation matrices $\alphabetpermutation$}
    \end{align*}
    that the row data processing inequality $\mutualinformationabstract(\alphabetpermutation \jointdistribution) \leq \mutualinformationabstract(\jointdistribution)$ holds on permutation matrices $\alphabetpermutation$. Thus, all that remains to be shown is that it holds on all matrices of the form
    \begin{align*}
        \lambda \begin{pmatrix}
            1 & 1 \\
            0 & 0
        \end{pmatrix} + (1 - \lambda) \begin{pmatrix}
            1 & 0 \\
            0 & 1
        \end{pmatrix}
    \end{align*}
    for some $\lambda \in [0, 1]$.
    
    Let a joint distribution $\jointdistribution \in \Delta([2] \times [2])$ be given. Then, using \Cref{prop:stochastic-diagonal-decomposition}, we decompose the joint distribution $\jointdistribution$ into a column-stochastic matrix $S = \begin{pmatrix}
        1 - \varepsilon & \delta \\
        \varepsilon & 1 - \delta
    \end{pmatrix}$ and a diagonal matrix $\diagonalmatrix = \begin{pmatrix}
        a & 0 \\
        0 & 1 - a
    \end{pmatrix}$, i.e. we write $S \diagonalmatrix = \jointdistribution$. We then consider an alternative joint distribution of the form
    \begin{align*}
        \widetilde{\jointdistribution}(\lambda) \triangleq \left( \lambda \begin{pmatrix}
            1 & 1 \\
            0 & 0
        \end{pmatrix} + (1 - \lambda) S \right) \diagonalmatrix = \begin{pmatrix}
            a (\lambda + (1 - \lambda)(1 - \varepsilon)) & (1 - a) (\lambda + (1 - \lambda) \delta) \\
            a (1 - \lambda) \varepsilon & (1 - a) (1 - \lambda) (1 - \varepsilon)
        \end{pmatrix}
    \end{align*}
    produced by multiplying $\jointdistribution$ on the left by the "simple" column stochastic matrix
    \begin{align*}
        \lambda \begin{pmatrix}
            1 & 1 \\
            0 & 0
        \end{pmatrix} + (1 - \lambda) \begin{pmatrix}
            1 & 0 \\
            0 & 1
        \end{pmatrix} \text{.}
    \end{align*}
    We observe that the following four identities hold:
    \begin{align*}
        \frac{d\widetilde{\jointdistributionentry{1}{1}}}{d\lambda} &= a \varepsilon = \jointdistributionentry{2}{1}
        \\ \frac{d\widetilde{\jointdistributionentry{2}{1}}}{d\lambda} &= -a \varepsilon = -\jointdistributionentry{2}{1}
        \\ \frac{d\widetilde{\jointdistributionentry{1}{2}}}{d\lambda} &= (1 - a) (1 - \delta) = \jointdistributionentry{2}{2}
        \\ \frac{d\widetilde{\jointdistributionentry{2}{2}}}{d\lambda} &= -(1 - a) (1 - \delta) = -\jointdistributionentry{2}{2}
    \end{align*}
    
    Taking the total derivative of $\mutualinformationabstract(\widetilde{\jointdistribution}(\lambda))$ with respect to $\lambda$ and setting $\lambda = 0$, these four identities let us obtain
    \begin{align*}
        \frac{d(\mutualinformationabstract(\widetilde{\jointdistribution}(\lambda)))}{d\lambda} \vert_{\lambda = 0} &= \frac{\partial \mutualinformationabstract}{\partial \jointdistributionentry{1}{1}} (\jointdistribution) \jointdistributionentry{2}{1} + \frac{\partial \mutualinformationabstract}{\partial \jointdistributionentry{2}{1}} (\jointdistribution) (-\jointdistributionentry{2}{1}) + \frac{\partial \mutualinformationabstract}{\partial \jointdistributionentry{1}{2}} (\jointdistribution) \jointdistributionentry{2}{2} + \frac{\partial \mutualinformationabstract}{\partial \jointdistributionentry{2}{2}} (\jointdistribution) (-\jointdistributionentry{2}{2})
        \\ &= \left[ \frac{\partial \mutualinformationabstract}{\partial \jointdistributionentry{1}{1}} (\jointdistribution) - \frac{\partial \mutualinformationabstract}{\partial \jointdistributionentry{2}{1}} (\jointdistribution) \right] \jointdistributionentry{2}{1} + \left[ \frac{\partial \mutualinformationabstract}{\partial \jointdistributionentry{1}{2}} (\jointdistribution) - \frac{\partial \mutualinformationabstract}{\partial \jointdistributionentry{2}{2}} (\jointdistribution) \right] \jointdistributionentry{2}{2}
    \end{align*}

    Note that one of our assumptions for the backwards direction is that the above expression is negative for all joint distributions $\jointdistribution$.

    Furthermore, let us consider that the simple column stochastic matrix
    \begin{align*}
        \lambda \begin{pmatrix}
            1 & 1 \\
            0 & 0
        \end{pmatrix} + (1 - \lambda) \begin{pmatrix}
            1 & 0 \\
            0 & 1
        \end{pmatrix}
    \end{align*}
    can be written as the product of two simple matrices with associated coefficients $\lambda_{1}, \lambda_{2} \in [0, 1]$ where $\lambda = 1 - (1 - \lambda_{1}) (1 - \lambda_{2})$, i.e.
    \begin{align*}
        &\phantom{=} \lambda \begin{pmatrix}
            1 & 1 \\
            0 & 0
        \end{pmatrix} + (1 - \lambda) \begin{pmatrix}
            1 & 0 \\
            0 & 1
        \end{pmatrix}
        \\ &= \left( \lambda_{1} \begin{pmatrix}
            1 & 1 \\
            0 & 0
        \end{pmatrix} + (1 - \lambda_{1}) \begin{pmatrix}
            1 & 0 \\
            0 & 1
        \end{pmatrix} \right) \left( \lambda_{2} \begin{pmatrix}
            1 & 1 \\
            0 & 0
        \end{pmatrix} + (1 - \lambda_{2}) \begin{pmatrix}
            1 & 0 \\
            0 & 1
        \end{pmatrix} \right)
    \end{align*}
    Suppose that we fix $\lambda_{2} = z \in [0, 1]$ and the total value of $\lambda$ fixed as well. Then, rearranging our expression for $\lambda$, we obtain
    \begin{align*}
        \lambda_{2} = 1 - \frac{1 - \lambda}{1 - z} \text{.}
    \end{align*}
    From this, we can rewrite our total derivative for $\mutualinformationabstract(\widetilde{\jointdistribution}(\lambda))$ with respect to $\lambda$ at non-zero values of $\lambda = z$ as
    \begin{align*}
        &\phantom{=} \frac{d(\mutualinformationabstract(\widetilde{\jointdistribution}(\lambda)))}{d\lambda} \vert_{\lambda = z}
        \\ &= \frac{d(\mutualinformationabstract( \left((1 - \frac{1 - \lambda}{1 - z}) \begin{pmatrix}
            1 & 1 \\
            0 & 0
        \end{pmatrix} + (1 - (1 - \frac{1 - \lambda}{1 - z})) \begin{pmatrix}
            1 & 0 \\
            0 & 1
        \end{pmatrix} \right) \widetilde{\jointdistribution}(z)))}{d\lambda} \vert_{\lambda = z}
        \\ &= \frac{1}{1 - z} \left[ \left[ \frac{\partial \mutualinformationabstract}{\partial \jointdistributionentry{1}{1}} (\widetilde{\jointdistribution}(z)) - \frac{\partial \mutualinformationabstract}{\partial \jointdistributionentry{2}{1}} (\widetilde{\jointdistribution}(z)) \right] \widetilde{\jointdistributionentry{2}{1}}(z) + \left[ \frac{\partial \mutualinformationabstract}{\partial \jointdistributionentry{1}{2}} (\widetilde{\jointdistribution}(z)) - \frac{\partial \mutualinformationabstract}{\partial \jointdistributionentry{2}{2}} (\widetilde{\jointdistribution}(z)) \right] \widetilde{\jointdistributionentry{2}{2}}(z) \right]
        \\ &\leq 0 \text{.}
    \end{align*}
    This directly allows us to show that the data processing inequality holds on our class of simple matrices, because
    \begin{align*}
        \mutualinformationabstract(\left( \lambda \begin{pmatrix}
            1 & 1 \\
            0 & 0
        \end{pmatrix} + (1 - \lambda) \begin{pmatrix}
            1 & 0 \\
            0 & 1
        \end{pmatrix} \right) \jointdistribution) - \mutualinformationabstract(\jointdistribution) &= \int_{0}^{\lambda} \frac{d(\mutualinformationabstract(\jointdistribution))}{d\lambda} \vert_{\lambda = z} dz \leq 0 \text{.}
    \end{align*}
    This completes the backwards direction.

    For the forwards direction, consider that for any 2 by 2 permutation matrix $\alphabetpermutation$, it is the case that $\alphabetpermutation^{2}$ is the identity matrix. Thus, if the row data processing inequality holds, then we have
    \begin{align*}
        \mutualinformationabstract(\jointdistribution) = \mutualinformationabstract(\alphabetpermutation^{2} \jointdistribution) \leq \mutualinformationabstract(\alphabetpermutation \jointdistribution) \leq \mutualinformationabstract(\jointdistribution) \text{,}
    \end{align*}
    and so it follows that $\mutualinformationabstract(\alphabetpermutation \jointdistribution) = \mutualinformationabstract(\jointdistribution)$ for any joint distribution matrix $\jointdistribution \in \Delta([2] \times [2])$. Furthermore, we observe that the data processing inequality on all matrices implies that it holds on all matrices of the form
    \begin{align*}
        \lambda \begin{pmatrix}
            1 & 1 \\
            0 & 0
        \end{pmatrix} + (1 - \lambda) \begin{pmatrix}
            1 & 0 \\
            0 & 1
        \end{pmatrix}
    \end{align*}
    for some $\lambda \in [0, 1]$, including arbitrarily small $\lambda$ used to calculate the derivative
    \begin{align*}
        \frac{d(\mutualinformationabstract(\widetilde{\jointdistribution}(\lambda)))}{d\lambda} \vert_{\lambda = 0} &= \left[ \frac{\partial \mutualinformationabstract}{\partial \jointdistributionentry{1}{1}} (\jointdistribution) - \frac{\partial \mutualinformationabstract}{\partial \jointdistributionentry{2}{1}} (\jointdistribution) \right] \jointdistributionentry{2}{1} + \left[ \frac{\partial \mutualinformationabstract}{\partial \jointdistributionentry{1}{2}} (\jointdistribution) - \frac{\partial \mutualinformationabstract}{\partial \jointdistributionentry{2}{2}} (\jointdistribution) \right] \jointdistributionentry{2}{2}
    \end{align*}
    from our proof of the backwards direction. Having shown both our desired properties of $\mutualinformationabstract$, we conclude the forwards direction of the proof.
\end{proof}

We now prove that there exists another mutual information with a 6-sample estimator. 

\dminotunique*

For completeness, we also provide an estimator for the mutual information given in \Cref{thm:dmi-is-not-unique-at-six-samples}. 

Note that can divide an estimator into estimating $\det(\jointdistribution)^{2}$, or DMI, and the quantity
\begin{align}
	1 + (\jointdistributionentry{1}{1} + \jointdistributionentry{1}{2} - \jointdistributionentry{2}{1} - \jointdistributionentry{2}{2})^{2} + (\jointdistributionentry{1}{1} + \jointdistributionentry{2}{1} - \jointdistributionentry{1}{2} - \jointdistributionentry{2}{2})^{2}.\label{eqn:6-sample-est-term}
\end{align}
We have several estimators for DMI which take four samples, so it is sufficient to provide a 2-sample estimator for the above term. 

\begin{lemma}
There is a 2-sample unbiased estimator for $1 + (\jointdistributionentry{1}{1} + \jointdistributionentry{1}{2} - \jointdistributionentry{2}{1} - \jointdistributionentry{2}{2})^{2} + (\jointdistributionentry{1}{1} + \jointdistributionentry{2}{1} - \jointdistributionentry{1}{2} - \jointdistributionentry{2}{2})^{2}$, as follows:
\begin{align*}
\estimator_2(\randomsample{1}, \randomsample{2}) &= 1 + (\indicate{\randomrow{1}=1} - \indicate{\randomrow{1} = 2})(\indicate{\randomrow{2} = 1} - \indicate{\randomrow{2} = 2})\\
&\phantom{=1}+ (\indicate{\randomcolumn{1} = 1} - \indicate{\randomcolumn{1} = 2})(\indicate{\randomcolumn{2} = 1} - \indicate{\randomcolumn{2} = 2}). 
\end{align*}
\end{lemma}

\begin{proof}
We take the expectation of our estimator for some joint distribution $F$:
\begin{align*}
\E_{\randomsample{1}, \randomsample{2} \iid \jointdistribution}[\estimator_2(\randomsample{1}, \randomsample{2})] &= 1 + \E_{\randomsample{1}, \randomsample{2} \iid \jointdistribution}[(\indicate{\randomrow{1} = 1} - \indicate{\randomrow{1} = 2})(\indicate{\randomrow{2} = 1} - \indicate{\randomrow{2} = 2})]\\
&\phantom{=1}+ \E_{\randomsample{1}, \randomsample{2} \iid \jointdistribution}[(\indicate{\randomcolumn{1} = 1} - \indicate{\randomcolumn{1} = 2})(\indicate{\randomcolumn{2} = 1} - \indicate{\randomcolumn{2} = 2})]\\
&= 1 + \E_{\randomsample{1}\iid \jointdistribution}[\indicate{\randomrow{1} = 1} - \indicate{\randomrow{1} = 2}]\E_{\randomsample{2}\iid \jointdistribution}[\indicate{\randomrow{2} = 1} - \indicate{\randomrow{2} = 2}]\\
&\phantom{=1}+ \E_{\randomsample{1}\iid \jointdistribution}[\indicate{\randomcolumn{1} = 1} - \indicate{\randomcolumn{1} = 2}]\E_{\randomsample{2}\iid \jointdistribution}[\indicate{\randomcolumn{2} = 1} - \indicate{\randomcolumn{2} = 2}]\\
&= 1 + (\jointdistributionentry{1}{1} + \jointdistributionentry{1}{2} - \jointdistributionentry{2}{1} - \jointdistributionentry{2}{2})^2 + (\jointdistributionentry{1}{1} + \jointdistributionentry{2}{1} - \jointdistributionentry{1}{2} - \jointdistributionentry{2}{2})^{2},
\end{align*}
as claimed.
\end{proof}

We can combine these estimators to get an unbiased six-sample estimator for our non-DMI mutual information. 

\begin{definition}[Six-Sample Mutual Information Estimator]\label{def:six-sample-estimator}
Given samples $(\randomsample{1}, \dots, \randomsample{6})$, 
\begin{enumerate}
    \item Run a $2\times 2$ DMI estimator on samples $(\randomsample{1}, \dots, \randomsample{6})$
    \item Run $\estimator_2(\randomsample{5}, \randomsample{6})$
    \item return the product of the values from steps (1) and (2). 
\end{enumerate}
\end{definition}

We now prove that the function in \Cref{thm:dmi-is-not-unique-at-six-samples} is, in fact, a mutual information. 

\begin{proof}[Proof of \Cref{thm:dmi-is-not-unique-at-six-samples}]
    We begin by rewriting the formula for $\mutualinformationabstract(\jointdistribution)$ as
    \begin{align*}
        \mutualinformationabstract(\jointdistribution) = 16 \det(\jointdistribution)^{2} ((\jointdistributionentry{1}{1} + \jointdistributionentry{1}{2} + \jointdistributionentry{2}{1} + \jointdistributionentry{2}{2})^{2} + (\jointdistributionentry{1}{1} + \jointdistributionentry{1}{2} - \jointdistributionentry{2}{1} - \jointdistributionentry{2}{2})^{2} + (\jointdistributionentry{1}{1} + \jointdistributionentry{2}{1} - \jointdistributionentry{1}{2} - \jointdistributionentry{2}{2})^{2})
    \end{align*}
    using the fact that $\jointdistributionentry{1}{1} + \jointdistributionentry{1}{2} + \jointdistributionentry{2}{1} + \jointdistributionentry{2}{2} = 1$. This makes it clear that $\mutualinformationabstract(\jointdistribution)$ can be written as a degree $6$ \textit{homogeneous} polynomial.

    Since $\mutualinformationabstract(\jointdistribution)$ is a degree $6$ homogeneous polynomial in the entries of $\jointdistribution$, it must possess a $6$-sample unbiased estimator. All that remains to be shown is that $\mutualinformationabstract$ satisfies the mutual information axioms.

    First, we know that $\mutualinformationabstract$ satisfies the "zero on independent play" axiom, because $\det(\jointdistribution) = 0$ on independent play, and thus $\mutualinformationabstract(\jointdistribution) = 16 \det(\jointdistribution)^{2} ((\jointdistributionentry{1}{1} + \jointdistributionentry{1}{2} + \jointdistributionentry{2}{1} + \jointdistributionentry{2}{2})^{2} + (\jointdistributionentry{1}{1} + \jointdistributionentry{1}{2} - \jointdistributionentry{2}{1} - \jointdistributionentry{2}{2})^{2} + (\jointdistributionentry{1}{1} + \jointdistributionentry{2}{1} - \jointdistributionentry{1}{2} - \jointdistributionentry{2}{2})^{2}) = 0$ on independent play.
    
    Next, to show that $\mutualinformationabstract$ satisfies the data processing inequality, we consider \Cref{lma:binary-alphabet-data-processing-inequality-for-mutual-informations}. It is clear that $\mutualinformationabstract$ is invariant to permutations of its rows and columns. Therefore, we need only show that the inequality
    \begin{align*}
        0 \geq \left[ \frac{\partial \mutualinformationabstract}{\partial \jointdistributionentry{1}{1}} (\jointdistribution) - \frac{\partial \mutualinformationabstract}{\partial \jointdistributionentry{2}{1}} (\jointdistribution) \right] \jointdistributionentry{2}{1} + \left[ \frac{\partial \mutualinformationabstract}{\partial \jointdistributionentry{1}{2}} (\jointdistribution) - \frac{\partial \mutualinformationabstract}{\partial \jointdistributionentry{2}{2}} (\jointdistribution) \right] \jointdistributionentry{2}{2}
    \end{align*}
    holds for all joint distributions $\jointdistribution$.
    
    We begin by computing the partial derivatives of $\mutualinformationabstract$ with respect to the four entries of $\jointdistribution$:
    \begin{align*}
        \frac{\partial \mutualinformationabstract}{\partial \jointdistributionentry{1}{1}}(\jointdistribution) &= 32 \det(\jointdistribution) (\jointdistributionentry{2}{2})(1 + (\jointdistributionentry{1}{1} + \jointdistributionentry{1}{2} - \jointdistributionentry{2}{1} - \jointdistributionentry{2}{2})^{2} + (\jointdistributionentry{1}{1} + \jointdistributionentry{2}{1} - \jointdistributionentry{1}{2} - \jointdistributionentry{2}{2})^{2})
        \\ &+ 16 \det(\jointdistribution)^{2} (2 (\jointdistributionentry{1}{1} + \jointdistributionentry{1}{2} - \jointdistributionentry{2}{1} - \jointdistributionentry{2}{2}) + 2 (\jointdistributionentry{1}{1} + \jointdistributionentry{2}{1} - \jointdistributionentry{1}{2} - \jointdistributionentry{2}{2}))
        \\ \frac{\partial \mutualinformationabstract}{\partial \jointdistributionentry{1}{2}} &= 32 \det(\jointdistribution) (-\jointdistributionentry{2}{1})(1 + (\jointdistributionentry{1}{1} + \jointdistributionentry{1}{2} - \jointdistributionentry{2}{1} - \jointdistributionentry{2}{2})^{2} + (\jointdistributionentry{1}{1} + \jointdistributionentry{2}{1} - \jointdistributionentry{1}{2} - \jointdistributionentry{2}{2})^{2})
        \\ &+ 16 \det(\jointdistribution)^{2} (2 (\jointdistributionentry{1}{1} + \jointdistributionentry{1}{2} - \jointdistributionentry{2}{1} - \jointdistributionentry{2}{2}) - 2 (\jointdistributionentry{1}{1} + \jointdistributionentry{2}{1} - \jointdistributionentry{1}{2} - \jointdistributionentry{2}{2}))
        \\ \frac{\partial \mutualinformationabstract}{\partial \jointdistributionentry{2}{1}} &= 32 \det(\jointdistribution) (-\jointdistributionentry{1}{2})(1 + (\jointdistributionentry{1}{1} + \jointdistributionentry{1}{2} - \jointdistributionentry{2}{1} - \jointdistributionentry{2}{2})^{2} + (\jointdistributionentry{1}{1} + \jointdistributionentry{2}{1} - \jointdistributionentry{1}{2} - \jointdistributionentry{2}{2})^{2})
        \\ &+ 16 \det(\jointdistribution)^{2} (-2 (\jointdistributionentry{1}{1} + \jointdistributionentry{1}{2} - \jointdistributionentry{2}{1} - \jointdistributionentry{2}{2}) + 2 (\jointdistributionentry{1}{1} + \jointdistributionentry{2}{1} - \jointdistributionentry{1}{2} - \jointdistributionentry{2}{2}))
        \\ \frac{\partial \mutualinformationabstract}{\partial \jointdistributionentry{2}{2}} &= 32 \det(\jointdistribution) (\jointdistributionentry{1}{1})(1 + (\jointdistributionentry{1}{1} + \jointdistributionentry{1}{2} - \jointdistributionentry{2}{1} - \jointdistributionentry{2}{2})^{2} + (\jointdistributionentry{1}{1} + \jointdistributionentry{2}{1} - \jointdistributionentry{1}{2} - \jointdistributionentry{2}{2})^{2})
        \\ &+ 16 \det(\jointdistribution)^{2} (-2 (\jointdistributionentry{1}{1} + \jointdistributionentry{1}{2} - \jointdistributionentry{2}{1} - \jointdistributionentry{2}{2}) - 2 (\jointdistributionentry{1}{1} + \jointdistributionentry{2}{1} - \jointdistributionentry{1}{2} - \jointdistributionentry{2}{2}))
    \end{align*}
    
    More succinctly, these four partial derivatives can be written as the following formula:
    \begin{align*}
        \frac{\partial \mutualinformationabstract}{\partial \jointdistributionentry{i}{j}}(\jointdistribution) &= 32 \det(\jointdistribution) ((-1)^{\mathbb{1}(i \neq j)} \jointdistributionentry{i}{j})(1 + (\jointdistributionentry{1}{1} + \jointdistributionentry{1}{2} - \jointdistributionentry{2}{1} - \jointdistributionentry{2}{2})^{2} + (\jointdistributionentry{1}{1} + \jointdistributionentry{2}{1} - \jointdistributionentry{1}{2} - \jointdistributionentry{2}{2})^{2})
        \\ &+ 16 \det(\jointdistribution)^{2} ((-1)^{\mathbb{1}(i = 2)} 2 (\jointdistributionentry{1}{1} + \jointdistributionentry{1}{2} - \jointdistributionentry{2}{1} - \jointdistributionentry{2}{2}) + (-1)^{\mathbb{1}(j = 2)} 2 (\jointdistributionentry{1}{1} + \jointdistributionentry{2}{1} - \jointdistributionentry{1}{2} - \jointdistributionentry{2}{2}))
    \end{align*}
    
    Plugging these into the right-hand side of our inequality, we obtain
    \begin{align*}
        &\phantom{=} 32 \det(\widetilde{\jointdistribution}(\lambda)) (\jointdistributionentry{2}{1} \jointdistributionentry{2}{2} + \jointdistributionentry{1}{1} (-\jointdistributionentry{2}{2}) - \jointdistributionentry{2}{2} \jointdistributionentry{2}{1} - \jointdistributionentry{1}{2} (-\jointdistributionentry{2}{1}))
        \\ &\quad \cdot ((\jointdistributionentry{1}{1} + \jointdistributionentry{1}{2} + \jointdistributionentry{2}{1} + \jointdistributionentry{2}{2})^{2} + (\jointdistributionentry{1}{1} + \jointdistributionentry{1}{2} - \jointdistributionentry{2}{1} - \jointdistributionentry{2}{2})^{2} + (\jointdistributionentry{1}{1} + \jointdistributionentry{2}{1} - \jointdistributionentry{1}{2} - \jointdistributionentry{2}{2})^{2})
        \\ &+ 16 \det(\jointdistribution)^{2} (2(\jointdistributionentry{1}{1} + \jointdistributionentry{2}{1} - \jointdistributionentry{1}{2} - \jointdistributionentry{2}{2})(\jointdistributionentry{2}{1} + (-\jointdistributionentry{2}{1}) - \jointdistributionentry{2}{2} - (-\jointdistributionentry{2}{2}))
        \\ &\quad + 2(\jointdistributionentry{1}{1} + \jointdistributionentry{1}{2} - \jointdistributionentry{2}{1} - \jointdistributionentry{2}{2})(\jointdistributionentry{2}{1} + \jointdistributionentry{2}{2} - (-\jointdistributionentry{2}{1}) - (-\jointdistributionentry{2}{2})))
        \\ &= 32 \det(\jointdistribution)^{2} [ -(\jointdistributionentry{1}{1} + \jointdistributionentry{1}{2} + \jointdistributionentry{2}{1} + \jointdistributionentry{2}{2})^{2} - (\jointdistributionentry{1}{1} + \jointdistributionentry{1}{2} - \jointdistributionentry{2}{1} - \jointdistributionentry{2}{2})^{2}
        \\ &\quad - (\jointdistributionentry{1}{1} + \jointdistributionentry{2}{1} - \jointdistributionentry{1}{2} - \jointdistributionentry{2}{2})^{2} + 2(\jointdistributionentry{1}{1} + \jointdistributionentry{1}{2} - \jointdistributionentry{2}{1} - \jointdistributionentry{2}{2})(\jointdistributionentry{2}{1} + \jointdistributionentry{2}{2})]
    \end{align*}
    Let $R_{1} \triangleq \jointdistributionentry{1}{1} + \jointdistributionentry{1}{2}$ denote the first row sum of the joint distribution matrix $\jointdistribution$, and note that when restricted to the simplex, we can write the second row sum as $R_{2} = 1 - R_{1} = \jointdistributionentry{2}{1} + \jointdistributionentry{2}{2}$. Our right-hand side can then be written as
    \begin{align*}
        32 \det(\jointdistribution)^{2} [ -(1)^{2} - (R_{1} - (1 - R_{1}))^{2}
       \quad - (\jointdistributionentry{1}{1} + \jointdistributionentry{2}{1} - \jointdistributionentry{1}{2} - \jointdistributionentry{2}{2})^{2} + 2(R_{1} - (1 - R_{1}))(1 - R_{1})]
    \end{align*}
    Note that the term $- (\jointdistributionentry{1}{1} + \jointdistributionentry{2}{1} - \jointdistributionentry{1}{2} - \jointdistributionentry{2}{2})^{2}$ is bounded above by $0$, and that this inequality is tight when $\jointdistributionentry{1}{1} = \jointdistributionentry{1}{2}$ and $\jointdistributionentry{2}{1} = \jointdistributionentry{2}{2}$ for any fixed row sums $R_{1}, R_{2}$. Furthermore, the term $32 \det(\jointdistribution)^{2}$ is also non-negative. Thus, the derivative as a whole is non-positive if the term
    \begin{align*}
        -(1)^{2} - (R_{1} - (1 - R_{1}))^{2} + 2(R_{1} - (1 - R_{1}))(1 - R_{1})
    \end{align*}
    is non-positive for all first row sums $R_{1} \in [0, 1]$. This can be seen by factoring the above polynomial into the form
    \begin{align*}
        -\frac{7}{8} - 8 \left( R_{1} - \frac{5}{8} \right)^{2}
    \end{align*}
    which is non-positive. We have now shown that $\mutualinformationabstract$ is a mutual information.

    We lastly observe that the mutual information $\mutualinformationabstract$ is not equal to DMI up to a scalar multiple. Recall that when the formula for DMI is not restricted to the simplex and instead treated as a degree $6$ homogeneous polynomial, its formula becomes
    \begin{align*}
        16 \det(\jointdistribution)^{2} (\jointdistributionentry{1}{1} + \jointdistributionentry{1}{2} + \jointdistributionentry{2}{1} + \jointdistributionentry{2}{2})^{2}
    \end{align*}
    whereas our alternative mutual information is
    \begin{align*}
        \mutualinformationabstract(\jointdistribution) = 16 \det(\jointdistribution)^{2} ((\jointdistributionentry{1}{1} + \jointdistributionentry{1}{2} + \jointdistributionentry{2}{1} + \jointdistributionentry{2}{2})^{2} + (\jointdistributionentry{1}{1} + \jointdistributionentry{1}{2} - \jointdistributionentry{2}{1} - \jointdistributionentry{2}{2})^{2} + (\jointdistributionentry{1}{1} + \jointdistributionentry{2}{1} - \jointdistributionentry{1}{2} - \jointdistributionentry{2}{2})^{2}) \text{.}
    \end{align*}
    For these polynomials to be the same when restricted to the simplex, we would need
    \begin{align*}
        0 = (\jointdistributionentry{1}{1} + \jointdistributionentry{1}{2} - \jointdistributionentry{2}{1} - \jointdistributionentry{2}{2})^{2} + (\jointdistributionentry{1}{1} + \jointdistributionentry{2}{1} - \jointdistributionentry{1}{2} - \jointdistributionentry{2}{2})^{2}
    \end{align*}
    for any $\jointdistribution \in \Delta([2] \times [2])$, but since the right-hand side is not the zero polynomial, this is not the case. Therefore, the two mutual informations are distinct.
\end{proof}

\subsection{Omitted Proofs from \Cref{sec:optimal-estimator-tools} (Convex-Minimal Estimators)}

\makenewestimators*
\begin{proof}
    To show that $\estimator^{*}$ and $\estimator$ are unbiased estimators of the same mutual information, we must show that both possess the same mean for any joint distribution 
    $\jointdistribution$ over outcomes. 

    We compute the mean of $\estimator^{*}$, i.e. the MI $\mutualinformation{\estimator^{*}}$ it estimates, as follows:
    \begin{align*}
        \mutualinformation{\estimator^{*}}(\jointdistribution) &= \expect_{(\randomrow{\timeindex}, \randomcolumn{\timeindex}) \overset{\textrm{i.i.d.}}{\sim} \jointdistribution}[\estimator^{*}((\randomrow{1}, \randomcolumn{1}), ..., (\randomrow{\numsamples}, \randomcolumn{\numsamples}))]
        \\ &= \expect_{(\randomrow{\timeindex}, \randomcolumn{\timeindex}) \overset{\textrm{i.i.d.}}{\sim} \jointdistribution}[\expect_{\indexpermutation}[\estimator((\randomrow{\indexpermutation(1)}, \randomcolumn{\indexpermutation(1)}), ..., (\randomrow{\indexpermutation(\numsamples)}, \randomcolumn{\indexpermutation(\numsamples)}))]]
        \\ &= \expect_{\indexpermutation}[\expect_{(\randomrow{\timeindex}, \randomcolumn{\timeindex}) \overset{\textrm{i.i.d.}}{\sim} \jointdistribution}[\estimator((\randomrow{\indexpermutation(1)}, \randomcolumn{\indexpermutation(1)}), ..., (\randomrow{\indexpermutation(\numsamples)}, \randomcolumn{\indexpermutation(\numsamples)}))]] & & \text{(independence of $\indexpermutation$ and $\{ (\randomrow{\timeindex}, \randomcolumn{\timeindex}) \}_{\timeindex = 1}^{\numsamples}$)}
        \\ &= \expect_{\indexpermutation}[\expect_{(\randomrow{\timeindex}, \randomcolumn{\timeindex}) \overset{\textrm{i.i.d.}}{\sim} \jointdistribution}[\estimator((\randomrow{1}, \randomcolumn{1}), ..., (\randomrow{\numsamples}, \randomcolumn{\numsamples}))]] & & \text{(since $\{ (\randomrow{\timeindex}, \randomcolumn{\timeindex}) \}_{\timeindex = 1}^{\numsamples}$ are i.i.d.)}
        \\ &= \expect_{(\randomrow{\timeindex}, \randomcolumn{\timeindex}) \overset{\textrm{i.i.d.}}{\sim} \jointdistribution}[\estimator((\randomrow{1}, \randomcolumn{1}), ..., (\randomrow{\numsamples}, \randomcolumn{\numsamples}))]
        \\ &= \mutualinformation{\estimator}(\jointdistribution) \text{.}
    \end{align*}
    Thus, both $\estimator$ and $\estimator^{*}$ are unbiased estimators of the same mutual information.
\end{proof}

\sampleorderinvariance*
\begin{proof}
    
    We will show that $\estimator^{*}$ is convex-dominated by $\estimator$ by showing that for any convex function $\convexfunction$,
    \begin{align*}
        \expect_{(\randomrow{\timeindex}, \randomcolumn{\timeindex}) \overset{\textrm{i.i.d.}}{\sim} \jointdistribution}[\convexfunction(\estimator^{*}((\randomrow{1}, \randomcolumn{1}), ..., (\randomrow{\numsamples}, \randomcolumn{\numsamples})))] \leq \expect_{(\randomrow{\timeindex}, \randomcolumn{\timeindex}) \overset{\textrm{i.i.d.}}{\sim} \jointdistribution}[\convexfunction(\estimator((\randomrow{1}, \randomcolumn{1}), ..., (\randomrow{\numsamples}, \randomcolumn{\numsamples})))] \text{.}
    \end{align*}

    We upper-bound the left-had side as follows:
    \begin{align*}
        &\mathrel{\phantom{=}} \expect_{(\randomrow{\timeindex}, \randomcolumn{\timeindex}) \overset{\textrm{i.i.d.}}{\sim} \jointdistribution}[\convexfunction(\estimator^{*}((\randomrow{1}, \randomcolumn{1}), ..., (\randomrow{\numsamples}, \randomcolumn{\numsamples})))]
        \\ &= \expect_{(\randomrow{\timeindex}, \randomcolumn{\timeindex}) \overset{\textrm{i.i.d.}}{\sim} \jointdistribution}[\convexfunction(\expect_{\indexpermutation}[\estimator((\randomrow{\indexpermutation(1)}, \randomcolumn{\indexpermutation(1)}), ..., (\randomrow{\indexpermutation(\numsamples)}, \randomcolumn{\indexpermutation(\numsamples)}))])]
        \\ &\leq \expect_{(\randomrow{\timeindex}, \randomcolumn{\timeindex}) \overset{\textrm{i.i.d.}}{\sim} \jointdistribution}[\expect_{\indexpermutation}[\convexfunction(\estimator((\randomrow{\indexpermutation(1)}, \randomcolumn{\indexpermutation(1)}), ..., (\randomrow{\indexpermutation(\numsamples)}, \randomcolumn{\indexpermutation(\numsamples)})))]]  & & \text{(Jensen's Inequality)}
        \\ &= \expect_{\indexpermutation}[\expect_{(\randomrow{\timeindex}, \randomcolumn{\timeindex}) \overset{\textrm{i.i.d.}}{\sim} \jointdistribution}[\convexfunction(\estimator((\randomrow{\indexpermutation(1)}, \randomcolumn{\indexpermutation(1)}), ..., (\randomrow{\indexpermutation(\numsamples)}, \randomcolumn{\indexpermutation(\numsamples)})))]] & & \text{(independence of $\indexpermutation$ and $\{ (\randomrow{\timeindex}, \randomcolumn{\timeindex}) \}_{\timeindex = 1}^{\numsamples}$)}
        \\ &= \expect_{\indexpermutation}[\expect_{(\randomrow{\timeindex}, \randomcolumn{\timeindex}) \overset{\textrm{i.i.d.}}{\sim} \jointdistribution}[\convexfunction(\estimator((\randomrow{1}, \randomcolumn{1}), ..., (\randomrow{\numsamples}, \randomcolumn{\numsamples})))]] & & \text{(since $\{ (\randomrow{\timeindex}, \randomcolumn{\timeindex}) \}_{\timeindex = 1}^{\numsamples}$ are i.i.d.)}
        \\ &= \expect_{(\randomrow{\timeindex}, \randomcolumn{\timeindex}) \overset{\textrm{i.i.d.}}{\sim} \jointdistribution}[\convexfunction(\estimator((\randomrow{1}, \randomcolumn{1}), ..., (\randomrow{\numsamples}, \randomcolumn{\numsamples})))]
    \end{align*}

    Thus, for any convex function $\convexfunction$,
    \begin{align*}
        \expect_{(\randomrow{\timeindex}, \randomcolumn{\timeindex}) \overset{\textrm{i.i.d.}}{\sim} \jointdistribution}[\convexfunction(\estimator^{*}((\randomrow{1}, \randomcolumn{1}), ..., (\randomrow{\numsamples}, \randomcolumn{\numsamples})))] \leq \expect_{(\randomrow{\timeindex}, \randomcolumn{\timeindex}) \overset{\textrm{i.i.d.}}{\sim} \jointdistribution}[\convexfunction(\estimator((\randomrow{1}, \randomcolumn{1}), ..., (\randomrow{\numsamples}, \randomcolumn{\numsamples})))]
    \end{align*}
    and therefore $\estimator^{*}$ is convex-dominated by $\estimator$.

    We now argue that if $\estimator$ is not sample-order-invariant, $\estimator^*$ is strictly convex dominated by $\estimator$. 
    
    Let $\jointdistribution$ be the uniform matrix, so any sequence of $\numsamples$ samples has positive probability of occurring. 
    
    Note that the inequality comes entirely from Jensen's inequality. 
    The strict version of this inequality states that for strictly convex $\convexfunction$, this inequality is strict \textit{when the expectation of $\estimator$ over $\pi$ is not constant} for some sequence drawn from our uniform $\jointdistribution$, which happens with positive probability. 
    Thus, $\estimator$ strictly convex-dominates $\estimator^*$ whenever $\estimator$ is not itself sample-order-invariant.
\end{proof}

\MIsampleorderinvariant*

\begin{proof}
    Suppose that for our fixed row and column alphabet sizes $\numactions, \numactionscolumn$ and number of samples $\numsamples$, there exist two sample-order-invariant unbiased estimators, $\estimator$ and $\estimator'$, of the mutual information $\mutualinformationabstract$. In other words, assume that $\mutualinformation{\estimator} = \mutualinformation{\estimator'} = \mutualinformationabstract$. Then, we can write
    \begin{align*}
        0 &= \mutualinformation{\estimator}(\jointdistribution) - \mutualinformation{\estimator'}(\jointdistribution)
        \\ &= \expect_{\tallymatrixrandom \sim \multinomialdistribution{\numsamples}{\jointdistribution}}[\estimator(\tallymatrixrandom)] - \expect_{\tallymatrixrandom \sim \multinomialdistribution{\numsamples}{\jointdistribution}}[\estimator'(\tallymatrixrandom)]
        \\ &= \expect_{\tallymatrixrandom \sim \multinomialdistribution{\numsamples}{\jointdistribution}}[\estimator(\tallymatrixrandom) - \estimator'(\tallymatrixrandom)]
        \\ &= \sum_{\tallymatrix \in \tallyset{\numsamples}{\numactions}{\numactionscolumn}} (\estimator(\tallymatrix) - \estimator'(\tallymatrix)) \frac{\numsamples!}{\prod_{i = 1}^{\numactions} \prod_{j = 1}^{\numactionscolumn} (\tallymatrixentry{i}{j}!)} \prod_{i = 1}^{\numactions} \prod_{j = 1}^{\numactionscolumn} \jointdistributionentry{i}{j}^{\tallymatrixentry{i}{j}}
    \end{align*}
    which is a homogeneous multivariate polynomial of degree $\numsamples$. In order for this polynomial to be identically $0$ on all inputs $\jointdistribution$, its coefficients must always be $0$. The multinomial coefficient $\frac{\numsamples!}{\prod_{i = 1}^{\numactions} \prod_{j = 1}^{\numactionscolumn} (\tallymatrixentry{i}{j}!)}$ is always non-zero. Therefore, $\estimator(\tallymatrix) = \estimator'(\tallymatrix)$ for all tally matrices $\tallymatrix \in \tallyset{\numsamples}{\numactions}{\numactionscolumn}$. There can thus exist at most one sample-order-invariant estimator from $\numsamples$ samples and with row and column alphabet sizes $\numactions, \numactionscolumn$ of any mutual information.
\end{proof}

\subsection{Omitted Proofs from \Cref{sec:optimal-estimator-results} (Convex-Minimal Estimation of DMI)}

\convexminimalbinary*

\begin{proof}
    From \Cref{thm:convex-minimal-dmi-estimator-with-binary-alphabet}, we know that for $\numsamples = 4$ samples, the convex-minimal estimator $\estimator$ of DMI $16 \det(\jointdistribution)$ is
    \begin{align*}
        \estimator \left( \begin{bmatrix}
            2 & 0 \\
            0 & 2
        \end{bmatrix} \right)
        =
        \estimator \left( \begin{bmatrix}
            0 & 2 \\
            2 & 0
        \end{bmatrix} \right)
        &=
        \frac{8}{3}
        \\
        \\
        \estimator \left( \begin{bmatrix}
            1 & 1 \\
            1 & 1
        \end{bmatrix} \right)
        &=
        -\frac{4}{3}
        \\
        \\
        \estimator(\text{ anything else }) &= 0
    \end{align*}
    and we can use this estimator to construct the convex-minimal estimator when $\numsamples > 4$. We do this as follows: Consider an estimator on $\numsamples > 4$ samples that takes the first $4$ samples (ignoring the others) and returns $\estimator$ on those first $4$ samples. Then, apply the permutation variance grinder to this randomized estimator. By \Cref{cor:output-of-permutation-variance-grinder-is-unique-and-convex-minimal}, we know that this estimator is sample-order-invariant and convex-minimal.

    We can compute the explicit formula for this convex-minimal estimator of DMI. Let $\tallymatrix^{\indexpermutation}[:4]$ be the tally matrix solely considering the first $4$ samples in the order $\indexpermutation$, where the order $\indexpermutation$ is drawn uniformly at random from the set of all permutations of $\numsamples$ indices. We then compute our convex-minimal estimator $\estimator^{*}$ as
    \begin{align*}
        \estimator^{*}(\tallymatrix) \equiv \expect_{\indexpermutation}[\estimator(\tallymatrix^{\indexpermutation}[:4])] \text{.}
    \end{align*}
    Observe that the possible tally matrices $\tallymatrix^{\indexpermutation}[:4]$ that produce a non-zero value of $\estimator$ are solely
    \begin{align*}
        \begin{bmatrix}
            2 & 0 \\
            0 & 2
        \end{bmatrix}
        \text{, }
        \begin{bmatrix}
            0 & 2 \\
            2 & 0
        \end{bmatrix}
        \text{, and }
        \begin{bmatrix}
            1 & 1 \\
            1 & 1
        \end{bmatrix} \text{.}
    \end{align*}
    Furthermore, note that the probability of observing any of these three tally matrices can be written as the probability mass function of the multivariate hypergeometric distribution, i.e. the distribution induced by sampling $4$ elements without replacement from a bin that contains $\tallymatrixentry{1}{1}$ elements of type $(1, 1)$, $\tallymatrixentry{2}{1}$ elements of type $(2, 1)$, $\tallymatrixentry{1}{2}$ elements of type $(1, 2)$, and $\tallymatrixentry{2}{2}$ elements of type $(2, 2)$. The probability this multivariate hypergeometric distribution assigns to observing empirical matrix
    \begin{align*}
        \left( \begin{bmatrix}
            s_{1, 1} & s_{1, 2} \\
            s_{2, 1} & s_{2, 2}
        \end{bmatrix} \right)
    \end{align*}
    for samples $s_{1, 1} + s_{2, 1} + s_{1, 2} + s_{2, 2} = 4$ is the standard formula
    \begin{align*}
        \frac{\frac{\tallymatrixentry{1}{1}!}{s_{1, 1}! (\tallymatrixentry{1}{1} - s_{1, 1})!} \cdot \frac{\tallymatrixentry{2}{1}!}{s_{2, 1}! (\tallymatrixentry{2}{1} - s_{2, 1})!} \cdot \frac{\tallymatrixentry{1}{2}!}{s_{1, 2}! (\tallymatrixentry{1}{2} - s_{1, 2})!} \cdot \frac{\tallymatrixentry{2}{2}!}{s_{2, 2}! (\tallymatrixentry{2}{2} - s_{2, 2})!}}{\frac{\numsamples!}{4! (\numsamples - 4)!}} \text{,}
    \end{align*}
    i.e. the probability is the product of the binomial coefficients associated with each sample index, divided by the binomial coefficient for all samples.

    Thus, we can write our expectation over index permutations $\indexpermutation$ as
    \begin{align*}
        \estimator^{*}(\tallymatrix) &\equiv \expect_{\indexpermutation}[\estimator(\tallymatrix^{\indexpermutation}[:4])]
        \\ &= \estimator \left( \begin{bmatrix}
            2 & 0 \\
            0 & 2
        \end{bmatrix} \right) \frac{\frac{\tallymatrixentry{1}{1}!}{2! (\tallymatrixentry{1}{1} - 2)!} \cdot \frac{\tallymatrixentry{2}{2}!}{2! (\tallymatrixentry{2}{2} - 2)!}}{\frac{\numsamples!}{4! (\numsamples - 4)!}}
        \\ &\mathrel{\phantom{=}} + \estimator \left( \begin{bmatrix}
            0 & 2 \\
            2 & 0
        \end{bmatrix} \right) \frac{\frac{\tallymatrixentry{2}{1}!}{2! (\tallymatrixentry{2}{1} - 2)!} \cdot \frac{\tallymatrixentry{1}{2}!}{2! (\tallymatrixentry{1}{2} - 2)!}}{\frac{\numsamples!}{4! (\numsamples - 4)!}}
        \\ &\mathrel{\phantom{=}} + \estimator \left( \begin{bmatrix}
            1 & 1 \\
            1 & 1
        \end{bmatrix} \right) \frac{\tallymatrixentry{1}{1} \cdot \tallymatrixentry{2}{1} \cdot \tallymatrixentry{1}{2} \cdot \tallymatrixentry{2}{2}}{\frac{\numsamples!}{4! (\numsamples - 4)!}}
        \\ &= \left( \frac{4! (\numsamples - 4)!}{\numsamples!} \right) (\estimator \left( \begin{bmatrix}
            2 & 0 \\
            0 & 2
        \end{bmatrix} \right) \cdot \frac{\tallymatrixentry{1}{1} (\tallymatrixentry{1}{1} - 1)}{2} \cdot \frac{\tallymatrixentry{2}{2} (\tallymatrixentry{2}{2} - 1)}{2}
        \\ &\mathrel{\phantom{=}} + \estimator \left( \begin{bmatrix}
            0 & 2 \\
            2 & 0
        \end{bmatrix} \right) \cdot \frac{\tallymatrixentry{2}{1} (\tallymatrixentry{2}{1} - 1)}{2} \cdot \frac{\tallymatrixentry{1}{2} (\tallymatrixentry{1}{2} - 1)}{2}
        \\ &\mathrel{\phantom{=}} + \estimator \left( \begin{bmatrix}
            1 & 1 \\
            1 & 1
        \end{bmatrix} \right) \cdot \tallymatrixentry{1}{1} \cdot \tallymatrixentry{2}{1} \cdot \tallymatrixentry{1}{2} \cdot \tallymatrixentry{2}{2} )
    \end{align*}
    Then, inserting our values of
    \begin{align*}
        \estimator \left( \begin{bmatrix}
            2 & 0 \\
            0 & 2
        \end{bmatrix} \right)
        =
        \estimator \left( \begin{bmatrix}
            0 & 2 \\
            2 & 0
        \end{bmatrix} \right)
        &=
        \frac{8}{3}
        \\
        \\
        \text{and } \estimator \left( \begin{bmatrix}
            1 & 1 \\
            1 & 1
        \end{bmatrix} \right)
        &=
        -\frac{4}{3} \text{,}
    \end{align*}
    we obtain
    \begin{align*}
        \estimator^{*}(\tallymatrix) = \left( \frac{4! (\numsamples - 4)!}{\numsamples!} \right) (&\frac{8}{3} \cdot \frac{\tallymatrixentry{1}{1} (\tallymatrixentry{1}{1} - 1)}{2} \cdot \frac{\tallymatrixentry{2}{2} (\tallymatrixentry{2}{2} - 1)}{2}
        \\ + &\frac{8}{3} \cdot \frac{\tallymatrixentry{2}{1} (\tallymatrixentry{2}{1} - 1)}{2} \cdot \frac{\tallymatrixentry{1}{2} (\tallymatrixentry{1}{2} - 1)}{2}
        \\ - &\frac{4}{3} \cdot \tallymatrixentry{1}{1} \cdot \tallymatrixentry{2}{1} \cdot \tallymatrixentry{1}{2} \cdot \tallymatrixentry{2}{2} )
    \end{align*}
    which, after grouping together coefficients, becomes
    \begin{align*}
        \estimator^{*}(\tallymatrix) = \frac{16 (\numsamples - 4)!}{\numsamples!} (&\tallymatrixentry{1}{1} (\tallymatrixentry{1}{1} - 1) \tallymatrixentry{2}{2} (\tallymatrixentry{2}{2} - 1)
        \\ + &\tallymatrixentry{2}{1} (\tallymatrixentry{2}{1} - 1) \tallymatrixentry{1}{2} (\tallymatrixentry{1}{2} - 1)
        \\ - &2 \tallymatrixentry{1}{1} \tallymatrixentry{2}{1} \tallymatrixentry{1}{2} \tallymatrixentry{2}{2} ) \text{.}
    \end{align*}
    Splitting the term $2 \tallymatrixentry{1}{1} \tallymatrixentry{2}{1} \tallymatrixentry{1}{2} \tallymatrixentry{2}{2}$ into two copies of $\tallymatrixentry{1}{1} \tallymatrixentry{2}{1} \tallymatrixentry{1}{2} \tallymatrixentry{2}{2}$ and grouping by like coefficients, we obtain our final formula for the convex-minimal estimator of DMI:
    \begin{align*}
        \frac{16 (\numsamples - 4)!}{\numsamples!} ( &\tallymatrixentry{1}{1} \tallymatrixentry{2}{2} ((\tallymatrixentry{1}{1} - 1) (\tallymatrixentry{2}{2} - 1) - \tallymatrixentry{2}{1} \tallymatrixentry{1}{2})
        \\ + &\tallymatrixentry{2}{1} \tallymatrixentry{1}{2} ((\tallymatrixentry{2}{1} - 1) (\tallymatrixentry{1}{2} - 1) - \tallymatrixentry{1}{1} \tallymatrixentry{2}{2}) )
    \end{align*}
\end{proof}

\foursample*

\begin{proof}
    There is a finite number of inputs to the formula from \Cref{thm:convex-minimal-dmi-estimator-with-binary-alphabet-variance} on $4$ samples. Enumerating all the $2 \times 2$ tally matrices on $4$ samples and evaluating each them produces the desired lookup table.
\end{proof}

\section{Additional Search Space Reductions on Mutual Informations}\label{app:extended-search-space-reductions}

Though they are not used directly in any proofs from \Cref{sec:fixed-sample-mutual-information-estimators}, we also prove several tools for narrowing the search space for fixed-sample unbiased estimators of mutual informations. We provide these tools here.

After having restricted the search space to sample-order-invariant estimators using \Cref{cor:output-of-permutation-variance-grinder-is-unique-and-convex-minimal}, we can further reduce the search space by observing that sample-order-invariant fixed-sample estimators of a mutual information are \textit{alphabet invariant}; that is, permuting the labels of either the row or column player's actions should not change the value of the estimator. We state this property formally in \Cref{lma:sample-order-invariant-mutual-information-estimators-are-alphabet-invariant}: 

\begin{restatable}[Sample-order-invariant MI estimators are alphabet invariant]{lemma}{alphabetinvariant}
\label{lma:sample-order-invariant-mutual-information-estimators-are-alphabet-invariant}
    For any sample-order-invariant $\numsamples$-sample estimator $\estimator$ that induces a mutual information $\mutualinformation{\estimator}(\jointdistribution) := \expect_{\tallymatrixrandom \sim \multinomialdistribution{\numsamples}{\jointdistribution}}[\estimator(\tallymatrixrandom)]$, the estimator $\estimator$ is alphabet invariant, i.e. for all tally matrices $\tallymatrix \in \tallyset{\numsamples}{\numactions}{\numactionscolumn}$, $\numactions \times \numactions$ permutation matrices $\alphabetpermutation$, and $\numactionscolumn \times \numactionscolumn$ permutation matrices $\alphabetpermutation'$, we have $\estimator(\alphabetpermutation \tallymatrix) = \estimator(\tallymatrix \alphabetpermutation') = \estimator(\tallymatrix)$.
\end{restatable}

\begin{proof}
    The claim follows directly from \Cref{thm:any-mutual-information-has-at-most-one-sample-order-invariant-estimator}, but we can also prove it directly.

    We know from the data processing inequality that for any unbiased estimator $\estimator$ of a mutual information $\mutualinformation{\estimator}$ and any permutation matrix $P$, $\mutualinformation{\estimator}(\alphabetpermutation \jointdistribution) = \mutualinformation{\estimator}(\jointdistribution)$, because
    \begin{align*}
        \mutualinformation{\estimator}(\jointdistribution) = \mutualinformation{\estimator}(\alphabetpermutation^{-1} \alphabetpermutation \jointdistribution) \leq \mutualinformation{\estimator}(\alphabetpermutation \jointdistribution) \leq \mutualinformation{\estimator}(\jointdistribution) \text{.}
    \end{align*}
    Thus, we can write
    \begin{align*}
        0 &= \mutualinformation{\estimator}(\alphabetpermutation \jointdistribution) - \mutualinformation{\estimator}(\jointdistribution)
        \\ &= \expect_{\tallymatrixrandom \sim \multinomialdistribution{\numsamples}{\alphabetpermutation \jointdistribution}}[\estimator(\tallymatrixrandom)] - \expect_{\tallymatrixrandom \sim \multinomialdistribution{\numsamples}{\jointdistribution}}[\estimator(\tallymatrixrandom)]
        \\ &= \expect_{\tallymatrixrandom \sim \multinomialdistribution{\numsamples}{\jointdistribution}}[\estimator(\alphabetpermutation \tallymatrixrandom)] - \expect_{\tallymatrixrandom \sim \multinomialdistribution{\numsamples}{\jointdistribution}}[\estimator(\tallymatrixrandom)]
        \\ &= \expect_{\tallymatrixrandom \sim \multinomialdistribution{\numsamples}{\jointdistribution}}[\estimator(\alphabetpermutation \tallymatrixrandom) - \estimator(\tallymatrixrandom)]
        \\ &= \sum_{\tallymatrix \in \tallyset{\numsamples}{\numactions}{\numactionscolumn}} (\estimator(\alphabetpermutation \tallymatrix) - \estimator(\tallymatrix)) \frac{\numsamples!}{\prod_{i = 1}^{\numactions} \prod_{j = 1}^{\numactionscolumn} (\tallymatrixentry{i}{j}!)} \prod_{i = 1}^{\numactions} \prod_{j = 1}^{\numactionscolumn} \jointdistributionentry{i}{j}^{\tallymatrixentry{i}{j}}
    \end{align*}
    which is a homogeneous multivariate polynomial of degree $\numsamples$. In order for this polynomial to be identically $0$ on all inputs $\jointdistribution$, its coefficients must always be $0$. Thus, for all tally matrices $\tallymatrix \in \tallyset{\numsamples}{\numactions}{\numactionscolumn}$ and $\numactions \times \numactions$ permutation matrices $\alphabetpermutation$, we have $\estimator(\alphabetpermutation \tallymatrix) = \estimator(\tallymatrix)$. The proof is identical for showing that $\estimator(\tallymatrix \alphabetpermutation') = \estimator(\tallymatrix)$.
\end{proof}

We then turn to the mutual information axioms from \Cref{def:mutual-information} to narrow the search space even further. We begin by invoking the \textit{zero on independent play} axiom, which states that a mutual information is zero on rank one joint distributions. \Cref{thm:independence-induced-linear-constraints} proves a powerful consequence of this axiom and \Cref{lma:sample-order-invariant-mutual-information-estimators-are-alphabet-invariant}, inducing several linear constraints on any sample-order-invariant estimator of a mutual information.

\begin{restatable}[Independence-induced linear constraints]{theorem}{linearconstraints}
\label{thm:independence-induced-linear-constraints}
    For any sample-order-invariant $\numsamples$-sample unbiased estimator $\estimator$ of a mutual information and for all row sums $\rowsum{1} \geq ... \geq \rowsum{\numactions} \in \mathbb{N}$ and column sums $\columnsum{1} \geq ... \geq \columnsum{\numactionscolumn} \in \mathbb{N}$ such that $\sum_{i = 1}^{\numactions} \rowsum{i} = \numsamples = \sum_{j = 1}^{\numactionscolumn} \columnsum{j}$, we have
    \begin{align*}
        0 = \sum_{\substack{\tallymatrix \in \tallyset{\numsamples}{\numactions}{\numactionscolumn} : \\ \forall i \in [\numactions], \sum_{j = 1}^{\numactionscolumn} \tallymatrixentry{i}{j} = \rowsum{i} \\ \forall j \in [\numactionscolumn], \sum_{i = 1}^{\numactions} \tallymatrixentry{i}{j} = \columnsum{j}}} \frac{\numsamples!}{\prod_{i = 1}^{\numactions} \prod_{j = 1}^{\numactionscolumn} (\tallymatrixentry{i}{j}!)} \estimator(\tallymatrix)
    \end{align*}
\end{restatable}

Before proving Theorem \ref{thm:independence-induced-linear-constraints}, we provide a brief overview of the proof to provide intuition.
    The probability mass function of the multinomial distribution allows the expectation of a sample-order-invariant $\numsamples$-sample estimator to be written as a degree $\numsamples$ homogeneous polynomial in the entries of the joint distribution matrix $\jointdistribution \in \Delta([\numactions] \times [\numactionscolumn])$. When $\jointdistribution$ results from independent play, that expectation must be \textit{the zero polynomial}. Thus, all of its coefficients must be zero. The right-hand side of the equality constraints stated in \Cref{thm:independence-induced-linear-constraints} correspond to each of these coefficients. An explanation of the exact functional form of the constraints is provided in the full proof next.

\begin{proof}
    Let probability vectors $\probvectorone \in \Delta_{\numactions}, \probvectortwo \in \Delta_{\numactionscolumn}$ be given, and observe that the joint distribution $\probvectorone \transpose{\probvectortwo}$ results from independent play. By the \textit{zero on independent play} axiom, we have that for any sample-order-invariant $\numsamples$-sample unbiased estimator $\estimator$ of a mutual information $\mutualinformation{\estimator}$,
    \begin{align*}
        \mutualinformation{\estimator}(\probvectorone \transpose{\probvectortwo}) \triangleq \expect_{\tallymatrixrandom \sim \multinomialdistribution{\numsamples}{\probvectorone \transpose{\probvectortwo}}}[\estimator(\tallymatrixrandom)] = 0 \text{.}
    \end{align*}
    We can therefore compute as follows:
    \begin{align*}
        0 &= \expect_{\tallymatrixrandom \sim \multinomialdistribution{\numsamples}{\probvectorone \transpose{\probvectortwo}}}[\estimator(\tallymatrixrandom)]
        \\ &= \sum_{\tallymatrix \in \tallyset{\numsamples}{\numactions}{\numactionscolumn}} \estimator(\tallymatrix) \prob(\tallymatrixrandom = \tallymatrix)
        \\ &= \sum_{\tallymatrix \in \tallyset{\numsamples}{\numactions}{\numactionscolumn}} \estimator(\tallymatrix) \frac{\numsamples!}{\prod_{i = 1}^{\numactions} \prod_{j = 1}^{\numactionscolumn} (\tallymatrixentry{i}{j}!)} \prod_{i = 1}^{\numactions} \prod_{j = 1}^{\numactionscolumn} (\probvectoroneentry{i} \probvectortwoentry{j})^{\tallymatrixentry{i}{j}} & & \text{(multinomial PMF)}
        \\ &= \sum_{\tallymatrix \in \tallyset{\numsamples}{\numactions}{\numactionscolumn}} \estimator(\tallymatrix) \frac{\numsamples!}{\prod_{i = 1}^{\numactions} \prod_{j = 1}^{\numactionscolumn} (\tallymatrixentry{i}{j}!)} \left( \prod_{i = 1}^{\numactions} \prod_{j = 1}^{\numactionscolumn} \probvectoroneentry{i}^{\tallymatrixentry{i}{j}} \right) \left( \prod_{i = 1}^{\numactions} \prod_{j = 1}^{\numactionscolumn} \probvectortwoentry{j}^{\tallymatrixentry{i}{j}} \right) & & \text{(distribute exponents $\tallymatrixentry{i}{j}$)}
        \\ &= \sum_{\tallymatrix \in \tallyset{\numsamples}{\numactions}{\numactionscolumn}} \frac{\numsamples!}{\prod_{i = 1}^{\numactions} \prod_{j = 1}^{\numactionscolumn} (\tallymatrixentry{i}{j}!)} \estimator(\tallymatrix) \left( \prod_{i = 1}^{\numactions} \probvectoroneentry{i}^{\sum_{j = 1}^{\numactionscolumn} \tallymatrixentry{i}{j}} \right) \left( \prod_{j = 1}^{\numactionscolumn} \probvectortwoentry{j}^{\sum_{i = 1}^{\numactions} \tallymatrixentry{i}{j}} \right) & & \text{(group exponents by $\probvectoroneentry{i}, \probvectortwoentry{j}$)}
        \\ &= \sum_{\tallymatrix \in \tallyset{\numsamples}{\numactions}{\numactionscolumn}} \frac{\numsamples!}{\prod_{i = 1}^{\numactions} \prod_{j = 1}^{\numactionscolumn} (\tallymatrixentry{i}{j}!)} \estimator(\tallymatrix) \left( \prod_{i = 1}^{\numactions} \probvectoroneentry{i}^{\sum_{j = 1}^{\numactionscolumn} \tallymatrixentry{i}{j}} \right) \left( \prod_{j = 1}^{\numactionscolumn} \probvectortwoentry{j}^{\sum_{i = 1}^{\numactions} \tallymatrixentry{i}{j}} \right) & & \text{(swap multiplication order)}
    \end{align*}
    This is a polynomial in the entries of $\probvectorone$ and $\probvectortwo$, such that the exponent of any variable is a column sum $\sum_{i = 1}^{\numactions} \tallymatrixentry{i}{j}$ or row sum $\sum_{j = 1}^{\numactionscolumn} \tallymatrixentry{i}{j}$. We can therefore rewrite the outer summation to be over collections of tally matrices with particular row and column sums, i.e.
    \begin{align*}
        0 = \sum_{\substack{\rowsum{1}, ..., \rowsum{\numactions} \in \mathbb{N} : \\ \sum_{i = 1}^{\numactions} \rowsum{i} = \numsamples}} \sum_{\substack{\columnsum{1}, ..., \columnsum{\numactionscolumn} \in \mathbb{N} : \\ \sum_{j = 1}^{\numactionscolumn} \columnsum{i} = \numsamples}} \left( \sum_{\substack{\tallymatrix \in \tallyset{\numsamples}{\numactions}{\numactionscolumn} : \\ \forall i \in [\numactions], \sum_{j = 1}^{\numactionscolumn} \tallymatrixentry{i}{j} = \rowsum{i} \\ \forall j \in [\numactionscolumn], \sum_{i = 1}^{\numactions} \tallymatrixentry{i}{j} = \columnsum{j}}} \frac{\numsamples!}{\prod_{i = 1}^{\numactions} \prod_{j = 1}^{\numactionscolumn} (\tallymatrixentry{i}{j}!)} \estimator(\tallymatrix) \right) \left( \prod_{i = 1}^{\numactions} \probvectoroneentry{i}^{\rowsum{i}} \right) \left( \prod_{j = 1}^{\numactionscolumn} \probvectortwoentry{j}^{\columnsum{j}} \right)
    \end{align*}
    where $\rowsum{1}, ..., \rowsum{\numactions}$ and $\columnsum{1}, ..., \columnsum{\numactionscolumn}$ are different possible row and column sums for a tally matrix $\tallymatrix$.

    In order for the polynomial on the right-hand side to be the zero polynomial, all of its coefficients must be zero. This implies a collection of linear constraints on the values of $\frac{\numsamples!}{\prod_{i = 1}^{\numactions} \prod_{j = 1}^{\numactionscolumn} (\tallymatrixentry{i}{j}!)} \estimator(\tallymatrix)$ of the form
    \begin{align*}
        0 = \sum_{\substack{\tallymatrix \in \tallyset{\numsamples}{\numactions}{\numactionscolumn} : \\ \forall i \in [\numactions], \sum_{j = 1}^{\numactionscolumn} \tallymatrixentry{i}{j} = \rowsum{i} \\ \forall j \in [\numactionscolumn], \sum_{i = 1}^{\numactions} \tallymatrixentry{i}{j} = \columnsum{j}}} \frac{\numsamples!}{\prod_{i = 1}^{\numactions} \prod_{j = 1}^{\numactionscolumn} (\tallymatrixentry{i}{j}!)} \estimator(\tallymatrix).
    \end{align*}
    for every collection of row and column sums $\rowsum{1}, ..., \rowsum{\numactions}$ and $\columnsum{1}, ..., \columnsum{\numactionscolumn}$.

    However, as per \Cref{lma:sample-order-invariant-mutual-information-estimators-are-alphabet-invariant}, the estimator $\estimator$ must be alphabet invariant, i.e. unchanged by permutations of the rows or columns of its input $\tallymatrix$. Furthermore, observe that the multinomial coefficient
    \begin{align*}
        \frac{\numsamples!}{\prod_{i = 1}^{\numactions} \prod_{j = 1}^{\numactionscolumn} (\tallymatrixentry{i}{j}!)}
    \end{align*}
    is also invariant to permutations of the rows or columns of the tally matrix $\tallymatrix$. Therefore, it follows that several of these linear equalities provide redundant constraints, since any permutation of a given collection of row and column sums will produce an identical linear constraint. To only consider linear constraints that do not imply one another, we can restrict our attention to collections of row and column sums $\rowsum{1}, ..., \rowsum{\numactions}$ and $\columnsum{1}, ..., \columnsum{\numactionscolumn}$ that cannot be re-ordered to be equal to one another; without loss of generality, we will select collections of row and column sums that are ordered from largest to smallest. Thus, we conclude that
    \begin{align*}
        0 = \sum_{\substack{\tallymatrix \in \tallyset{\numsamples}{\numactions}{\numactionscolumn} : \\ \forall i \in [\numactions], \sum_{j = 1}^{\numactionscolumn} \tallymatrixentry{i}{j} = \rowsum{i} \\ \forall j \in [\numactionscolumn], \sum_{i = 1}^{\numactions} \tallymatrixentry{i}{j} = \columnsum{j}}} \frac{\numsamples!}{\prod_{i = 1}^{\numactions} \prod_{j = 1}^{\numactionscolumn} (\tallymatrixentry{i}{j}!)} \estimator(\tallymatrix)
    \end{align*}
    for every collection of ordered row and column sums $\rowsum{1} \geq ... \geq \rowsum{\numactions} \in \mathbb{N}$ and $\columnsum{1} \geq ... \geq \columnsum{\numactionscolumn} \in \mathbb{N}$ such that $\sum_{i = 1}^{\numactions} \rowsum{i} = \numsamples = \sum_{j = 1}^{\numactionscolumn} \columnsum{j}$.
\end{proof}

\Cref{thm:independence-induced-linear-constraints} reduces the search space dimension for mutual information estimators considerably. For example, when the row and column alphabet sizes are $2$, \Cref{thm:independence-induced-linear-constraints} and the theorems which precede it collectively reduce the $\numsamples = 4$ sample estimator search space dimension from $256$ to $2$, and the $\numsamples = 5$ sample estimator search space dimension from $1024$ to $5$.

A final important tool to reduce the search space is a corollary of \Cref{lma:binary-alphabet-data-processing-inequality-for-mutual-informations}, which rewrites the formula for the inequality which is equivalent to data processing inequality of \Cref{lma:binary-alphabet-data-processing-inequality-for-mutual-informations} from a formula that references a \textit{mutual information}, to instead operate on a \textit{fixed-sample unbiased estimator of a mutual information.} If one instead wishes to prove whether an \textit{estimator} has an expectation which satisfies the data processing inequality, the following corollary can be used:

\begin{corollary}
[Binary alphabet data processing inequality for estimators]
\label{cor:binary-alphabet-data-processing-inequality-for-estimators}
    Fix row and column alphabet sizes $\numactions = \numactionscolumn = 2$. Then, for any sample-order-invariant estimator $\estimator$ on $\numsamples$ samples, its expectation $\mutualinformation{\estimator}$ satisfies the row data processing inequality
    \begin{align*}
        \mutualinformation{\estimator}(S \jointdistribution) \leq \mutualinformation{\estimator}(\jointdistribution) \text{ for all joint distributions $\jointdistribution$ and column-stochastic $S$}
    \end{align*}
    if and only if it satisfies both row permutation invariance, i.e.
    \begin{align*}
        \mutualinformation{\estimator}(\alphabetpermutation \jointdistribution) = \mutualinformation{\estimator}(\jointdistribution) \text{ for all joint distributions $\jointdistribution$ and permutation matrices $\alphabetpermutation$}
    \end{align*}
    and the following inequality:
    \begin{align*}
        0 \geq \expect_{\tallymatrixrandom \sim \multinomialdistribution{\numsamples}{\jointdistribution}}[ &(\estimator\left( \tallymatrixrandom + \begin{bmatrix}
            1 & 0 \\
            -1 & 0
        \end{bmatrix} \right) - \estimator(\tallymatrixrandom)) \tallymatrixrandomentry{2}{1}
        \\ + &(\estimator\left( \tallymatrixrandom + \begin{bmatrix}
            0 & 1 \\
            0 & -1
        \end{bmatrix} \right) - \estimator(\tallymatrixrandom)) \tallymatrixrandomentry{2}{2} ] \text{ for all joint distributions $\jointdistribution$}
    \end{align*}
\end{corollary}

\begin{proof}
    We already know that the data processing equality is satisfied by row permutations of the joint distribution, because sample-order-invariant estimators of a mutual information are invariant to row permutations according to \Cref{lma:sample-order-invariant-mutual-information-estimators-are-alphabet-invariant}. Therefore, all that remains is to show a condition under which the mutual information decreases during any move toward $(\xparallelogram, \yparallelogram) = (0, 1)$ from \Cref{lma:stochastic-matrix-decomposition}. A necessary and sufficient condition would be that a mutual information decreases for any infinitesimal move toward $(\xparallelogram, \yparallelogram) = (0, 1)$.

    More formally, consider selecting the path through column-stochastic matrices $T : [0, 1] \rightarrow \{ 2 \times 2 \text{ column-stochastic matrices} \}$ that moves from a starting point $(\xparallelograminit, \yparallelograminit)$ to the point $(0, 1)$ as
    \begin{align*}
        T(\lambda) = (1 - \lambda) \begin{pmatrix}
            1 - \xparallelograminit & \yparallelograminit \\
            \xparallelograminit & 1 - \yparallelograminit
        \end{pmatrix} + \lambda \begin{pmatrix}
            1 & 1 \\
            0 & 0
        \end{pmatrix} = \begin{pmatrix}
            1 - (1 - \lambda) \xparallelograminit & 1 - (1 - \lambda) (1 - \yparallelograminit) \\
            (1 - \lambda) \xparallelograminit & (1 - \lambda) (1 - \yparallelograminit)
        \end{pmatrix}
    \end{align*}
    and if we can show that the mutual information
    \begin{align*}
        \mutualinformation{\estimator}(T(\lambda) \diagonalmatrix) = \sum_{\tallymatrix \in \tallyset{\numsamples}{2}{2}} \estimator(\tallymatrix) & \frac{\numsamples!}{\tallymatrixentry{1}{1}! \cdot \tallymatrixentry{1}{2}! \cdot \tallymatrixentry{2}{1}! \cdot \tallymatrixentry{2}{2}!}
        \\ & \cdot (1 - (1 - \lambda) \xparallelograminit)^{\tallymatrixentry{1}{1}}
        \\ & \cdot (1 - (1 - \lambda) (1 - \yparallelograminit))^{\tallymatrixentry{1}{2}}
        \\ & \cdot ((1 - \lambda) \xparallelograminit)^{\tallymatrixentry{2}{1}}
        \\ & \cdot ((1 - \lambda) (1 - \yparallelograminit))^{\tallymatrixentry{2}{2}}
        \\ & \cdot \diagtopleft^{\tallymatrixentry{1}{1} + \tallymatrixentry{2}{1}} (1 - \diagtopleft)^{\tallymatrixentry{1}{2} + \tallymatrixentry{2}{2}}
    \end{align*}
    is decreasing in $\lambda$ for any choice of diagonal matrix $\diagonalmatrix = \begin{pmatrix}
        \diagtopleft & 0 \\
        0 & 1 - \diagtopleft \\
    \end{pmatrix}$ and any choice of $(\xparallelograminit, \yparallelograminit)$, then this condition will be equivalent to the "row" data processing inequality and invariance to permutations.

    Taking the total derivative with respect to lambda, we compute that
    \begin{align*}
        \frac{d}{d \lambda} \mutualinformation{\estimator}(T(\lambda) \diagonalmatrix) = \sum_{\tallymatrix \in \tallyset{\numsamples}{2}{2}} \estimator(\tallymatrix) & \frac{\numsamples!}{\tallymatrixentry{1}{1}! \cdot \tallymatrixentry{1}{2}! \cdot \tallymatrixentry{2}{1}! \cdot \tallymatrixentry{2}{2}!}
        \\ & \cdot \left( \frac{\xparallelograminit \tallymatrixentry{1}{1}}{1 - (1 - \lambda) \xparallelograminit} + \frac{(1 - \yparallelograminit) \tallymatrixentry{1}{2}}{1 - (1 - \lambda) (1 - \yparallelograminit)} + \frac{-\xparallelograminit \tallymatrixentry{2}{1}}{(1 - \lambda) \xparallelograminit} + \frac{-(1 - \yparallelograminit) \tallymatrixentry{2}{2}}{(1 - \lambda) (1 - \yparallelograminit)} \right)
        \\ & \cdot (1 - (1 - \lambda) \xparallelograminit)^{\tallymatrixentry{1}{1}}
        \\ & \cdot (1 - (1 - \lambda) (1 - \yparallelograminit))^{\tallymatrixentry{1}{2}}
        \\ & \cdot ((1 - \lambda) \xparallelograminit)^{\tallymatrixentry{2}{1}}
        \\ & \cdot ((1 - \lambda) (1 - \yparallelograminit))^{\tallymatrixentry{2}{2}}
        \\ & \cdot \diagtopleft^{\tallymatrixentry{1}{1} + \tallymatrixentry{2}{1}} (1 - \diagtopleft)^{\tallymatrixentry{1}{2} + \tallymatrixentry{2}{2}}
        \intertext{and setting $\lambda$ to $0$ to compute the infinitesimal change in the mutual information at the start of the path, this becomes}
        \frac{d}{d \lambda} \mutualinformation{\estimator}(T(\lambda) \diagonalmatrix) \rvert_{\lambda = 0} = \sum_{\tallymatrix \in \tallyset{\numsamples}{2}{2}} \estimator(\tallymatrix) & \frac{\numsamples!}{\tallymatrixentry{1}{1}! \cdot \tallymatrixentry{1}{2}! \cdot \tallymatrixentry{2}{1}! \cdot \tallymatrixentry{2}{2}!}
        \\ & \cdot \left( \frac{\xparallelograminit}{1 - \xparallelograminit} \tallymatrixentry{1}{1} + \frac{1 - \yparallelograminit}{\yparallelograminit} \tallymatrixentry{1}{2} - \tallymatrixentry{2}{1} - \tallymatrixentry{2}{2} \right)
        \\ & \cdot (1 - \xparallelograminit)^{\tallymatrixentry{1}{1}}
        \\ & \cdot \yparallelograminit^{\tallymatrixentry{1}{2}}
        \\ & \cdot \xparallelograminit^{\tallymatrixentry{2}{1}}
        \\ & \cdot (1 - \yparallelograminit)^{\tallymatrixentry{2}{2}}
        \\ & \cdot \diagtopleft^{\tallymatrixentry{1}{1} + \tallymatrixentry{2}{1}} (1 - \diagtopleft)^{\tallymatrixentry{1}{2} + \tallymatrixentry{2}{2}}
        \intertext{which we can further simplify by rewriting in terms of the joint distribution $\jointdistribution \triangleq T(0) \diagonalmatrix$ at the start of the path, by using the identities $\jointdistributionentry{1}{1} = (1 - \xparallelograminit) \diagtopleft$, $\jointdistributionentry{1}{2} = \yparallelograminit (1 - \diagtopleft)$, $\jointdistributionentry{2}{1} = \xparallelograminit \diagtopleft$, and $\jointdistributionentry{2}{2} = (1 - \yparallelograminit) (1 - \diagtopleft)$:}
        \sum_{\tallymatrix \in \tallyset{\numsamples}{2}{2}} \estimator(\tallymatrix) & \frac{\numsamples!}{\tallymatrixentry{1}{1}! \cdot \tallymatrixentry{1}{2}! \cdot \tallymatrixentry{2}{1}! \cdot \tallymatrixentry{2}{2}!}
        \\ & \cdot \left( \frac{\xparallelograminit \diagtopleft}{(1 - \xparallelograminit) \diagtopleft} \tallymatrixentry{1}{1} + \frac{(1 - \yparallelograminit) (1 - \diagtopleft)}{\yparallelograminit (1 - \diagtopleft)} \tallymatrixentry{1}{2} - \tallymatrixentry{2}{1} - \tallymatrixentry{2}{2} \right)
        \\ & \cdot ((1 - \xparallelograminit) \diagtopleft)^{\tallymatrixentry{1}{1}}
        \\ & \cdot (\delta (1 - \diagtopleft))^{\tallymatrixentry{1}{2}}
        \\ & \cdot (\xparallelograminit \diagtopleft)^{\tallymatrixentry{2}{1}}
        \\ & \cdot ((1 - \yparallelograminit) (1 - \diagtopleft))^{\tallymatrixentry{2}{2}}
        \\ = \sum_{\tallymatrix \in \tallyset{\numsamples}{2}{2}} \estimator(\tallymatrix) & \frac{\numsamples!}{\tallymatrixentry{1}{1}! \cdot \tallymatrixentry{1}{2}! \cdot \tallymatrixentry{2}{1}! \cdot \tallymatrixentry{2}{2}!}
        \\ & \cdot \left( \frac{\jointdistributionentry{2}{1}}{\jointdistributionentry{1}{1}} \tallymatrixentry{1}{1} + \frac{\jointdistributionentry{2}{2}}{\jointdistributionentry{1}{2}} \tallymatrixentry{1}{2} - \tallymatrixentry{2}{1} - \tallymatrixentry{2}{2} \right)
        \\ & \cdot \jointdistributionentry{1}{1}^{\tallymatrixentry{1}{1}}
        \\ & \cdot \jointdistributionentry{1}{2}^{\tallymatrixentry{1}{2}}
        \\ & \cdot \jointdistributionentry{2}{1}^{\tallymatrixentry{2}{1}}
        \\ & \cdot \jointdistributionentry{2}{2}^{\tallymatrixentry{2}{2}}
    \end{align*}

    Recall that this quantity must be negative in order for the mutual information to be decreasing as we move toward $(\xparallelogram, \yparallelogram) = (0, 1)$ on the parallelogram. Splitting the summation over each of the four terms and restricting each sum to only sum over non-zero terms, we must therefore have
    \begin{align*}
        0 &\geq \sum_{\tallymatrix \in \tallyset{\numsamples}{2}{2}} \estimator(\tallymatrix) \frac{\numsamples!}{\tallymatrixentry{1}{1}! \cdot \tallymatrixentry{1}{2}! \cdot \tallymatrixentry{2}{1}! \cdot \tallymatrixentry{2}{2}!} \left( \frac{\jointdistributionentry{2}{1}}{\jointdistributionentry{1}{1}} \tallymatrixentry{1}{1} + \frac{\jointdistributionentry{2}{2}}{\jointdistributionentry{1}{2}} \tallymatrixentry{1}{2} - \tallymatrixentry{2}{1} - \tallymatrixentry{2}{2} \right) \jointdistributionentry{1}{1}^{\tallymatrixentry{1}{1}} \jointdistributionentry{1}{2}^{\tallymatrixentry{1}{2}} \jointdistributionentry{2}{1}^{\tallymatrixentry{2}{1}} \jointdistributionentry{2}{2}^{\tallymatrixentry{2}{2}}
        \\ &= \sum_{\substack{\tallymatrix \in \tallyset{\numsamples}{2}{2}: \\ \tallymatrixentry{1}{1} > 0}} \estimator(\tallymatrix) \frac{\numsamples!}{\tallymatrixentry{1}{1}! \cdot \tallymatrixentry{1}{2}! \cdot \tallymatrixentry{2}{1}! \cdot \tallymatrixentry{2}{2}!} \left( \frac{\jointdistributionentry{2}{1}}{\jointdistributionentry{1}{1}} \tallymatrixentry{1}{1} \right) \jointdistributionentry{1}{1}^{\tallymatrixentry{1}{1}} \jointdistributionentry{1}{2}^{\tallymatrixentry{1}{2}} \jointdistributionentry{2}{1}^{\tallymatrixentry{2}{1}} \jointdistributionentry{2}{2}^{\tallymatrixentry{2}{2}}
        \\ &+ \sum_{\substack{\tallymatrix \in \tallyset{\numsamples}{2}{2}: \\ \tallymatrixentry{1}{2} > 0}} \estimator(\tallymatrix) \frac{\numsamples!}{\tallymatrixentry{1}{1}! \cdot \tallymatrixentry{1}{2}! \cdot \tallymatrixentry{2}{1}! \cdot \tallymatrixentry{2}{2}!} \left( \frac{\jointdistributionentry{2}{2}}{\jointdistributionentry{1}{2}} \tallymatrixentry{1}{2} \right) \jointdistributionentry{1}{1}^{\tallymatrixentry{1}{1}} \jointdistributionentry{1}{2}^{\tallymatrixentry{1}{2}} \jointdistributionentry{2}{1}^{\tallymatrixentry{2}{1}} \jointdistributionentry{2}{2}^{\tallymatrixentry{2}{2}}
        \\ &- \sum_{\substack{\tallymatrix \in \tallyset{\numsamples}{2}{2}: \\ \tallymatrixentry{2}{1} > 0}} \estimator(\tallymatrix) \frac{\numsamples!}{\tallymatrixentry{1}{1}! \cdot \tallymatrixentry{1}{2}! \cdot \tallymatrixentry{2}{1}! \cdot \tallymatrixentry{2}{2}!} \left( \tallymatrixentry{2}{1} \right) \jointdistributionentry{1}{1}^{\tallymatrixentry{1}{1}} \jointdistributionentry{1}{2}^{\tallymatrixentry{1}{2}} \jointdistributionentry{2}{1}^{\tallymatrixentry{2}{1}} \jointdistributionentry{2}{2}^{\tallymatrixentry{2}{2}}
        \\ &- \sum_{\substack{\tallymatrix \in \tallyset{\numsamples}{2}{2}: \\ \tallymatrixentry{2}{2} > 0}} \estimator(\tallymatrix) \frac{\numsamples!}{\tallymatrixentry{1}{1}! \cdot \tallymatrixentry{1}{2}! \cdot \tallymatrixentry{2}{1}! \cdot \tallymatrixentry{2}{2}!} \left( \tallymatrixentry{2}{2} \right) \jointdistributionentry{1}{1}^{\tallymatrixentry{1}{1}} \jointdistributionentry{1}{2}^{\tallymatrixentry{1}{2}} \jointdistributionentry{2}{1}^{\tallymatrixentry{2}{1}} \jointdistributionentry{2}{2}^{\tallymatrixentry{2}{2}}
        \\ &= \sum_{\substack{\tallymatrix \in \tallyset{\numsamples}{2}{2}: \\ \tallymatrixentry{1}{1} > 0}} \estimator(\tallymatrix) \frac{\numsamples!}{(\tallymatrixentry{1}{1} - 1)! \cdot \tallymatrixentry{1}{2}! \cdot (\tallymatrixentry{2}{1} + 1)! \cdot \tallymatrixentry{2}{2}!} \left( \tallymatrixentry{2}{1} + 1 \right) \jointdistributionentry{1}{1}^{(\tallymatrixentry{1}{1} - 1)} \jointdistributionentry{1}{2}^{\tallymatrixentry{1}{2}} \jointdistributionentry{2}{1}^{(\tallymatrixentry{2}{1} + 1)} \jointdistributionentry{2}{2}^{\tallymatrixentry{2}{2}}
        \\ &+ \sum_{\substack{\tallymatrix \in \tallyset{\numsamples}{2}{2}: \\ \tallymatrixentry{1}{2} > 0}} \estimator(\tallymatrix) \frac{\numsamples!}{\tallymatrixentry{1}{1}! \cdot (\tallymatrixentry{1}{2} - 1)! \cdot \tallymatrixentry{2}{1}! \cdot (\tallymatrixentry{2}{2} + 1)!} \left( \tallymatrixentry{2}{2} + 1 \right) \jointdistributionentry{1}{1}^{\tallymatrixentry{1}{1}} \jointdistributionentry{1}{2}^{(\tallymatrixentry{1}{2} - 1)} \jointdistributionentry{2}{1}^{\tallymatrixentry{2}{1}} \jointdistributionentry{2}{2}^{(\tallymatrixentry{2}{2} + 1)}
        \\ &- \sum_{\substack{\tallymatrix \in \tallyset{\numsamples}{2}{2}: \\ \tallymatrixentry{2}{1} > 0}} \estimator(\tallymatrix) \frac{\numsamples!}{\tallymatrixentry{1}{1}! \cdot \tallymatrixentry{1}{2}! \cdot \tallymatrixentry{2}{1}! \cdot \tallymatrixentry{2}{2}!} \left( \tallymatrixentry{2}{1} \right) \jointdistributionentry{1}{1}^{\tallymatrixentry{1}{1}} \jointdistributionentry{1}{2}^{\tallymatrixentry{1}{2}} \jointdistributionentry{2}{1}^{\tallymatrixentry{2}{1}} \jointdistributionentry{2}{2}^{\tallymatrixentry{2}{2}}
        \\ &- \sum_{\substack{\tallymatrix \in \tallyset{\numsamples}{2}{2}: \\ \tallymatrixentry{2}{2} > 0}} \estimator(\tallymatrix) \frac{\numsamples!}{\tallymatrixentry{1}{1}! \cdot \tallymatrixentry{1}{2}! \cdot \tallymatrixentry{2}{1}! \cdot \tallymatrixentry{2}{2}!} \left( \tallymatrixentry{2}{2} \right) \jointdistributionentry{1}{1}^{\tallymatrixentry{1}{1}} \jointdistributionentry{1}{2}^{\tallymatrixentry{1}{2}} \jointdistributionentry{2}{1}^{\tallymatrixentry{2}{1}} \jointdistributionentry{2}{2}^{\tallymatrixentry{2}{2}}
    \end{align*}

    In the above, we will now transform the first two summations into a form such that they can be combined with the last two summations. Observe that the mapping
    \begin{align*}
        \tallymatrix \mapsto \tallymatrix + \begin{bmatrix}
            -1 & 0 \\
            1 & 0
        \end{bmatrix}
    \end{align*}
    is a bijection from the set of tally matrices $\tallyset{\numsamples}{2}{2}$ where $\tallymatrixentry{1}{1} > 0$ to the set of tally matrices where $\tallymatrixentry{2}{1} > 0$, and that the mapping
    \begin{align*}
        \tallymatrix \mapsto \tallymatrix + \begin{bmatrix}
            0 & -1 \\
            0 & 1
        \end{bmatrix}
    \end{align*}
    is a bijection from the set of tally matrices $\tallyset{\numsamples}{2}{2}$ where $\tallymatrixentry{1}{2} > 0$ to the set of tally matrices where $\tallymatrixentry{2}{2} > 0$. Applying these bijections to the first and second summations, respectively, and the inverse of these bijections to every term in each sum, we obtain
    \begin{align*}
        0 &\geq \sum_{\substack{\tallymatrix \in \tallyset{\numsamples}{2}{2}: \\ \tallymatrixentry{2}{1} > 0}} \estimator\left( \tallymatrix + \begin{bmatrix}
            1 & 0 \\
            -1 & 0
        \end{bmatrix} \right) \frac{\numsamples!}{\tallymatrixentry{1}{1}! \cdot \tallymatrixentry{1}{2}! \cdot \tallymatrixentry{2}{1}! \cdot \tallymatrixentry{2}{2}!} \left( \tallymatrixentry{2}{1} \right) \jointdistributionentry{1}{1}^{\tallymatrixentry{1}{1}} \jointdistributionentry{1}{2}^{\tallymatrixentry{1}{2}} \jointdistributionentry{2}{1}^{\tallymatrixentry{2}{1}} \jointdistributionentry{2}{2}^{\tallymatrixentry{2}{2}}
        \\ &+ \sum_{\substack{\tallymatrix \in \tallyset{\numsamples}{2}{2}: \\ \tallymatrixentry{2}{2} > 0}} \estimator\left( \tallymatrix + \begin{bmatrix}
            0 & 1 \\
            0 & -1
        \end{bmatrix} \right) \frac{\numsamples!}{\tallymatrixentry{1}{1}! \cdot \tallymatrixentry{1}{2}! \cdot \tallymatrixentry{2}{1}! \cdot \tallymatrixentry{2}{2}!} \left( \tallymatrixentry{2}{2} \right) \jointdistributionentry{1}{1}^{\tallymatrixentry{1}{1}} \jointdistributionentry{1}{2}^{\tallymatrixentry{1}{2}} \jointdistributionentry{2}{1}^{\tallymatrixentry{2}{1}} \jointdistributionentry{2}{2}^{\tallymatrixentry{2}{2}}
        \\ &- \sum_{\substack{\tallymatrix \in \tallyset{\numsamples}{2}{2}: \\ \tallymatrixentry{2}{1} > 0}} \estimator(\tallymatrix) \frac{\numsamples!}{\tallymatrixentry{1}{1}! \cdot \tallymatrixentry{1}{2}! \cdot \tallymatrixentry{2}{1}! \cdot \tallymatrixentry{2}{2}!} \left( \tallymatrixentry{2}{1} \right) \jointdistributionentry{1}{1}^{\tallymatrixentry{1}{1}} \jointdistributionentry{1}{2}^{\tallymatrixentry{1}{2}} \jointdistributionentry{2}{1}^{\tallymatrixentry{2}{1}} \jointdistributionentry{2}{2}^{\tallymatrixentry{2}{2}}
        \\ &- \sum_{\substack{\tallymatrix \in \tallyset{\numsamples}{2}{2}: \\ \tallymatrixentry{2}{2} > 0}} \estimator(\tallymatrix) \frac{\numsamples!}{\tallymatrixentry{1}{1}! \cdot \tallymatrixentry{1}{2}! \cdot \tallymatrixentry{2}{1}! \cdot \tallymatrixentry{2}{2}!} \left( \tallymatrixentry{2}{2} \right) \jointdistributionentry{1}{1}^{\tallymatrixentry{1}{1}} \jointdistributionentry{1}{2}^{\tallymatrixentry{1}{2}} \jointdistributionentry{2}{1}^{\tallymatrixentry{2}{1}} \jointdistributionentry{2}{2}^{\tallymatrixentry{2}{2}}
        \\ &= \sum_{\substack{\tallymatrix \in \tallyset{\numsamples}{2}{2}: \\ \tallymatrixentry{2}{1} > 0}} \left( \estimator\left( \tallymatrix + \begin{bmatrix}
            1 & 0 \\
            -1 & 0
        \end{bmatrix} \right) - \estimator(\tallymatrix) \right) \tallymatrixentry{2}{1} \frac{\numsamples!}{\tallymatrixentry{1}{1}! \cdot \tallymatrixentry{1}{2}! \cdot \tallymatrixentry{2}{1}! \cdot \tallymatrixentry{2}{2}!} \jointdistributionentry{1}{1}^{\tallymatrixentry{1}{1}} \jointdistributionentry{1}{2}^{\tallymatrixentry{1}{2}} \jointdistributionentry{2}{1}^{\tallymatrixentry{2}{1}} \jointdistributionentry{2}{2}^{\tallymatrixentry{2}{2}}
        \\ &+ \sum_{\substack{\tallymatrix \in \tallyset{\numsamples}{2}{2}: \\ \tallymatrixentry{2}{2} > 0}} \left( \estimator\left( \tallymatrix + \begin{bmatrix}
            0 & 1 \\
            0 & -1
        \end{bmatrix} \right) - \estimator(\tallymatrix) \right) \tallymatrixentry{2}{2} \frac{\numsamples!}{\tallymatrixentry{1}{1}! \cdot \tallymatrixentry{1}{2}! \cdot \tallymatrixentry{2}{1}! \cdot \tallymatrixentry{2}{2}!} \jointdistributionentry{1}{1}^{\tallymatrixentry{1}{1}} \jointdistributionentry{1}{2}^{\tallymatrixentry{1}{2}} \jointdistributionentry{2}{1}^{\tallymatrixentry{2}{1}} \jointdistributionentry{2}{2}^{\tallymatrixentry{2}{2}}
        \\ &= \sum_{\substack{\tallymatrix \in \tallyset{\numsamples}{2}{2}: \\ \tallymatrixentry{2}{1} > 0}} \left( \estimator\left( \tallymatrix + \begin{bmatrix}
            1 & 0 \\
            -1 & 0
        \end{bmatrix} \right) - \estimator(\tallymatrix) \right) \tallymatrixentry{2}{1} \prob(\tallymatrixrandom = \tallymatrix)
        \\ &+ \sum_{\substack{\tallymatrix \in \tallyset{\numsamples}{2}{2}: \\ \tallymatrixentry{2}{2} > 0}} \left( \estimator\left( \tallymatrix + \begin{bmatrix}
            0 & 1 \\
            0 & -1
        \end{bmatrix} \right) - \estimator(\tallymatrix) \right) \tallymatrixentry{2}{2} \prob(\tallymatrixrandom = \tallymatrix)
    \end{align*}
    Note that even though the terms $\estimator\left( \tallymatrix + \begin{bmatrix}
        1 & 0 \\
        -1 & 0
    \end{bmatrix} \right) - \estimator(\tallymatrix)$ and $\estimator\left( \tallymatrix + \begin{bmatrix}
        0 & 1 \\
        0 & -1
    \end{bmatrix} \right) - \estimator(\tallymatrix)$ in the sums are not well-defined when $\tallymatrixentry{2}{1} = 0$ and $\tallymatrixentry{2}{2} = 0$, respectively, they are multiplied by $\tallymatrixentry{2}{1}$ and $\tallymatrixentry{2}{2}$ and so that term in each summation becomes $0$ whenever they would not be well-defined. This justifies us dropping the "$\tallymatrixentry{2}{1} > 0$" and "$\tallymatrixentry{2}{2} > 0$" conditions from each summation and allowing us to write them as expectations:
    \begin{align*}
        0 \geq \expect_{\tallymatrixrandom \sim \multinomialdistribution{\numsamples}{\jointdistribution}}[ &(\estimator\left( \tallymatrixrandom + \begin{bmatrix}
            1 & 0 \\
            -1 & 0
        \end{bmatrix} \right) - \estimator(\tallymatrixrandom)) \tallymatrixrandomentry{2}{1}
        \\ + &(\estimator\left( \tallymatrixrandom + \begin{bmatrix}
            0 & 1 \\
            0 & -1
        \end{bmatrix} \right) - \estimator(\tallymatrixrandom)) \tallymatrixrandomentry{2}{2} ]
    \end{align*}
    Recall from \Cref{prop:stochastic-diagonal-decomposition} that every joint distribution $\jointdistribution$ has a decomposition into a column-stochastic matrix and a diagonal distribution matrix. Since we have shown that the above inequality holds for any column-stochastic matrix multiplied by a diagonal distribution matrix, we have shown that the above inequality holds for all joint distributions $\jointdistribution$.

    To show the second inequality in our claim, i.e.
    \begin{align*}
        0 \geq \expect_{\tallymatrixrandom \sim \multinomialdistribution{\numsamples}{\jointdistribution}}[ &(\estimator\left( \tallymatrixrandom + \begin{bmatrix}
            1 & -1 \\
            0 & 0
        \end{bmatrix} \right) - \estimator(\tallymatrixrandom)) \tallymatrixrandomentry{1}{2}
        \\ + &(\estimator\left( \tallymatrixrandom + \begin{bmatrix}
            0 & 0 \\
            1 & -1
        \end{bmatrix} \right) - \estimator(\tallymatrixrandom)) \tallymatrixrandomentry{2}{2} ]
    \end{align*}
    and invariance to row permutations is equivalent to the "column" data processing inequality, we will repeat the above steps on the \textit{transpose} of any starting joint distribution $\jointdistribution$ and an alternative estimator $\estimator'$ defined that takes the transpose of its input and then runs it through $\estimator$. Since the "column" data processing inequality holds for the expectation of an estimator if and only if the "row" data processing inequality holds for the expectation of an estimator that operates on the transpose of its original input, it is sufficient to prove our second desired inequality and invariance to column permutations by merely repeating our original series of steps for our estimator composed with a transpose operator.
    
    More formally, we first define
    \begin{align*}
        \estimator'(\tallymatrix) := \estimator(\transpose{\tallymatrix})
    \end{align*}
    and compute that
    \begin{align*}
        0 \geq \expect_{\tallymatrixrandom \sim \multinomialdistribution{\numsamples}{\transpose{\jointdistribution}}}[ &(\estimator'\left( \tallymatrixrandom + \begin{bmatrix}
            1 & 0 \\
            -1 & 0
        \end{bmatrix} \right) - \estimator'(\tallymatrixrandom)) \tallymatrixrandomentry{2}{1}
        \\ + &(\estimator'\left( \tallymatrixrandom + \begin{bmatrix}
            0 & 1 \\
            0 & -1
        \end{bmatrix} \right) - \estimator'(\tallymatrixrandom)) \tallymatrixrandomentry{2}{2} ]
        \\ = \expect_{\tallymatrixrandom \sim \multinomialdistribution{\numsamples}{\jointdistribution}}[ &(\estimator'\left( \transpose{\tallymatrixrandom} + \begin{bmatrix}
            1 & 0 \\
            -1 & 0
        \end{bmatrix} \right) - \estimator'(\transpose{\tallymatrixrandom})) \tallymatrixrandomentry{1}{2}
        \\ + &(\estimator'\left( \transpose{\tallymatrixrandom} + \begin{bmatrix}
            0 & 1 \\
            0 & -1
        \end{bmatrix} \right) - \estimator'(\transpose{\tallymatrixrandom})) \tallymatrixrandomentry{2}{2} ]
        \\ = \expect_{\tallymatrixrandom \sim \multinomialdistribution{\numsamples}{\jointdistribution}}[ &(\estimator\left( \tallymatrixrandom + \begin{bmatrix}
            1 & -1 \\
            0 & 0
        \end{bmatrix} \right) - \estimator(\tallymatrixrandom)) \tallymatrixrandomentry{1}{2}
        \\ + &(\estimator\left( \tallymatrixrandom + \begin{bmatrix}
            0 & 0 \\
            1 & -1
        \end{bmatrix} \right) - \estimator(\tallymatrixrandom)) \tallymatrixrandomentry{2}{2} ]
    \end{align*}

    This completes our proof.
\end{proof}

It is possible to use \Cref{thm:independence-induced-linear-constraints} and \Cref{cor:binary-alphabet-data-processing-inequality-for-estimators} to show a straightforward-but-long proof of \Cref{thm:dmi-is-unique} without relying on Hilbert's nullstellensatz, i.e. \Cref{lma:mutual-informations-are-divisible-by-determinant-squared-of-joint-distribution-matrix}. We provide this appendix because these theorems may be useful for proving uniqueness in future work, where Hilbert's nullstellensatz is less directly applicable.

\section{Convex-minimal DMI estimator convergence rate}\label{app:convergence-rate}

This appendix discusses the convergence rate of the estimator of DMI from \Cref{thm:convex-minimal-dmi-estimator-with-binary-alphabet}.

The convergence rate displays qualitatively different behavior when the joint distribution is any of $\jointdistribution = \begin{pmatrix}
    1/2 & 0 \\
    0 & 1/2
\end{pmatrix}$, $\begin{pmatrix}
    0 & 1/2 \\
    1/2 & 0
\end{pmatrix}$, or $\begin{pmatrix}
    1/4 & 1/4 \\
    1/4 & 1/4
\end{pmatrix}$. Under three of these cases, the variance asympototically decreases inversely proportional to $\numsamples^{2}$, while in other cases it decreases inversely proportional to $\numsamples$.

\begin{theorem}[Variance of optimal DMI estimator on binary alphabet]\label{thm:convex-minimal-dmi-estimator-with-binary-alphabet-variance}
    For binary alphabets $\numactions = \numactionscolumn = 2$, the convex-minimal $\numsamples$-sample unbiased estimator $\estimator$ of DMI, i.e. of $16 \det(\jointdistribution)^{2}$, has variance
    \begin{align*}
        &\phantom{=} \Var(\estimator(\tallymatrixrandom))
        \\ &= \left( \frac{16 (\numsamples - 4)!}{\numsamples!} \right)^{2} [ (4 (\jointdistributionentry{1}{1} \jointdistributionentry{2}{2})^{2} + 4 \jointdistributionentry{1}{1} \jointdistributionentry{2}{2} \jointdistributionentry{2}{1} \jointdistributionentry{1}{2} + 4 (\jointdistributionentry{2}{1} \jointdistributionentry{1}{2})^{2}) \frac{\numsamples!}{(\numsamples - 4)!}
        \\ &+ (8 (\jointdistributionentry{1}{1} + \jointdistributionentry{2}{2}) (\jointdistributionentry{1}{1} \jointdistributionentry{2}{2})^{2} + 8 (\jointdistributionentry{2}{1} + \jointdistributionentry{1}{2}) (\jointdistributionentry{2}{1} \jointdistributionentry{1}{2})^{2} + 4 \jointdistributionentry{1}{1} \jointdistributionentry{2}{2} \jointdistributionentry{2}{1} \jointdistributionentry{1}{2}) \frac{\numsamples!}{(\numsamples - 5)!}
        \\ &+ (12 (\jointdistributionentry{1}{1} \jointdistributionentry{2}{2} - \jointdistributionentry{2}{1} \jointdistributionentry{1}{2})^{2} (\jointdistributionentry{1}{1} \jointdistributionentry{2}{2} + \jointdistributionentry{2}{1} \jointdistributionentry{1}{2}) + 2 ((\jointdistributionentry{1}{1} + \jointdistributionentry{2}{2}) \jointdistributionentry{1}{1} \jointdistributionentry{2}{2} + (\jointdistributionentry{2}{1} + \jointdistributionentry{1}{2}) \jointdistributionentry{2}{1} \jointdistributionentry{1}{2})^{2}) \frac{\numsamples!}{(\numsamples - 6)!}
        \\ &+ (4 (\jointdistributionentry{1}{1} \jointdistributionentry{2}{2} - \jointdistributionentry{2}{1} \jointdistributionentry{1}{2})^{2} ((\jointdistributionentry{1}{1} + \jointdistributionentry{2}{2}) (\jointdistributionentry{1}{1} \jointdistributionentry{2}{2}) + (\jointdistributionentry{2}{1} + \jointdistributionentry{1}{2}) (\jointdistributionentry{2}{1} \jointdistributionentry{1}{2}))) \frac{\numsamples!}{(\numsamples - 7)!}
        \\ &+ (\jointdistributionentry{1}{1} \jointdistributionentry{2}{2} - \jointdistributionentry{2}{1} \jointdistributionentry{1}{2})^{4} \frac{\numsamples!}{(\numsamples - 8)!} ] - 16^{2} (\jointdistributionentry{1}{1} \jointdistributionentry{2}{2} - \jointdistributionentry{2}{1} \jointdistributionentry{1}{2})^{4}
    \end{align*}
    for joint distribution $\jointdistribution \in \Delta([2] \times [2])$. It is the case that
    \begin{align*}
        \Var(\estimator(\tallymatrixrandom)) \numsamples \in O(1)
    \end{align*}
    and under the special cases where
    $\jointdistribution = \begin{pmatrix}
    1/2 & 0 \\
    0 & 1/2
    \end{pmatrix}$, $\begin{pmatrix}
        0 & 1/2 \\
        1/2 & 0
    \end{pmatrix}$, or $\begin{pmatrix}
        1/4 & 1/4 \\
        1/4 & 1/4
    \end{pmatrix}$, we also have
    \begin{align*}
        \Var(\estimator(\tallymatrixrandom)) \numsamples^{2} \in O(1) \text{.}
    \end{align*}
\end{theorem}

\begin{proof}
    This proof will proceed by using a moment-generating function. Consider the multivariate moment-generating function for a tally matrix $\tallymatrixrandom \sim \multinomialdistribution{\numsamples}{\jointdistribution}$, which we calculate as follows:
    \begin{align*}
        &\phantom{=} \expect[\exp(\mgfparam{1}{1} \tallymatrixrandomentry{1}{1} + \mgfparam{1}{2} \tallymatrixrandomentry{1}{1} +  \mgfparam{2}{1} \tallymatrixrandomentry{2}{1} +  \mgfparam{2}{2} \tallymatrixrandomentry{2}{2})]
        \\ &= \sum_{\tallymatrix \in \tallyset{\numsamples}{2}{2}} \exp(\mgfparam{1}{1} \tallymatrixentry{1}{1} + \mgfparam{1}{2} \tallymatrixentry{1}{1} + \mgfparam{2}{1} \tallymatrixentry{2}{1} + \mgfparam{2}{2} \tallymatrixentry{2}{2}) \frac{\numsamples!}{\tallymatrixentry{1}{1}! \cdot \tallymatrixentry{1}{2}! \cdot \tallymatrixentry{2}{1}! \cdot \tallymatrixentry{2}{2}!} \jointdistributionentry{1}{1}^{\tallymatrixentry{1}{1}} \jointdistributionentry{1}{2}^{\tallymatrixentry{1}{2}} \jointdistributionentry{2}{1}^{\tallymatrixentry{2}{1}} \jointdistributionentry{2}{2}^{\tallymatrixentry{2}{2}}
        \\ &= \sum_{\tallymatrix \in \tallyset{\numsamples}{2}{2}} \frac{\numsamples!}{\tallymatrixentry{1}{1}! \cdot \tallymatrixentry{1}{2}! \cdot \tallymatrixentry{2}{1}! \cdot \tallymatrixentry{2}{2}!} (\jointdistributionentry{1}{1} \exp(\mgfparam{1}{1}))^{\tallymatrixentry{1}{1}} (\jointdistributionentry{1}{2} \exp(\mgfparam{1}{2}))^{\tallymatrixentry{1}{2}} (\jointdistributionentry{2}{1} \exp(\mgfparam{2}{1}))^{\tallymatrixentry{2}{1}} (\jointdistributionentry{2}{2} \exp(\mgfparam{2}{2}))^{\tallymatrixentry{2}{2}}
        \\ &= (\jointdistributionentry{1}{1} \exp(\mgfparam{1}{1}) + \jointdistributionentry{1}{2} \exp(\mgfparam{1}{2}) + \jointdistributionentry{2}{1} \exp(\mgfparam{2}{1}) + \jointdistributionentry{2}{2} \exp(\mgfparam{2}{2}))^{\numsamples}
    \end{align*}
    This expression is a function of the four parameters $\mgfparam{1}{1}, \mgfparam{1}{2}, \mgfparam{2}{1}, \mgfparam{2}{2} \in \mathbb{R}$. As shorthand, we will write
    \begin{align*}
        \mgfexp &\triangleq \exp(\mgfparam{1}{1} \tallymatrixrandomentry{1}{1} + \mgfparam{1}{2} \tallymatrixrandomentry{1}{1} + \mgfparam{2}{1} \tallymatrixrandomentry{2}{1} + \mgfparam{2}{2} \tallymatrixrandomentry{2}{2})
        \\ \mgfsum{m} &\triangleq \frac{\numsamples!}{(\numsamples - m)!} (\jointdistributionentry{1}{1} \exp(\mgfparam{1}{1}) + \jointdistributionentry{1}{2} \exp(\mgfparam{1}{2}) + \jointdistributionentry{2}{1} \exp(\mgfparam{2}{1}) + \jointdistributionentry{2}{2} \exp(\mgfparam{2}{2}))^{\numsamples - m}
        \\ \mgfA &\triangleq \jointdistributionentry{1}{1} \exp(\mgfparam{1}{1})
        \\ \mgfB &\triangleq \jointdistributionentry{2}{1} \exp(\mgfparam{2}{1})
        \\ \mgfC &\triangleq \jointdistributionentry{1}{2} \exp(\mgfparam{1}{2})
        \\ \mgfD &\triangleq \jointdistributionentry{2}{2} \exp(\mgfparam{2}{2})
    \end{align*}
    Using this shorthand, we have
    \begin{align*}
        \expect[\mgfexp] = \mgfsum{\numsamples} \text{.}
    \end{align*}
    We will now take repeated derivatives of this function, and repeatedly subtract those derivatives from one another, in order to derive the expectation $\expect[\estimator(\tallymatrixrandom)^{2}]$.
    
    Taking the first derivative with respect to $\mgfparam{1}{1}$, the function $\expect[\mgfexp]$ becomes
    \begin{align*}
        \expect[\tallymatrixrandomentry{1}{1} \mgfexp] = \mgfA \mgfsum{\numsamples - 1} \text{.}
    \end{align*}
    Subtracting $\expect[\mgfexp]$ from both sides of the result, we obtain
    \begin{align*}
        \expect[(\tallymatrixrandomentry{1}{1} - 1) \mgfexp] = -\mgfsum{\numsamples} + \mgfA \mgfsum{\numsamples - 1} \text{.}
    \end{align*}
    Taking the first derivative with respect to $\mgfparam{2}{2}$, this becomes
    \begin{align*}
        \expect[(\tallymatrixrandomentry{1}{1} - 1) \tallymatrixrandomentry{2}{2} \mgfexp] = -\mgfD \mgfsum{\numsamples - 1} + \mgfA \mgfD \mgfsum{\numsamples - 2} \text{.}
    \end{align*}
    Subtracting $\expect[(\tallymatrixrandomentry{1}{1} - 1) \mgfexp]$ from this value, we obtain
    \begin{align*}
        \expect[(\tallymatrixrandomentry{1}{1} - 1) (\tallymatrixrandomentry{2}{2} - 1) \mgfexp] = \mgfsum{\numsamples} - (\mgfA + \mgfD) \mgfsum{\numsamples - 1} + (\mgfA \mgfD) \mgfsum{\numsamples - 2} \text{.}
    \end{align*}
    Next, we compute that
    \begin{align*}
        \expect[\tallymatrixrandomentry{2}{1} \tallymatrixrandomentry{1}{2} \mgfexp] = \mgfB \mgfC \mgfsum{\numsamples - 2}
    \end{align*}
    and subtract this expression from the previous expectation to obtain
    \begin{align*}
        \expect[((\tallymatrixrandomentry{1}{1} - 1) (\tallymatrixrandomentry{2}{2} - 1) - \tallymatrixrandomentry{2}{1} \tallymatrixrandomentry{1}{2}) \mgfexp] = \mgfsum{\numsamples} - (\mgfA + \mgfD) \mgfsum{\numsamples - 1} + (\mgfA \mgfD - \mgfB \mgfC) \mgfsum{\numsamples - 2} \text{.}
    \end{align*}
    Taking the first derivative with respect to $\mgfparam{1}{1}$, and we obtain
    \begin{align*}
        &\phantom{=} \expect[\tallymatrixrandomentry{1}{1} ((\tallymatrixrandomentry{1}{1} - 1) (\tallymatrixrandomentry{2}{2} - 1) - \tallymatrixrandomentry{2}{1} \tallymatrixrandomentry{1}{2}) \mgfexp]
        \\ &= \mgfA \mgfsum{\numsamples - 1} - \mgfA \mgfsum{\numsamples - 1} - \mgfA (\mgfA + \mgfD) \mgfsum{\numsamples - 2} + (\mgfA \mgfD) \mgfsum{\numsamples - 2} + \mgfA (\mgfA \mgfD - \mgfB \mgfC) \mgfsum{\numsamples - 3}
        \\ &= (- \mgfA (\mgfA + \mgfD) + (\mgfA \mgfD)) \mgfsum{\numsamples - 2} + \mgfA (\mgfA \mgfD - \mgfB \mgfC) \mgfsum{\numsamples - 3}
        \\ &= \mgfA [-\mgfA \mgfsum{\numsamples - 2} + (\mgfA \mgfD - \mgfB \mgfC) \mgfsum{\numsamples - 3} ] \text{.}
    \end{align*}
    Taking the first derivative with respect to $\mgfparam{2}{2}$, this then becomes
    \begin{align*}
        &\phantom{=} \expect[\tallymatrixrandomentry{1}{1} \tallymatrixrandomentry{2}{2} ((\tallymatrixrandomentry{1}{1} - 1) (\tallymatrixrandomentry{2}{2} - 1) - \tallymatrixrandomentry{2}{1} \tallymatrixrandomentry{1}{2}) \mgfexp]
        \\ &= \mgfA [-\mgfA \mgfD \mgfsum{\numsamples - 3} + (\mgfA \mgfD) \mgfsum{\numsamples - 3} + D (\mgfA \mgfD - \mgfB \mgfC) \mgfsum{\numsamples - 4} ]
        \\ &= \mgfA \mgfD (\mgfA \mgfD - \mgfB \mgfC) \mgfsum{\numsamples - 4} \text{.}
    \end{align*}
    We can thus deduce that if we had instead performed all derivative operations with respect to $\mgfparam{2}{1}$ and $\mgfparam{1}{2}$ wherever we used $\mgfparam{1}{1}$ and $\mgfparam{2}{2}$, and vice versa, we would derive
    \begin{align*}
        &\phantom{=} \expect[\tallymatrixrandomentry{2}{1} \tallymatrixrandomentry{1}{2} ((\tallymatrixrandomentry{2}{1} - 1) (\tallymatrixrandomentry{1}{2} - 1) - \tallymatrixrandomentry{1}{1} \tallymatrixrandomentry{2}{2}) \mgfexp]
        \\ &= \mgfB \mgfC (\mgfB \mgfC - \mgfA \mgfD) \mgfsum{\numsamples - 4} \text{.}
    \end{align*}
    Adding one to the other, and we obtain
    \begin{align*}
        \expect[\frac{\numsamples!}{16 (\numsamples - 4)!} \estimator(\tallymatrixrandom) \mgfexp] = (\mgfA \mgfD - \mgfB \mgfC)^{2} \mgfsum{\numsamples - 4}
    \end{align*}
    by the definition of $\estimator(\tallymatrixrandom)$.

    Next, we take the derivative of this expression with respect to $\mgfparam{1}{1}$, and obtain
    \begin{align*}
        \expect[\frac{\numsamples!}{16 (\numsamples - 4)!} \estimator(\tallymatrixrandom) \tallymatrixrandomentry{1}{1} \mgfexp] = 2 (\mgfA \mgfD) (\mgfA \mgfD - \mgfB \mgfC) \mgfsum{\numsamples - 4} + \mgfA (\mgfA \mgfD - \mgfB \mgfC)^{2} \mgfsum{\numsamples - 5} \text{.}
    \end{align*}
    Subtracting one from the other, and we obtain
    \begin{align*}
        &\phantom{=} \expect[\frac{\numsamples!}{16 (\numsamples - 4)!} \estimator(\tallymatrixrandom) (\tallymatrixrandomentry{1}{1} - 1) \mgfexp]
        \\ &= [2 (\mgfA \mgfD) (\mgfA \mgfD - \mgfB \mgfC) - (\mgfA \mgfD - \mgfB \mgfC)^{2}] \mgfsum{\numsamples - 4} + \mgfA (\mgfA \mgfD - \mgfB \mgfC)^{2} \mgfsum{\numsamples - 5}
        \\ &= [(\mgfA \mgfD)^{2} - (\mgfB \mgfC)^{2}] \mgfsum{\numsamples - 4} + \mgfA (\mgfA \mgfD - \mgfB \mgfC)^{2} \mgfsum{\numsamples - 5}
        \text{.}
    \end{align*}
    Then, taking the first derivative with respect to $\mgfparam{2}{2}$, the expression becomes
    \begin{align*}
        &\phantom{=} \expect[\frac{\numsamples!}{16 (\numsamples - 4)!} \estimator(\tallymatrixrandom) (\tallymatrixrandomentry{1}{1} - 1) \tallymatrixrandomentry{2}{2} \mgfexp]
        \\ &= 2 (\mgfA \mgfD)^{2} \mgfsum{\numsamples - 4} + \mgfD [(\mgfA \mgfD)^{2} - (\mgfB \mgfC)^{2}] \mgfsum{\numsamples - 5} + 2 \mgfA (\mgfA \mgfD) (\mgfA \mgfD - \mgfB \mgfC) \mgfsum{\numsamples - 5} + (\mgfA \mgfD) (\mgfA \mgfD - \mgfB \mgfC)^{2} \mgfsum{\numsamples - 6}
        \\ &= 2 (\mgfA \mgfD)^{2} \mgfsum{\numsamples - 4} + [\mgfD ((\mgfA \mgfD)^{2} - (\mgfB \mgfC)^{2}) + 2 \mgfA (\mgfA \mgfD) (\mgfA \mgfD - \mgfB \mgfC)] \mgfsum{\numsamples - 5} + (\mgfA \mgfD) (\mgfA \mgfD - \mgfB \mgfC)^{2} \mgfsum{\numsamples - 6} \text{.}
    \end{align*}
    Subtracting the last two expressions from each other, and we get
    \begin{align*}
        &\phantom{=} \expect[\frac{\numsamples!}{16 (\numsamples - 4)!} \estimator(\tallymatrixrandom) (\tallymatrixrandomentry{1}{1} - 1) (\tallymatrixrandomentry{2}{2} - 1) \mgfexp]
        \\ &= [(\mgfA \mgfD)^{2} + (\mgfB \mgfC)^{2}] \mgfsum{\numsamples - 4} + (\mgfA + \mgfD)[(\mgfA \mgfD)^{2} - (\mgfB \mgfC)^{2}] \mgfsum{\numsamples - 5} + (\mgfA \mgfD) (\mgfA \mgfD - \mgfB \mgfC)^{2} \mgfsum{\numsamples - 6} \text{.}
    \end{align*}
    We now set out to calculate $\expect[\frac{\numsamples!}{16 (\numsamples - 4)!} \estimator(\tallymatrixrandom) \tallymatrixrandomentry{2}{1} \tallymatrixrandomentry{1}{2} \mgfexp]$, so that we can subtract it from the previous expression. We compute that
    \begin{align*}
        &\phantom{=} \expect[\frac{\numsamples!}{16 (\numsamples - 4)!} \estimator(\tallymatrixrandom) \tallymatrixrandomentry{2}{1} \mgfexp]
        \\ &= -2 (\mgfB \mgfC) (\mgfA \mgfD - \mgfB \mgfC) \mgfsum{\numsamples - 4} + \mgfB (\mgfA \mgfD - \mgfB \mgfC)^{2} \mgfsum{\numsamples - 5}
    \end{align*}
    and then compute that
    \begin{align*}
        &\phantom{=} \expect[\frac{\numsamples!}{16 (\numsamples - 4)!} \estimator(\tallymatrixrandom) \tallymatrixrandomentry{2}{1} \tallymatrixrandomentry{1}{2} \mgfexp]
        \\ &= [4 (\mgfB \mgfC)^{2} - 2 (\mgfA \mgfD \mgfB \mgfC)] \mgfsum{\numsamples - 4} - 2 (\mgfB + \mgfC) (\mgfB \mgfC) (\mgfA \mgfD - \mgfB \mgfC) \mgfsum{\numsamples - 5} + (\mgfB \mgfC) (\mgfA \mgfD - \mgfB \mgfC)^{2} \mgfsum{\numsamples - 6} \text{.}
    \end{align*}
    Then, subtracting, we have
    \begin{align*}
        &\phantom{=} \expect[\frac{\numsamples!}{16 (\numsamples - 4)!} \estimator(\tallymatrixrandom) ((\tallymatrixrandomentry{1}{1} - 1) (\tallymatrixrandomentry{2}{2} - 1) - \tallymatrixrandomentry{2}{1} \tallymatrixrandomentry{1}{2}) \mgfexp]
        \\ &= [(\mgfA \mgfD + \mgfB \mgfC)^{2} - 4 (\mgfB \mgfC)^{2}] \mgfsum{\numsamples - 4}
        \\ &\phantom{=} + [(\mgfA + \mgfD) ((\mgfA \mgfD)^{2} - (\mgfB \mgfC)^{2}) + 2 (\mgfB + \mgfC) (\mgfB \mgfC) (\mgfA \mgfD - \mgfB \mgfC)] \mgfsum{\numsamples - 5}
        \\ &\phantom{=} + [ (\mgfA \mgfD - \mgfB \mgfC)^{3}] \mgfsum{\numsamples - 6} \text{.}
    \end{align*}
    Then, taking the first derivative with respect to $\mgfparam{1}{1}$, this becomes
    \begin{align*}
        &\phantom{=} \expect[\frac{\numsamples!}{16 (\numsamples - 4)!} \estimator(\tallymatrixrandom) \tallymatrixrandomentry{1}{1} ((\tallymatrixrandomentry{1}{1} - 1) (\tallymatrixrandomentry{2}{2} - 1) - \tallymatrixrandomentry{2}{1} \tallymatrixrandomentry{1}{2}) \mgfexp]
        \\ &= 2 (\mgfA \mgfD) (\mgfA \mgfD + \mgfB \mgfC) \mgfsum{\numsamples - 4}
        \\ &\phantom{=} + \mgfA [(\mgfA \mgfD + \mgfB \mgfC)^{2} - 4 (\mgfB \mgfC)^{2}] \mgfsum{\numsamples - 5}
        \\ &\phantom{=} + [\mgfA ((\mgfA \mgfD)^{2} - (\mgfB \mgfC)^{2}) + 2 (\mgfA + \mgfD) (\mgfA \mgfD)^{2} + 2 (\mgfB + \mgfC) (\mgfA \mgfD) (\mgfB \mgfC)] \mgfsum{\numsamples - 5}
        \\ &\phantom{=} + \mgfA [(\mgfA + \mgfD) ((\mgfA \mgfD)^{2} - (\mgfB \mgfC)^{2}) + 2 (\mgfB + \mgfC) (\mgfB \mgfC) (\mgfA \mgfD - \mgfB \mgfC)] \mgfsum{\numsamples - 6}
        \\ &\phantom{=} + [3 (\mgfA \mgfD) (\mgfA \mgfD - \mgfB \mgfC)^{2}] \mgfsum{\numsamples - 6}
        \\ &\phantom{=} + \mgfA [ (\mgfA \mgfD - \mgfB \mgfC)^{3}] \mgfsum{\numsamples - 7} \text{.}
    \end{align*}
    Taking the first derivative once again with respect to $\mgfparam{2}{2}$, this then becomes
    \begin{align*}
        &\phantom{=} \expect[\frac{\numsamples!}{16 (\numsamples - 4)!} \estimator(\tallymatrixrandom) \tallymatrixrandomentry{1}{1} \tallymatrixrandomentry{2}{2} ((\tallymatrixrandomentry{1}{1} - 1) (\tallymatrixrandomentry{2}{2} - 1) - \tallymatrixrandomentry{2}{1} \tallymatrixrandomentry{1}{2}) \mgfexp]
        \\ &= [4 (\mgfA \mgfD)^{2} + 2 (\mgfA \mgfD) (\mgfB \mgfC)] \mgfsum{\numsamples - 4}
        \\ &\phantom{=} + [8 (\mgfA + \mgfD) (\mgfA \mgfD)^{2} + 2 (\mgfA + \mgfD + \mgfB + \mgfC) \mgfA \mgfD \mgfB \mgfC] \mgfsum{\numsamples - 5}
        \\ &\phantom{=} + (\mgfA \mgfD) [ 12 (\mgfA \mgfD)^{2} - 10 \mgfA \mgfD \mgfB \mgfC - 2 (\mgfB \mgfC)^{2} + 2 (\mgfA + \mgfD)^{2} (\mgfA \mgfD) + 2 (\mgfA + \mgfD) (\mgfA \mgfD) (\mgfB \mgfC)] \mgfsum{\numsamples - 6}
        \\ &\phantom{=} + (\mgfA \mgfD) [ (\mgfA + \mgfD) ((\mgfA \mgfD)^{2} - (\mgfB \mgfC)^{2} + 3 (\mgfA \mgfD - \mgfB \mgfC)^{2}) + 2 (\mgfB + \mgfC) (\mgfB \mgfC) (\mgfA \mgfD - \mgfB \mgfC)] \mgfsum{\numsamples - 7}
        \\ &\phantom{=} + (\mgfA \mgfD) (\mgfA \mgfD - \mgfB \mgfC)^{3} \mgfsum{\numsamples - 8} \text{.}
    \end{align*}
    Since $(\mgfA \mgfD - \mgfB \mgfC)^{2} = (\mgfB \mgfC - \mgfA \mgfD)^{2}$, we can easily derive the expression for $\expect[\frac{\numsamples!}{16 (\numsamples - 4)!} \estimator(\tallymatrixrandom) \tallymatrixrandomentry{2}{1} \tallymatrixrandomentry{1}{2} ((\tallymatrixrandomentry{2}{1} - 1) (\tallymatrixrandomentry{1}{2} - 1) - \tallymatrixrandomentry{1}{1} \tallymatrixrandomentry{2}{2}) \mgfexp]$ and add it to the above, giving us
    \begin{align*}
        &\phantom{=} \expect[(\frac{\numsamples!}{16 (\numsamples - 4)!})^{2} \estimator(\tallymatrixrandom)^{2} \mgfexp]
        \\ &= [ (4 (\mgfA \mgfD)^{2} + 4 \mgfA \mgfD \mgfB \mgfC + 4 (\mgfB \mgfC)^{2}) \mgfsum{\numsamples - 4}
        \\ &+ (8 (\mgfA + \mgfD) (\mgfA \mgfD)^{2} + 8 (\mgfB + \mgfC) (\mgfB \mgfC)^{2} + 4 (\mgfA + \mgfD + \mgfB + \mgfC) \mgfA \mgfD \mgfB \mgfC) \mgfsum{\numsamples - 5}
        \\ &+ (12 (\mgfA \mgfD - \mgfB \mgfC)^{2} (\mgfA \mgfD + \mgfB \mgfC) + 2 ((\mgfA + \mgfD) \mgfA \mgfD + (\mgfB + \mgfC) \mgfB \mgfC)^{2}) \mgfsum{\numsamples - 6}
        \\ &+ (4 (\mgfA \mgfD - \mgfB \mgfC)^{2} ((\mgfA + \mgfD) (\mgfA \mgfD) + (\mgfB + \mgfC) (\mgfB \mgfC))) \mgfsum{\numsamples - 7}
        \\ &+ (\mgfA \mgfD - \mgfB \mgfC)^{4} \mgfsum{\numsamples - 8} ]
    \end{align*}
    Multiplying by $\left( \frac{16 (\numsamples - 4)!}{\numsamples!} \right)^{2}$ gives us
        \begin{align*}
        &\phantom{=} \expect[\estimator(\tallymatrixrandom)^{2} \mgfexp]
        \\ &= \left( \frac{16 (\numsamples - 4)!}{\numsamples!} \right)^{2} [ (4 (\mgfA \mgfD)^{2} + 4 \mgfA \mgfD \mgfB \mgfC + 4 (\mgfB \mgfC)^{2}) \mgfsum{\numsamples - 4}
        \\ &+ (8 (\mgfA + \mgfD) (\mgfA \mgfD)^{2} + 8 (\mgfB + \mgfC) (\mgfB \mgfC)^{2} + 4 (\mgfA + \mgfD + \mgfB + \mgfC) \mgfA \mgfD \mgfB \mgfC) \mgfsum{\numsamples - 5}
        \\ &+ (12 (\mgfA \mgfD - \mgfB \mgfC)^{2} (\mgfA \mgfD + \mgfB \mgfC) + 2 ((\mgfA + \mgfD) \mgfA \mgfD + (\mgfB + \mgfC) \mgfB \mgfC)^{2}) \mgfsum{\numsamples - 6}
        \\ &+ (4 (\mgfA \mgfD - \mgfB \mgfC)^{2} ((\mgfA + \mgfD) (\mgfA \mgfD) + (\mgfB + \mgfC) (\mgfB \mgfC))) \mgfsum{\numsamples - 7}
        \\ &+ (\mgfA \mgfD - \mgfB \mgfC)^{4} \mgfsum{\numsamples - 8} ]
    \end{align*}
    Setting $\mgfparam{1}{1} = \mgfparam{1}{2} = \mgfparam{2}{1} = \mgfparam{2}{2} = 0$ and subtracting $\expect[\estimator(\tallymatrixrandom)]^{2} = (16)^{2} (\jointdistributionentry{1}{1} \jointdistributionentry{2}{2} - \jointdistributionentry{2}{1} \jointdistributionentry{1}{2})^{4}$ gives us our claim for the variance formula.

    Several of the terms in the variance formula vanish as $\numsamples \to \infty$. For example, of the terms in the variance formula
    \begin{align*}
        &\phantom{=} \Var(\estimator(\tallymatrixrandom))
        \\ &= \left( \frac{16 (\numsamples - 4)!}{\numsamples!} \right)^{2} [ (4 (\jointdistributionentry{1}{1} \jointdistributionentry{2}{2})^{2} + 4 \jointdistributionentry{1}{1} \jointdistributionentry{2}{2} \jointdistributionentry{2}{1} \jointdistributionentry{1}{2} + 4 (\jointdistributionentry{2}{1} \jointdistributionentry{1}{2})^{2}) \frac{\numsamples!}{(\numsamples - 4)!}
        \\ &+ (8 (\jointdistributionentry{1}{1} + \jointdistributionentry{2}{2}) (\jointdistributionentry{1}{1} \jointdistributionentry{2}{2})^{2} + 8 (\jointdistributionentry{2}{1} + \jointdistributionentry{1}{2}) (\jointdistributionentry{2}{1} \jointdistributionentry{1}{2})^{2} + 4 \jointdistributionentry{1}{1} \jointdistributionentry{2}{2} \jointdistributionentry{2}{1} \jointdistributionentry{1}{2}) \frac{\numsamples!}{(\numsamples - 5)!}
        \\ &+ (12 (\jointdistributionentry{1}{1} \jointdistributionentry{2}{2} - \jointdistributionentry{2}{1} \jointdistributionentry{1}{2})^{2} (\jointdistributionentry{1}{1} \jointdistributionentry{2}{2} + \jointdistributionentry{2}{1} \jointdistributionentry{1}{2}) + 2 ((\jointdistributionentry{1}{1} + \jointdistributionentry{2}{2}) \jointdistributionentry{1}{1} \jointdistributionentry{2}{2} + (\jointdistributionentry{2}{1} + \jointdistributionentry{1}{2}) \jointdistributionentry{2}{1} \jointdistributionentry{1}{2})^{2}) \frac{\numsamples!}{(\numsamples - 6)!}
        \\ &+ (4 (\jointdistributionentry{1}{1} \jointdistributionentry{2}{2} - \jointdistributionentry{2}{1} \jointdistributionentry{1}{2})^{2} ((\jointdistributionentry{1}{1} + \jointdistributionentry{2}{2}) (\jointdistributionentry{1}{1} \jointdistributionentry{2}{2}) + (\jointdistributionentry{2}{1} + \jointdistributionentry{1}{2}) (\jointdistributionentry{2}{1} \jointdistributionentry{1}{2}))) \frac{\numsamples!}{(\numsamples - 7)!}
        \\ &+ (\jointdistributionentry{1}{1} \jointdistributionentry{2}{2} - \jointdistributionentry{2}{1} \jointdistributionentry{1}{2})^{4} \frac{\numsamples!}{(\numsamples - 8)!} ] - 16^{2} (\jointdistributionentry{1}{1} \jointdistributionentry{2}{2} - \jointdistributionentry{2}{1} \jointdistributionentry{1}{2})^{4} \text{,}
    \end{align*}
    all but
    \begin{align*}
        &\phantom{=} \left( \frac{16 (\numsamples - 4)!}{\numsamples!} \right)^{2} [(4 (\jointdistributionentry{1}{1} \jointdistributionentry{2}{2} - \jointdistributionentry{2}{1} \jointdistributionentry{1}{2})^{2} ((\jointdistributionentry{1}{1} + \jointdistributionentry{2}{2}) (\jointdistributionentry{1}{1} \jointdistributionentry{2}{2}) + (\jointdistributionentry{2}{1} + \jointdistributionentry{1}{2}) (\jointdistributionentry{2}{1} \jointdistributionentry{1}{2}))) \frac{\numsamples!}{(\numsamples - 7)!}
        \\ &+ (\jointdistributionentry{1}{1} \jointdistributionentry{2}{2} - \jointdistributionentry{2}{1} \jointdistributionentry{1}{2})^{4} \frac{\numsamples!}{(\numsamples - 8)!} ] - 16^{2} (\jointdistributionentry{1}{1} \jointdistributionentry{2}{2} - \jointdistributionentry{2}{1} \jointdistributionentry{1}{2})^{4}
    \end{align*}
    decrease faster than $1 / \numsamples$.
    
    On the remaining terms, we can compute that the ratio of the variance to $1 / \numsamples$ approaches
    \begin{align*}
        \lim_{\numsamples \to \infty} \Var(\estimator(\tallymatrixrandom)) \numsamples &= \left( 16 \right)^{2} [(4 (\jointdistributionentry{1}{1} \jointdistributionentry{2}{2} - \jointdistributionentry{2}{1} \jointdistributionentry{1}{2})^{2} ((\jointdistributionentry{1}{1} + \jointdistributionentry{2}{2}) (\jointdistributionentry{1}{1} \jointdistributionentry{2}{2}) + (\jointdistributionentry{2}{1} + \jointdistributionentry{1}{2}) (\jointdistributionentry{2}{1} \jointdistributionentry{1}{2})))]
        \\ &- 16^{3} (\jointdistributionentry{1}{1} \jointdistributionentry{2}{2} - \jointdistributionentry{2}{1} \jointdistributionentry{1}{2})^{4}
    \end{align*}
    and observe that this expression is $0$ when the joint distribution $\jointdistribution$ is either
    \begin{align*}
        \begin{pmatrix}
            1/2 & 0 \\
            0 & 1/2
        \end{pmatrix} \text{, }
        \begin{pmatrix}
            0 & 1/2 \\
            1/2 & 0
        \end{pmatrix} \text{, or }
        \begin{pmatrix}
            1/4 & 1/4 \\
            1/4 & 1/4
        \end{pmatrix} \text{,}
    \end{align*}
    this expression is $0$, meaning that we must have $\Var(\estimator(\tallymatrixrandom)) \numsamples^{2} \in O(1)$. Otherwise, our default guarantee when the expression is non-zero is that is that $\Var(\estimator(\tallymatrixrandom)) \numsamples \in O(1)$.
\end{proof}

\section{Omitted Proofs and Additional Lemmas for \Cref{sec:ex-ante-bounded-sample-mutual-information-estimators}}\label{app:ex-ante-lemmas}

\subsection{Additional Results and Omitted Proofs from \Cref{sec:scoring-rule-based} (Scoring Rule-Based Mutual Informations)}

Let $p$ denote the prior marginal distribution over $\randomcolumn{}$ and $p_{\randomrow{}}$ denote the posterior distribution over $\randomcolumn{}$ conditioned on the signal $\randomrow{}$. 

\begin{lemma}\label{lemma:nonindependent-implies-changed-posterior}
    If $\jointdistribution$ is not generated by independent play, then there exists some row sample $\randomrow{}$ such that $p_{\randomrow{}}\neq p$.
\end{lemma}

\begin{proof}
    We prove the contrapositive.
    Assume that $p = p_{\randomrow{}}$ for all $\randomrow{}$. 
    We will conclude that $\jointdistribution$ must be generated by independent play. 
    Let $q\in \Delta([\numactions])$ denote the marginal distribution of the row player's action $\randomrow{}$.
    The probability of $(\randomrow{}, \randomcolumn{})$ is then 
    \begin{align*}
        \jointdistribution &= q(\randomrow{}) p_{\randomrow{}}(\randomcolumn{})\\
        &= q(\randomrow{}) p(\randomcolumn{}),
    \end{align*}
    and $\jointdistribution = q\transpose{p}$, the definition of independent play. 
\end{proof}

\begin{lemma}\label{lemma:garbling-prior}
    For a column-stochastic matrix $S$ and joint distribution $\jointdistribution$, 
    \begin{align*}
        \sum_{i=1}^\numactions \jointdistributionentry{i}{j} = \sum_{i=1}^\numactions (S\jointdistribution)_{ij} & ~~ \forall j\in[\numactionscolumn],
    \end{align*}
    i.e., garbling by the row player does not change the marginal distribution of the column player. 
\end{lemma}

\begin{proof}
For a given column $j$, 
\begin{equation*}
    \sum_{i=1}^\numactions (S\jointdistribution)_{ij} = \sum_{i=1}^\numactions \sum_{k = 1}^\numactions S_{ik}\jointdistribution_{kj} 
    = \sum_{k = 1}^\numactions \jointdistribution_{kj} \sum_{i = 1}^\numactions S_{ik}
    = \sum_{k = 1}^\numactions \jointdistribution_{kj}, 
\end{equation*}
since $S$ is column-stochastic and $\sum_{i = 1}^\numactions S_{ik} = 1$ for all $k$.
\end{proof}

\begin{lemma}\label{lemma:garbling-posterior}
    For garbling $S$, distribution $\jointdistribution$, $\randomrow{}'\sim S\jointdistribution$, and $\randomrow{}\sim\jointdistribution$, 
    \begin{align*}
        p_{\randomrow{}'} &= \expect_{\randomrow{}}[p_{\randomrow{}} ~|~ \randomrow{}'].
    \end{align*}
\end{lemma}

\begin{proof}
    Let $p_{\randomrow{}'}(\randomcolumn{})$ denote the probability of $\randomcolumn{}$ conditioned on observing $\randomrow{}'\sim S\jointdistribution$ and $p_{\randomrow{}}(\randomcolumn{})$ denote the same conditioned instead on observing $\randomrow{} \sim \jointdistribution$. 
    The posterior $p_{\randomrow{}'}(\randomcolumn{})$ is given by
    \begin{equation*}
        p_{\randomrow{}'}(\randomcolumn{}) = \sum_{\randomrow{}}S_{\randomrow{}'\randomrow{}} p_{\randomrow{}}(\randomcolumn{}) = \expect_{\randomrow{}} [p_{\randomrow{}}(\randomcolumn{}) ~|~ \randomrow{}'].
    \end{equation*}
\end{proof}

\begin{theorem}[\citet{M-56}, \citet{S-71}]\label{thm:convex-scoring}
    For any proper scoring rule $\scoringrule:\Delta(\mathcal{O})\times \mathcal{O}\to \mathbb{R}$, there exists a convex function $\convexfunction:\Delta(\mathcal{O})\to \mathbb{R}$ such that
    \begin{align*}
        \expect_{o\sim p}[\scoringrule(p, o)] &= \convexfunction(p).
    \end{align*}
\end{theorem}

\fixedsamplescoringruletrivial*

\begin{proof}
Let $\scoringrule$ have a fixed-sample estimator.
By the well-known characterization of proper scoring rules \citep{GR-07}, the function $\convexfunction(p) := \expect_{o \sim p}[\scoringrule(p,o)]$ is convex.
We will show this result in two parts. 
First, we will show that $\convexfunction$ is affine, i.e., can be written as $a\cdot p+b$ for vector $a$ and constant $b$.
Second, we will show that affine $\convexfunction$ implies that $\VoI{\jointdistribution}{\scoringrule}$ is identically zero, i.e. trivial.

To show $\convexfunction$ is affine, we consider $\VoI{\jointdistribution_\delta}{\scoringrule}$ on joint distributions $\jointdistribution_\delta$ defined for row vectors $u,v\in\Delta([\numactionscolumn])$ that sum to one and constant $\delta\in[0,1]$ as
\newcommand{\rowvectorrule}{\xleaders\hbox{\rule[0.45ex]{0.1em}{0.4pt}}\hfill}
\newcommand{\rowvectordisplay}[1]{\makebox[8em][c]{\rowvectorrule$\mspace{3mu}#1\mspace{3mu}$\rowvectorrule}}
\begin{align*}
    \jointdistribution_\delta
    =
    \begin{bmatrix}
        \rowvectordisplay{\delta u}\\
        \rowvectordisplay{(1-\delta)v}\\
        \rowvectordisplay{0}\\
        \vdots\\
        \rowvectordisplay{0}
    \end{bmatrix},
\end{align*}
where $0$ denotes zero row vectors and $\delta$ and $1-\delta$ are the row marginals.
Expanding out the definition of $\VoI{\jointdistribution_\delta}{\scoringrule}$:
\begin{align*}
    \VoI{\jointdistribution_\delta}{\scoringrule}
    =
    \delta\convexfunction(u)
    +
    (1-\delta)\convexfunction(v)
    -
    \convexfunction(v+\delta(u-v)).
\end{align*}

We can take the derivative of $\VoI{\jointdistribution_\delta}{\scoringrule}$ with respect to $\delta$ evaluated at $\delta=0$.
Fix vector $v$ as a constant and write this derivative, evaluated at $\delta=0$, as a function of $u$ as $D_v(u)$.
From the scoring rule definition of $\VoI{\jointdistribution_\delta}{\scoringrule}$, we have:
\begin{align}
\label{eq:Du}
    D_v(u)
    =
    \convexfunction(u)
    -
    \convexfunction(v)
    -
    \nabla\convexfunction(v)\cdot(u-v).
\end{align}

We now argue that $D_v(u)$ is an affine function of $u$.
Because the number of samples is fixed, $\VoI{\jointdistribution}{\scoringrule}$ is a polynomial.
Recall that when we take the derivative of $\VoI{\jointdistribution_{\delta}}{\scoringrule}$ with respect to $\delta$ and then set $\delta$ to zero, the only terms that survive are those with $\delta$ to the first power: The degree $0$ terms are gone
from the derivative, and the degree $2$ or higher terms are gone from
setting $\delta=0$.
Write the vector
$u=(u_1,\ldots,u_{\numactionscolumn})$ and consider the first row of $\jointdistribution_{\delta}$, which
is $\delta u = (\delta u_1,\ldots,\delta u_{\numactionscolumn}).$ Only
monomials with exactly one of these $\delta u_j$ terms can survive in
$D_v(u)$. Thus $D_v(u)$ is affine in $u$.

Because the right-hand side of \eqref{eq:Du} is $\convexfunction(u)$ plus an affine function of $u$, it must be that $\convexfunction(u)$ is affine as well.

We complete the proof by showing that if $\convexfunction(p) = a\cdot p + b$, i.e., affine, then $\VoI{\jointdistribution}{\scoringrule}=0$ identically for all $\jointdistribution$:
\begin{align*}
    \VoI{\jointdistribution}{\scoringrule} &= \E_{(\randomrow{}, \randomcolumn{})\sim F}[\scoringrule(p_{\randomrow{}}, \randomcolumn{})] - \E_{(\randomrow{}, \randomcolumn{})\sim F}[\scoringrule(p, \randomcolumn{})]\\
    &= \E_{\randomrow{}\sim F}[\E_{\randomcolumn{}\sim p_{\randomrow{}}}[\scoringrule(p_{\randomrow{}}, \randomcolumn{})]] - \convexfunction(p)\\
    &= \E_{\randomrow{}\sim F}[\convexfunction(p_{\randomrow{}})] - \convexfunction(p)\\
    &= \E_{\randomrow{}\sim F}[a\cdot p_{\randomrow{}} + b] - a\cdot p - b\\
    &= a \cdot \E_{\randomrow{}\sim F}[p_{\randomrow{}}] - a\cdot p = 0.\qedhere
\end{align*}
\end{proof}

\scoringMIstrict*

\begin{proof}
    Consider any joint distribution $\jointdistribution$ and strict row garbling $S$. Let $(\randomrow{},\randomcolumn{})\sim\jointdistribution$, let $\randomrow{}'$ be obtained by applying $S$ to $\randomrow{}$, and let $p_{\randomrow{}}$ and $p_{\randomrow{}'}$ denote the posterior distributions of $\randomcolumn{}$ conditional on $\randomrow{}$ and $\randomrow{}'$, respectively. Because $\randomrow{}'$ is a garbling of $\randomrow{}$,
    \begin{align*}
        p_{\randomrow{}'}=\expect[p_{\randomrow{}}\mid \randomrow{}'].
    \end{align*}
    Thus, the posterior before garbling is a mean-preserving spread of the posterior after garbling. Because $S$ is strict for $\jointdistribution$, this spread is strict. A strictly proper scoring rule $\scoringrule$ induces a strictly convex function $\convexfunction$ (\Cref{thm:convex-scoring}), so strict Jensen's inequality gives
    \begin{align*}
        \expect[\convexfunction(p_{\randomrow{}'})]
        < \expect[\convexfunction(p_{\randomrow{}})].
    \end{align*}
    The garbling leaves the marginal distribution $p$ of $\randomcolumn{}$ unchanged. Therefore,
    \begin{align*}
        \VoI{S\jointdistribution}{\scoringrule}
        = \expect[\convexfunction(p_{\randomrow{}'})]-\convexfunction(p)
        < \expect[\convexfunction(p_{\randomrow{}})]-\convexfunction(p)
        = \VoI{\jointdistribution}{\scoringrule}.
    \end{align*}
    Hence, $\VoI{\cdot}{\scoringrule}$ is a strict mutual information.
\end{proof}

\subsection{Additional Results and Omitted Proofs from \Cref{sec:collision-estimators} (Collision Estimators)}

\begin{lemma}\label{lemma:quad-score-expected-value}
    With beliefs $p$, the expected value of the quadratic score is
    \begin{align*}
        \expect_{O\sim p}[\quadscore(p, O)] &= ||p||_2^2. 
    \end{align*}
\end{lemma}
\begin{proof}
    We expand the expectation as a sum over expected outcomes:
    \begin{align*}
        \expect_{O\sim p}[\quadscore(p, O)] &= \expect_{O\sim p}[2p(O) - \sum_i p(i)^2]\\
        &= \expect_{O\sim p}[2p(O)] - \sum_i p(i)^2\\
        &= \sum_ip(O)(2p(O))- \sum_i p(i)^2\\
        &= 2\sum_ip(O)^2 - \sum_i p(i)^2\\
        &= \sum_i p(i)^2 = ||p||_2^2. 
    \end{align*}
\end{proof}

\begin{lemma}\label{lemma:collision-estimator-expected-samples}
    The expected number of samples for the collision estimator is $n^* + 1$. 
\end{lemma}
\begin{proof}
    Without loss of generality, assume the samples with positive probability are $1,\dots, n^*$. 
    Let $\stoppingtime$ be a random variable for the number of samples drawn.
    If $\Pr[\randomrow{} = i] = 1$ for any $i$, then $\stoppingtime = 2$ deterministically. 
    We assume that $\Pr[\randomrow{} = i] < 1$ for all $i$ for the remainder of this analysis, and let $\randomrow{}$ denote a generic row sample from $\jointdistribution$. 
    The expectation of $\stoppingtime$ is
    \begin{align*}
        \E[\stoppingtime] &= \sum_{i=1}^{n^*}\Pr[\randomrow{1} = i]\sum_{t = 2}^\infty t \Pr[\randomrow{t} = i]\Pr[\randomrow{} \neq i]^{(t-2)}\\
        &= \sum_{i=1}^{n^*}\Pr[\randomrow{} = i]^2\sum_{t = 2}^\infty t \Pr[\randomrow{} \neq i]^{(t-2)}\\
        &= \sum_{i=1}^{n^*}\frac{\Pr[\randomrow{} = i]^2}{\Pr[\randomrow{} \neq i]^2}\sum_{t = 2}^\infty t \Pr[\randomrow{} \neq i]^{t}\\
        &= \sum_{i=1}^{n^*}\left(\frac{\Pr[\randomrow{} = i]^2}{\Pr[\randomrow{} \neq i]^2}\sum_{t = 1}^\infty t \Pr[\randomrow{} \neq i]^{t} - \frac{\Pr[\randomrow{} = i]^2}{\Pr[\randomrow{} \neq i]}\right).
    \end{align*}
    Applying the infinite series summation for arithmetico-geometric sequences gives
    \begin{align*}
        \E[T] &= \sum_{i=1}^{n^*}\left(\frac{\Pr[\randomrow{} = i]^2}{\Pr[\randomrow{} \neq i]^2}\frac{\Pr[\randomrow{} \neq i]}{\Pr[\randomrow{} = i]^2} - \frac{\Pr[\randomrow{} = i]^2}{\Pr[\randomrow{} \neq i]}\right)\\
        &= \sum_{i=1}^{n^*}\frac{1 - \Pr[\randomrow{} = i]^2}{1 - \Pr[\randomrow{} = i]}\\
        &= \sum_{i=1}^{n^*}\frac{(1 - \Pr[\randomrow{} = i])(1 + \Pr[\randomrow{} = i])}{1 - \Pr[\randomrow{} = i]}\\
        &= \sum_{i=1}^{n^*}(1 + \Pr[\randomrow{} = i])\\
        &= \numactions^* + \sum_{i}\Pr[\randomrow{} = i] = \numactions^* + 1. 
    \end{align*}
\end{proof}

\collisionestimator*

\begin{proof}
    Let $p$ be the marginal distribution of $\randomcolumn{}$, and let $p_i$ be the distribution of $\randomcolumn{}$ conditional on $\randomrow{}=i$. The quadratic-score mutual information (\Cref{lemma:quad-score-expected-value}) is,
    \begin{equation}\label{eq:quadratic-score-MI-collision}
        \E_{\randomrow{}}\left[\norm{p_{\randomrow{}}}^2\right]-\norm{p}^2.
    \end{equation}

    We construct unbiased estimators for these two terms.

    First, draw two independent samples $\randomcolumn{1},\randomcolumn{2}$ from $p$ and output $1$ if they are equal. The expectation of this process is,
    \[
        \textstyle \Pr[\randomcolumn{1}=\randomcolumn{2}]
        =\sum_j \Pr[\randomcolumn{1}=j]\Pr[\randomcolumn{2}=j]
        =\sum_j p(j)^2
        =\norm{p}^2.
    \]
    Thus, $\indicate{\randomcolumn{1}=\randomcolumn{2}}$ is an unbiased estimator of $\norm{p}^2$.

    Next, we estimate $\E_{\randomrow{}}[\norm{p_{\randomrow{}}}^2]$. Draw $(\randomrow{1},\randomcolumn{1})$, and then continue drawing samples until the first time $\stoppingtime>1$ at which $\randomrow{\stoppingtime}=\randomrow{1}$. Conditional on $\randomrow{1}=i$, the variables $\randomcolumn{1}$ and $\randomcolumn{\stoppingtime}$ are independent draws from $p_i$. Hence,
\[\textstyle \Pr[\randomcolumn{1}=\randomcolumn{\stoppingtime}\mid \randomrow{1}=i]
        =\sum_j p_i(j)^2
        =\norm{p_i}^2.
    \]
    Averaging over the random value of $\randomrow{1}$ gives
\[
\textstyle \E\left[\indicate{\randomcolumn{1}=\randomcolumn{\stoppingtime}}\right]
        =\sum_i \Pr[\randomrow{1}=i]\norm{p_i}^2
        =\E_{\randomrow{}}[\norm{p_{\randomrow{}}}^2].
    \]
    Thus, $\indicate{\randomcolumn{1}=\randomcolumn{\stoppingtime}}$ is an unbiased estimator of $\E_{\randomrow{}}[\norm{p_{\randomrow{}}}^2]$.   The collision estimator outputs the second estimator minus the first and is therefore an unbiased estimator of Equation \eqref{eq:quadratic-score-MI-collision}.
\end{proof}

\bettercollision*

\begin{proof}
    First, draw two samples, $(I_1,J_1)$ and $(I_2,J_2)$.
    The expected quadratic score of the prior $p$ is $\sum_{j \in [m]} p(j)^2$, which is the probability that $J_1=J_2$, which is estimated by the indicator.
    In other words, the expected quadratic score for predicting the column player's signal based on the prior is equal to the probability that two of the column player's signals match.
   
    It only remains to unbiasedly estimate the expected score for
    predicting the column player \emph{after} seeing $I$.

    Consider the probability that these two samples are equal in both
    coordinates: 
    \begin{align*}
        \expect \left[\indicate{I_1 = I_2 \text{ and } J_1 =
        J_2}\right] &= \sum_{i \in [n]} \prob \left[I_1 = i
      \right]\prob \left[I_2 = i\right] \prob\left[J_1 = J_2 | I_1 =
          I_2 = i\right]\\
          &= \sum_{i \in [n]} \prob \left[I_1 = i
      \right]\prob \left[I_2 = i\right] \sum_{j\in [\numactionscolumn]}\prob\left[J_1 = j | I_1 = i\right]\prob\left[J_2 = j | I_2 = i\right]\\
        &= \sum_{i \in [n]} \prob \left[I_1 = i
      \right]\prob \left[I_2 = i\right] \sum_{j\in [\numactionscolumn]}\prob\left[J_1 = j | I_1 = i\right]^2\\
    \end{align*}
    This is off of the desired $\expect_{I} \left[ \sum_{j \in [m]} p_I(j)^2 \right] = \sum_{i \in [n]} \Pr[I=i] \sum_{j \in [m]} \Pr[J=j \mid I=i]^2$ by an extra term of $\prob[I_2 = i]$ inside the sum.
    This missing term can be addressed by estimating the number of draws $K$ it takes to again draw $I = i$.  \Cref{def:fast-collision-estimator} gives this estimator.

    The expected number of samples is the expectation over $i \in
    [n]$ of the probability of match on $i$ times the expected
    number of samples on match which is $1/\prob[I = i]$.  The
    probability of an exact match on $i$ is at most the probability
    that $I_1 = I_2 = i$ which is $\prob[i = I]^2$. Conditioned on
    $I_1 = i$ the probability we have a match is at most
    $\prob{I_2 = i}$ and the number f samples we draw on a match
    is $1/prob[i = I]$. Thus, the expected number of samples is at most 1.
\end{proof}

\subsection{Omitted Proof from \Cref{sec:arbitrarily-few-expected-sample-estimators} (Arbitrarily Few-Expected Sample Estimators)}

\scalingincreasesvariance*

\begin{proof}
    Let $R'$ be the random payment of the scaled mechanism. 
    $R'$ is distributed as
    \begin{align*}
        R' &= \begin{cases}
            \frac{1}{\alpha}R, & \text{w.p. }\alpha\\
            0, & \text{w.p. }1-\alpha
        \end{cases}.
    \end{align*}
    We compute the variance of $R'$:
    \begin{align*}
        \Var(R') &= \E[(R')^2] - \E[R']^2\\
        &= 0 + \alpha \E\left[\left(\frac{1}{\alpha}R\right)^2\right] - \alpha^2\E\left[\left(\frac{1}{\alpha}R\right)\right]^2\\
        &= \frac{1}{\alpha} \E\left[R^2\right] - \E\left[R\right]^2\\
        &= \frac{1-\alpha}{\alpha} \E\left[R^2\right] + \E\left[R^2\right] - \E\left[R\right]^2\\
        &= \frac{1-\alpha}{\alpha} \E\left[R^2\right] + \Var(R).
    \end{align*}
\end{proof}

\section{Log Score Mutual Information and Entropy}\label{app:entropy}

We recover the classic definition of mutual information derived from information-theoretic entropy as introduced by \citet{S-48} from the log scoring rule. 
The log scoring rule is given by $\logscore(p, o) = \ln p(o)$ where $p(o)$ is the probability $p$ places on outcome $o$. 
The Shannon entropy of the column player's sample is $H(\randomcolumn{}) = -\sum_j p(j)\ln p(j)$, and the joint entropy is $H(\randomcolumn{}, \randomrow{}) = -\sum_{i,j} \jointdistributionentry{i}{j} \ln \jointdistributionentry{i}{j}$.
The conditional entropy, $H(\randomcolumn{}~|~\randomrow{}) = - \sum_{i,j}F_{ij}(\ln(F_{ij}) - \ln(p(i)))$. 

\begin{definition}[Classical Mutual Information]\label{def:classic-MI}
    The classic definition of mutual information is 
    \begin{align*}
        \MI(\jointdistribution) &= H(\randomcolumn{}) - H(\randomcolumn{} ~|~ \randomrow{}).
    \end{align*}
\end{definition}

\begin{proposition}\label{prop:VoI-classical-MI}
    The value of information with regard to the log score $\logscore$ of a joint distribution $\jointdistribution$ is
    \begin{align*}
        \VoI{F}{\logscore} &= \MI(\jointdistribution).
    \end{align*}
\end{proposition}

\begin{proof}
The definition of $\MI$ gives 
\begin{align*}
    \MI(\jointdistribution) &= H(\randomcolumn{}) - H(\randomcolumn{} ~|~ \randomrow{}).
\end{align*}
We observe that $H(\randomcolumn{}) =- \expect[\logscore(p,\randomcolumn{})]$. 
We then rewrite the conditional entropy $H(\randomcolumn{} ~|~ \randomrow{})$ in the language of scoring rules:
\begin{align*}
    H(\randomcolumn{} ~|~ \randomrow{}) &= H(\randomcolumn{}, \randomrow{}) - H(\randomrow{})\\
    &= -\sum_{i,j} \jointdistributionentry{i}{j} \ln \jointdistributionentry{i}{j} + \sum_i \Pr[\randomrow{} = i]\ln \Pr[\randomrow{} = i]\\
    &= \sum_{i}\Pr[\randomrow{} = i]\ln \Pr[\randomrow{} = i] - \sum_j \jointdistributionentry{i}{j} \ln \jointdistributionentry{i}{j}\\
    &= \sum_{i}\Pr[\randomrow{} = i]\ln \Pr[\randomrow{} = i] - \Pr[\randomrow{} = i] \sum_j \Pr[\randomcolumn{} = j ~|~ \randomrow{} = i] \ln \Pr[\randomrow{} = i] \Pr[\randomcolumn{} = j ~|~ \randomrow{} = i]\\
    &= \sum_{i}\Pr[\randomrow{} = i]\ln \Pr[\randomrow{} = i] - \Pr[\randomrow{} = i] \sum_j \Pr[\randomcolumn{} = j ~|~ \randomrow{} = i] \left(\ln \Pr[\randomrow{} = i] + \ln\Pr[\randomcolumn{} = j ~|~ \randomrow{} = i]\right)\\
    &= \sum_{i}\Pr[\randomrow{} = i]\ln \Pr[\randomrow{} = i] - \Pr[\randomrow{} = i] \ln \Pr[\randomrow{} = i] - \Pr[\randomrow{} = i] \sum_j\ln\Pr[\randomcolumn{} = j ~|~ \randomrow{} = i]\\
    &= - \sum_{i} \Pr[\randomrow{} = i] \sum_j\ln\Pr[\randomcolumn{} = j ~|~ \randomrow{} = i]\\
    &= -\expect_{(\randomrow{}, \randomcolumn{})\sim \jointdistribution}[\logscore(p_{\randomrow{}}, \randomcolumn{})].
\end{align*}
We then obtain
\begin{align*}
    \MI(\jointdistribution) &= -\expect[\logscore(p,\randomcolumn{})] + \expect_{(\randomrow{}, \randomcolumn{})\sim \jointdistribution}[\logscore(p_{\randomrow{}}, \randomcolumn{})]\\
    &= \expect_{(\randomrow{}, \randomcolumn{})\sim \jointdistribution}[\logscore(p_{\randomrow{}}, \randomcolumn{})] - \expect[\logscore(p,\randomcolumn{})]\\
    &= \VoI{\jointdistribution}{\logscore}. 
\end{align*}
\end{proof}
Furthermore, as $\MI(\jointdistribution)$ is symmetric (i.e., $\MI(\jointdistribution) = \MI(\transpose{\jointdistribution})$), it must also satisfy the column data processing inequality of \Cref{def:mutual-information}, making classic Shannon mutual information a full mutual information under our definitions. 
This justifies our use of the term as a generalization of classical mutual information. 

Unfortunately, there is no unbiased estimator for entropy \citep{P-03}, leading us to consider row mutual informations derived from other scoring rules.

\end{document}